\documentclass[moor,nonblindrev]{informs1}              

\usepackage{graphics} 
\usepackage{epsfig} 
\let\proof\relax

\let\theoremstyle\relax
\usepackage{amsthm}  
\usepackage{amscd}
\usepackage{float}
\usepackage{color}
\usepackage[utf8]{inputenc}
\usepackage[english]{babel}
\usepackage{algorithm}
\usepackage{multicol}

\usepackage{algpseudocode}
\newtheorem{theorem}{Theorem}
\newtheorem{lemma}{Lemma}

\newtheorem{proposition}{Proposition}
\newtheorem{corollary}{Corollary}
\newtheorem{definition}{Definition}
\newtheorem{assumption}{Assumption}

\usepackage{empheq}
\usepackage{xcolor}
\usepackage{epsfig}
\usepackage{amsfonts}
\usepackage{mathrsfs}
\usepackage{amssymb}
\usepackage{amstext}
\usepackage{amsmath}
\usepackage{xspace}
\usepackage{theorem}
\usepackage{comment}
\usepackage{multirow}
\usepackage{enumitem}
\usepackage{multirow,subfigure}

\usepackage{url}
\usepackage[sort]{natbib}

\usepackage{my_definitions}

\algdef{SE}{Iterate}{EndIterate}[1]{\textbf{Iterate} \(\mbox{#1}\) \textbf{do}}{\textbf{end}}

\usepackage{natbib}
 \NatBibNumeric
 \def\bibfont{\small}%
 \bibpunct[, ]{[}{]}{,}{n}{}{,}%

\usepackage[colorlinks=true,breaklinks=true,bookmarks=true,urlcolor=blue,
     citecolor=blue,linkcolor=blue,bookmarksopen=false,draft=false]{hyperref}

\def\EMAIL#1{\href{mailto:#1}{#1}}

\EquationsNumberedThrough    

\begin{document}


\RUNAUTHOR{Gao et al.}

\RUNTITLE{Finite-Sample Analysis of Decentralized Q-Learning for Stochastic Games}

\TITLE{Finite-Sample Analysis of Decentralized Q-Learning for Stochastic Games}

\ARTICLEAUTHORS{%
\AUTHOR{Zuguang Gao}
\AFF{The University of Chicago Booth School of Business, Chicago, Illinois 60637, \EMAIL{zuguang.gao@chicagobooth.edu}}
\AUTHOR{Qianqian Ma}
\AFF{Department of Electrical and Computer Engineering, Boston University, Boston, Massachusetts 02215, \EMAIL{maqq@bu.edu}}
\AUTHOR{Tamer Ba\c{s}ar}
\AFF{Coordinated Science Laboratory, University of Illinois at Urbana-Champaign, Urbana, Illinois 61801, \EMAIL{basar1@illinois.edu}}
\AUTHOR{John R. Birge}
\AFF{The University of Chicago Booth School of Business, Chicago, Illinois 60637, \EMAIL{john.birge@chicagobooth.edu}}
} 

\ABSTRACT{%
{Learning in stochastic games is arguably the most standard and fundamental setting in multi-agent reinforcement learning (MARL). In this paper, we consider decentralized MARL in stochastic games in the non-asymptotic regime.}  {In particular, we establish the finite-sample complexity of fully decentralized Q-learning algorithms in a significant class of general-sum stochastic games (SGs) -- weakly acyclic SGs, which includes the common cooperative MARL setting with an identical reward to all agents (a Markov team problem) as a special case. We focus on the practical while challenging setting of \emph{fully decentralized} MARL, where neither the rewards nor the actions of other agents can be observed by each agent.  In fact, each agent is completely oblivious to the presence of other decision makers. Both the tabular and the linear function approximation cases have been considered. In the tabular setting, we analyze the sample complexity for the decentralized Q-learning algorithm to converge to a Markov perfect equilibrium (Nash equilibrium).} With linear function approximation, the results are for convergence to a linear approximated equilibrium -- a new notion of equilibrium that we propose -- which describes that each agent's policy is a best reply (to other agents) within a linear space. Numerical experiments are also provided for both settings to demonstrate the results.
}%


\MSCCLASS{91A06, 91A10, 91A15, 91A68}
\ORMSCLASS{Primary: Games/group decisions: stochastic; secondary: Analysis of algorithms
}
\HISTORY{}

\maketitle
%
\section{Introduction.}\label{intro}

Multi-agent learning has received extensive attention in recent years, and has achieved considerable success in application areas including traffic control, network routing, energy distribution, robotic systems, and social economic problems, where multiple agents learn concurrently how to solve a task by interacting with the same environment~(\citet{stone2000multiagent}). 
The canonical model for dynamic multi-agent interactions is \emph{stochastic games}, also known as Markov games~(\citet{littman1994markov}), which were  first introduced by~\citet{shapley1953stochastic}; we refer the interested reader to~\citet{busoniu2008comprehensive} and~\citet{zhang2021multi} for comprehensive surveys. Compared to repeated games, i.e., repeated play of static games~(\citet{wu2021eliciting}), stochastic games are more general in the sense that each stage game is affected by the previous joint actions of all agents, who may or may not be in a network~(\citet{correa2021network,kempe2020inducing,qu2020scalable,lin2021multi}) through the system state evolution, and thus are applicable to a broader set of problem settings~(\citet{stern2020dynamic,birge2021interfere}). 

There are different attributes that may be associated with a stochastic game. Depending on the number of agents and their reward functions, a stochastic game can be a two-agent zero-sum game and/or a multi-agent general-sum game, where the former games have two agents, and the reward of one agent is always the negation of the reward of the other agent, representing a fully competitive relationship, and the latter games may have $N$~agents for any~$N\ge 2$, with no restrictions on their reward functions~(\citet{basar1999dynamic}). Depending on the length of the game, a stochastic game can have either finite horizon or infinite horizon, where the objective of each agent is to choose a policy (that maps from a state to an action) to maximize her total reward over the length of the horizon (if finite), or to maximize her total discounted reward or time-averaged reward over an infinite horizon.

In all of these different types of stochastic games, the most commonly studied notion on the agents' joint policy is the so-called \emph{Markov perfect equilibrium}~(\citep{maskin1988theory}), or simply \emph{Nash equilibrium}~(\citet{nash1950equilibrium}) in some literature~(\citet{das2015pure}). Roughly speaking, in a Markov perfect equilibrium, each agent's policy is a \emph{best reply} (i.e., maximizes her own total (discounted) reward) to all other agents’ joint policy. Mainstreams of research include analyzing the hardness to compute the equilibria~(e.g., \citet{daskalakis2009complexity,daskalakis2013complexity,garg2018r}), approximating and analyzing the equilibria~(e.g., \citet{branzei2021nash,adsul2021fast,boodaghians2021polynomial}), designing algorithms to find the equilibria with the knowledge of the transitions and rewards~(e.g., \citet{hu2003nash,hansen2013strategy}) or without such knowledge~(e.g. \citet{arslan2016decentralized}).

With the recent boom of reinforcement learning (RL), there is a growing interest in applying  the methods of RL to stochastic games; see~\citet{shoham2007if} and the references therein. The adaptive decision-making framework of RL, together with the context of multiple interacting learners, lead to multi-agent RL (MARL). MARL corresponds to the learning problem in a multi-agent system in which multiple RL agents learn simultaneously by interacting with the stochastic environment via a trial-and-error approach, from which they receive rewards for their actions. 
These MARL algorithms can be either centralized (meaning that there is a central controller who has full access to the game set up as well as each agent's actions and rewards, and who provides coordination among these agents) or decentralized (meaning that each agent makes decisions based on local information without a coordinator). Development convergent decentralized MARL algorithms, however, is well known to be challenging. In contrast to the single-agent scenario, such as bandit problems~(\citet{frazier2014incentivizing,emamjomeh2021adversarial}), the state evolution and the rewards earned by each agent depend on not only the current state and this agent's action, but also the actions taken by other agents. As a consequence, the existing single-agent learning algorithms cannot be directly extended to multi-agent settings, due to the fact that the environment is now non-stationary from each agent’s perspective, resulting in potential non-convergence, as shown in~\citet{tan1993multi,claus1998dynamics}. Such nonstationarity issue is the key challenge to address in order to develop a converging multi-agent learning algorithm. 

To address the nonstationarity issue, \citet{arslan2016decentralized} proposed a decentralized Q-learning algorithm and proved the asymptotic convergence of the algorithm for a subclass of stochastic games -- weakly acyclic games~(\citet{young2004strategic}). Roughly speaking, a weakly acyclic game is a stochastic game such that best-reply dynamics cannot enter inescapable oscillations (see Definition~\ref{def:best} later for a precise definition). The standard Q-learning, a widely used model-free, value-based RL algorithm, has been applied to specific multiple-agent systems~(\citet{tan1993multi, sen1994learning}), but no analytical results exist regarding the convergent properties of Q-learning in a stochastic game setting. In the decentralized Q-learning algorithm by~\citet{arslan2016decentralized}, agents do not update their policies for an extended period of time, which is called an \emph{exploration phase}, during which the environment becomes stationary from each agent's perspective. During the exploration phases, the Q-functions of each agent are still being updated. The policy of each agent is then updated at the end of each exploration phase, according to the current values of the Q-functions. This ``explore-then-update" procedure is repeated, and~\citet{arslan2016decentralized} have shown that the joint policy asymptotically converges to a Markov perfect equilibrium as the length of each exploration phase and the number of exploration phases go to infinity.

\subsection{Contributions.}

In this paper, we use the algorithm from~\citet{arslan2016decentralized} as a starting point, and build upon it with several new developments. We summarize our main contributions as follows.

\begin{itemize}
    \item We study the non-asymptotic convergence guarantee, namely, sample complexity, of the decentralized Q-learning algorithm in~\citet{arslan2016decentralized}~(Theorem~\ref{thm:1} of this paper). We note that many of the generalizations from asymptotic convergence to finite-sample analysis turn out to be nontrivial. A brief overview of this is provided in Section~\ref{sec:con}. To the best of our knowledge, this is the first sample complexity result for convergence to Markov perfect equilibrium on multi-player general-sum stochastic games with {infinite horizon}. Moreover, the sample complexity result is expressed explicitly in the game parameters~(Corollary~\ref{cor:1}), which is made possible by developing new bounds~(Proposition~\ref{prop:bounds}) on some implicit parameters~(e.g., the minimum probability of stationary distribution~$\mu_{\min}$ and the mixing time~$t_{\rm mix}$).
    \item We apply linear function approximation to approximate the Q-functions in this general-sum stochastic game. Instead of maintaining a Q-function on all state-action pairs, each agent now restricts her attention to a linear space with smaller dimensions, compared to the large state/action space. Under the restriction of a smaller-dimensional linear space, the original Q-functions will not be fully recovered. As a result, the original Markov perfect equilibria might not be reachable. To this end, we define a new notion of equilibrium -- \emph{linear approximated equilibrium}~(Definiton~\ref{def:laequi}). Roughly speaking, in a linear approximated equilibrium, each agent's policy is a best reply, according to the linearly approximated Q-functions (instead of the original ones) to all other agents' joint policy. The algorithm with linear function approximation~(Algorithm~\ref{al:2}) is shown in Section~\ref{sec:linear}. We again prove the sample complexity result for the algorithm, i.e., we provide finite lower bounds on the length of exploration phases and the number of exploration phases needed for the joint policy of all agents to converge to a linear approximated equilibrium with high probability. This work also appears to be the first to apply function approximation to general-sum stochastic games.
    \item We provide numerical experiments of both algorithms~(with and without linear function approximation) on the classical Grid World game~(\citet{sutton2018reinforcement}) with minor modifications. Specifically, we have two agents, one of whom moves in the vertical direction and the other one moves in the horizontal direction. At each time step, the state moves in the direction determined by the joint actions of both agents, and both agents receive a negative reward of~$-1$, except when the current state is at a ``terminal" state, where both agents receive a reward of~$0$ and the state does not change for all joint actions. This cooperative game belongs to a common subclass of stochastic games with identical reward functions for all agents, namely, the Markov team problems~(\citet{ho1980team,yuksel2013stochastic}). The experimental results confirm that, as the length and the number of exploration phases increase, the joint policy adopted by both agents converges to a Markov perfect equilibrium (in the tabular setting) or a linear approximated equilibrium (in the function approximation setting) with higher probability.
\end{itemize}

\subsection{Related Work.}
This paper is related to several sets of previous work. We mention below some of the most relevant ones.

\subsubsection{Stochastic Games (SGs).}
Stochastic games were proposed by~\citet{shapley1953stochastic} and can be viewed as a generalization of the Markov Decision Process~(MDP) to the multi-agent setting. Since then, this framework has become a classical model for multi-agent learning, and there is a long line of work in finding the Markov perfect equilibrium (Nash equilibrium) of different types of stochastic games under various assumptions. The works of \citet{littman1994markov,littman2001friend,hu2003nash,hansen2013strategy,wei2020linear} assumed full knowledge of the transition kernel and the reward functions, while \citet{wei2017online,jia2019feature,sidford2020solving,zhang2020model,wei2021last} assumed certain reachability conditions, e.g., access to some simulators that allow each agent to directly sample transitions and rewards for each state-action pair. Most of these works were aimed at designing algorithms that asymptotically converge to Markov perfect equilibria.

Another set of recent works focused on the non-asymptotic sample complexity/regret guarantees for learning in SGs. For two-player zero-sum games, \citet{bai2020provable,xie2020learning} developed the first provably-efficient
learning algorithms based on optimistic value iteration. \citet{liu2021sharp} improved upon these works with a model-based ``Optimistic Nash Value Iteration" algorithm
and achieves the best-known sample complexity for finding an $\epsilon$-Nash equilibrium. For multi-player general-sum games, \citet{liu2021sharp} provided the first sample complexity
guarantees for finding Markov perfect (Nash) equilibria, correlated equilibria~(CE), or coarse correlated equilibria~(CCE), where CE and CCE are some other notions of equilibria which can be viewed as relaxations of Markov perfect equilibria. 

It is worth noting that the algorithms that appeared in all of the aforementioned works are centralized algorithms. In contrast to those, \citet{daskalakis2020independent} established the sample complexity of independent policy gradient methods in zero-sum SGs. \citet{tian2020provably} proved sublinear regret in finite-horizon SGs, under the name of \emph{online agnostic learning}. \citet{sayin2021decentralized} provided a decentralized Q-learning algorithm for two-player zero-sum stochastic games and showed its asymptotic convergence. \citet{bai2020near} proposed a (decentralized) V-learning algorithm and proved sample complexity results for two-player zero-sum games. Concurrent works of~\citet{jin2021v,song2021can,mao2021provably} developed finite-sample convergence results~(to CE and CCE) of the V-learning algorithm for multi-player general-sum stochastic games with finite horizons (in episodic setting), while in this work, we study the finite-sample convergence to Markov perfect equilibrium for multi-player general-sum stochastic games with infinite horizon, along with the use of function approximation.


\subsubsection{Best Reply Process.} Best reply (BR) processes, also known as best response dynamics~(\citet{hopkins1999note}) or best response schemes~(\citet{lei2018distributed}), describe a policy update scheme for the agents, where each agent selects the policy that maximizes her payoff given other agents' joint policy~(\citet{fudenberg1998theory,basar1999dynamic}). Work on BR processes mostly falls into one of two directions: The first one studies whether (under certain assumptions or no assumption) the BR processes converge  to a Nash Equilibrium~(NE) if NE exists~(e.g., \citet{harks2012existence,milchtaich1996congestion}). The second one considers how fast it takes for the BR processes to converge to a NE~(e.g., \citet{even2005fast,fabrikant2004complexity,ieong2005fast,syrgkanis2010complexity}). It is well known that BR processes do not necessarily always converge to a NE, even if one exists. However, for the class of \emph{weakly acyclic games}~(\citet{fabrikant2010structure,apt2015classification}), which includes  all \emph{potential games} as special cases, BR processes are guaranteed to converge to one of the equilibria of the game~(\citet{monderer1996potential,rosenthal1973class}). 

Recently, there has been a growing interest in developing different variants of the BR process that may be applied to different classes of games, e.g., the proximal BR processes~(\citet{facchinei201012,pang2017two}), the Gauss–Seidel BR processes~(\citet{facchinei2011decomposition}), and the BR processes with a deviator~(\citet{feldman2017efficiency}). As for stochastic games, \citet{lei2020synchronous} proposed several generalizations of the proximal BR processes and showed their convergence. The asynchronous BR processes and their connections to block-coordinate descent~(BCD) schemes were further investigated in~\citet{lei2020asynchronous}. \citet{swenson2018best} and~\citet{leslie2020best} studied the convergence of \emph{continuous-time} BR processes in potential games and zero-sum stochastic games, respectively. 

In the standard (discrete-time) BR process, at each time when the joint policy is not an equilibrium, \emph{one} arbitrary agent is chosen to improve her best policy given the policies of others, i.e., most of the aforementioned works consider the \emph{asynchronous} BR process. In the weakly acyclic game that we study in this paper, all agents may update their policies synchronously. Moreover, it is sometimes unrealistic for each agent to compute the best reply policy given all other agents' joint policies. To accommodate the synchronous policy updates and to alleviate the computation of the best reply policies, \citet{young2004strategic} proposed the BR process with \emph{inertia}, which lets each agent  keep her current policy if it is a best reply. If her current policy is not a best reply, she still keeps the current policy with a certain probability~(inertia), but updates to another random policy in other cases. In~\citet{arslan2016decentralized} as well as in this work, an agent cannot even assess \emph{whether} her current policy is a best reply, which is the reason why we adopt Q-learning to mimic the BR process with inertia.

\subsubsection{Finite-sample Analysis for Q-learning.} Q-learning~(\citet{watkins1992q}) has been  recognized as one of the workhorses of RL. Beyond asymptotic convergence analysis, a  considerable amount of prior work has studied its finite-sample performance~(\citet{even2003learning,beck2012error,wainwright2019stochastic,qu2020finite,li2020sample}), with the sharpest results so far by~\citet{li2020sample}. To address the setting with massive state-action spaces, Q-learning has also been blended with linear function approximation~(\citet{melo2008analysis}), with its finite-sample analysis being investigated in \citet{chen2019finite}. It is not clear yet whether similar results  can be established for decentralized Q-learning in infinite horizon multi-agent general-sum SGs, which is the focus of our work.

\subsection{Organization.}
The rest of this paper is organized as follows. In Section~\ref{sec:pre}, we formally introduce stochastic games, weakly acyclic games, and Best Reply Process with Inertia~(\citet{young2004strategic}). Some useful properties of the Best Reply Process with Inertia are also developed. Then, in Section~\ref{sec:tabu}, we introduce the algorithm from~\citet{arslan2016decentralized}, called Algorithm~\ref{al:1}, and provide a finite-sample analysis, through which some bounds on the convergent measures are also established. We then move to Section~\ref{sec:linear}, and introduce the algorithm with linear function approximation~(Algorithm~\ref{al:2}), define the linear approximated equilibrium, and derive sample complexity results on  convergence of Algorithm~\ref{al:2} (either to a linear approximated equilibrium or to a Markov perfect equilibrium, under different assumptions). Full numerical studies on the classical Grid World game are provided at the end of Section~\ref{sec:tabu} as well as in Section~\ref{sec:linear}. Finally, we conclude the paper in Section~\ref{sec:con}, where we also highlight some of the technical difficulties in establishing our results. Proofs of two of the main results are provided in Appendices~\ref{sec:appa} and~\ref{sec:appb}.

\section{Preliminaries.}\label{sec:pre}
	
%
%

	\subsection{Stochastic games.}
	
	A (finite) discounted stochastic game has the following ingredients~(\citet{fink1964equilibrium}).
	\begin{itemize}
		\item A finite set of agents, with the $i$-th agent referred to as agent~$i$ for $i\in\{1,\dots,N\}=:[N]$;
		\item a finite set $\mathcal{S}$ of states;
		\item a finite set $\mathcal{A}^i$ of actions for each agent $i$;
		\item a nonnegative deterministic reward function $r^i$ for each agent~$i$, which determines agent~$i$'s reward, i.e., $r^i\left(s,a^1,\dots,a^N\right)\in[0,r^i_{\max}]$ at each state $s\in\mathcal{S}$ and for each joint action $\left(a^1,\dots,a^N\right)
		\in\mathcal{A}:=\mathcal{A}^1\times\cdots\times\mathcal{A}^N$;
		\item a discount factor $\gamma^i\in(0,1)$ for each agent~$i$;
		\item a random initial state $s_0\in\mathcal{S}$;
		\item a transition kernel  for the probability 
		$P[ s^{\prime} |s,a^1,\dots,a^N]$ of each state transition from $s\in\mathcal{S}$ to $s^{\prime}\in\mathcal{S}$ for each joint $N$-tuple of actions $(a^1,\dots,a^N)\in\mathcal{A}^1\times\cdots\times\mathcal{A}^N.$
	\end{itemize}
	The dynamic evolution of the game can be described as follows. At each time $t\ge 0$, each agent~$i$ observes the state $s_t$, and chooses an action $a_t^i\in \mathcal{A}^i$. The agent then receives a reward $r^i(s_t,a_t)\in[0,r^i_{\max}]$ where $a_t:=\left(a_t^1,\ldots,a_t^N\right)$, i.e., the reward of each agent, is determined by the state as well as the joint action selected by all agents. The system then transits to the next state $s_{t+1}$ according to the transition kernel $P[\cdot |s_t,a^1,\dots,a^N]$. We note that the {information structure} we consider here is \emph{fully decentralized},  in the sense that each agent, when choosing its action, has access to only the current (and past) states, as well as its own history of actions and rewards, while the rewards and the state transitions are determined by the joint actions of all agents. Each agent does not have access to other agents' actions and rewards. Although the reward $r^i$ that an agent receives depends on the state and the joint actions of all agents, the agents do \emph{not} have full knowledge of their own reward functions, but only observe the reward they receive. In fact, an agent can be completely oblivious to the presence of other agents. 
	
	A policy for an agent is a rule of choosing an action at any time, based on the agent's history of observations. While an agent may use any function of the available information as its policy, without loss of optimality, we focus on stationary (i.e., time-invariant) policies where an agent's action at time $t$ is  based solely on the current state $s_t$, i.e., a stationary policy of agent~$i$, denoted by $\pi^i$, is a mapping from the state space $\mathcal{S}$ to $\mathcal{P}\left(\mathcal{A}^i\right)$, the set of probability distributions on $\mathcal{A}^i$. The set of such stationary policies of agent~$i$ is denoted by $\Delta^i:=\left\{\pi^i: \mathcal{S}\to \mathcal{P}\left(\mathcal{A}^i\right)\right\}$. The set of \emph{deterministic} stationary policies of agent~$i$ is denoted by $\Pi^i:=\left\{\pi^i: \mathcal{S}\to \mathcal{A}^i\right\}\subset \Delta^i$. We let $\Delta:=\times_{i=1}^N\Delta^i$, $\Delta^{-i}:= \times_{j\ne i}\Delta^j$, and $\Pi:= \times_{i=1}^N\Pi^i$, $\Pi^{-i}:= \times_{j\ne i}\Pi^j$. Further, we denote the joint actions and joint policies by $a:=\left(a^1,\ldots,a^N\right)$ and $\pi:=\left(\pi^1,\ldots,\pi^N\right)$, respectively. The joint actions and policies of all agents except agent $i$ are denoted by $a^{-i}:=\left(a^1,\ldots,a^{i-1},a^{i+1},\ldots,a^N\right)$ and $\pi^{-i}:=\left(\pi^1,\ldots,\pi^{i-1},\pi^{i+1},\ldots,\pi^N\right)$, respectively. The joint actions and policies may also be written as $a = \left(a^i,a^{-i}\right)$ and $\pi = \left(\pi^i,\pi^{-i}\right)$, respectively.
	
	The objective of each agent~$i$ is to find a policy $\pi^i\in\Delta^i$ that maximizes the total expected discounted reward, or the \emph{value function}:
	\begin{align}
		V_{\pi}^i(s) = \E\left[\sum_{t=0}^\infty\left(\gamma^i\right)^tr^i\left(s_t,a_t\right) \ \bigg|\ s_0 = s\right],\quad \forall s\in\mathcal{S},
	\end{align}
	where the expectation is taken over the joint distribution of~$a$ given by $\pi(s)$, as well as the random state~$s$ given by the transition kernel at each step. The \emph{$Q$-function} (or \emph{action-value function}) of agent $i$, $Q_\pi^i:\mathcal{S}\times\mathcal{A}^i\to \mathbb{R}$ of a joint policy~$\pi$ is defined by
	\begin{align}
		Q_{\pi}^i(s,a^i) =\E\left[\sum_{t=0}^\infty\left(\gamma^i\right)^tr^i\left(s_t,a^i_t, a^{-i}_t\right) \ \bigg|\ s_0 = s, a_0^i = a^i\right],\quad \forall s\in\mathcal{S},
	\end{align}
	where the actions of~$i$ are taken according to the policy $\pi^i$ except the initial action $a^i_0=a^i$, and the joint actions $a^{-i}$ are taken according to the joint policy $\pi^{-i}$. In addition, for any $\pi^{-i}\in\Delta^{-i}$, we define the \emph{Bellman operator} of agent~$i$ as a self-mapping of $\mathcal{S}\times \mathcal{A}^i$:
	\begin{align}\label{eq:bellman}
		\mathcal{T}^i_{\pi^{-i}}(Q^i)(s,a^i) := \E_{\pi^{-i}(s)}\left[r^i\left(s,a^i,a^{-i}\right) + \gamma^i\sum_{s^{\prime}\in\mathcal{S}} P\left[s^{\prime}\mid s,a^i,a^{-i}\right]\max_{\hat{a}^i\in\mathcal{A}^i} Q^i(s^{\prime},\hat{a}^i)\right],\quad\forall (s,a^i),
	\end{align}
	where the expectation is taken over the joint distribution of $a^{-i}$ given by $\pi^{-i}(s)$. 
	
	We next define the \emph{Markov perfect equilibrium} of a stochastic game.
	\begin{definition}\label{def:equi}
		A joint policy $\pi^*=\left(\pi^{*1},\ldots,\pi^{*N}\right)\in\Delta$ is a (Markov perfect) equilibrium if 
		$$
		V_{\left(\pi^{*i},\pi^{*-i}\right)}^i(s) = \max_{\pi^i\in\Delta^i}V_{\left(\pi^{i},\pi^{*-i}\right)}^i(s),\quad \forall s\in\mathcal{S},\ i\in\left\{1,\ldots,N\right\}.
		$$
	\end{definition}
	We denote by $\Pi_{\rm eq}$ the set of all equilibrium joint policies. It is well known that any finite discounted stochastic game admits at least one equilibrium joint policy~(\citet{fudenberg1991game}). However, a deterministic equilibrium joint policy may not exist in general. In the following, we introduce a set of games, termed \emph{weakly acyclic games}, for which a deterministic joint equilibrium policy always exists.

	\subsection{Weakly acyclic games.}
	
	\begin{definition}
		A policy $\pi^{*i}\in \Delta^i$ is called a best reply to $\pi^{-i}\in\Delta^{-i}$ (for agent $i$) if 
		\begin{align*}
			V^i_{\left(\pi^{*i},\pi^{-i}\right)}(s) = \max_{\pi^i\in\Delta^i}V^i_{\left(\pi^{i},\pi^{-i}\right)}(s),\quad\forall s\in\mathcal{S}.
		\end{align*}
		A best reply $\pi^{*i}\in \Delta^i$ to $\pi^{-i}\in\Delta^{-i}$ is called a strict best reply to $\left(\pi^i,\pi^{-i}\right)$ if 
		\begin{align*}
			V^i_{\left(\pi^{*i},\pi^{-i}\right)}(s) > V^i_{\left(\pi^{i},\pi^{-i}\right)}(s),\quad\text{for some } s\in\mathcal{S}.
		\end{align*}
	\end{definition}
	
	We denote by $\Pi^i_{\pi^{-i}}$ the agent~$i$'s set of deterministic best replies to any $\pi^{-i}\in\Delta^{-i}$, i.e.,
	\begin{align*}
		\Pi_{\pi^{-i}}^i  := \left\{{\pi}^{*i}\in\Pi^i  : \ V^i_{\left(\pi^{*i},\pi^{-i}\right)}(s) = \max_{\pi^i\in\Delta^i}V^i_{\left(\pi^{i},\pi^{-i}\right)}(s),\quad\forall s\in\mathcal{S} \right\}.
	\end{align*}
	
	Agent $i$'s best replies to any $\pi^{-i}\in\Delta^{-i}$ can be characterized by the optimal $Q$-functions $Q_{(\pi^{*i},\pi^{-i})}^i$. With some abuse of notation, we simply write $Q_{\pi^{-i}}^i$ in place of $Q_{(\pi^{*i},\pi^{-i})}^i$. This optimal $Q$-function satisfies the fixed point equation of the Bellman operator:
	\begin{align}
		Q_{\pi^{-i}}^i(s,a^i) = \E_{\pi^{-i}(s)} \left[r^i\left(s,a^i,a^{-i}\right)   +\gamma^i \sum_{s^{\prime}\in\mathcal{S}} P\left[s^{\prime}\mid s,a^i,a^{-i}\right]\max_{\hat{a}^i\in\mathcal{A}^i} Q_{\pi^{-i}}^i(s^{\prime},\hat{a}^i) \right],\ \forall \left(s,a^i\right).
		\label{eq:Qfp}
	\end{align}
	The optimal $Q$-function $Q_{\pi^{-i}}^i(s,a^i)$ is agent~$i$'s expected discounted value-to-go from the initial state $s$, assuming that $i$ initially chooses action $a^i$ and uses an optimal policy thereafter while all other agents use the joint policy $\pi^{-i}$.
	Then, we can rewrite agent~$i$'s set of deterministic best replies to $\pi^{-i}$ as
	\begin{align}\label{eq:pi-i}
		\Pi_{\pi^{-i}}^i = \left\{{\pi}^{*i}\in\Pi^i: \ Q_{\pi^{-i}}^i\left(s,{\pi}^{*i}(s)\right)=\max_{a^i\in\mathcal{A}^i} Q_{\pi^{-i}}^i(s,a^i),   \quad \forall s\in\mathcal{S} \right\}.
	\end{align}
	From~\eqref{eq:pi-i} and Definition~\ref{def:equi}, we have that a deterministic joint policy $\pi^*\in\Pi_{\rm eq}$ if $\pi^{*i}\in\Pi^i_{(\pi^*)^{-i}}$ for all $i\in[N]$.
	
	We next define the \emph{best reply graph} on the set of deterministic joint policies $\Pi$. Specifically, each node (vertex) in the graph is a deterministic joint policy $\pi\in\Pi$, and there is a directed edge from $\pi_k$ to $\pi_{l}$ if for some $i\in[N]$, $\pi_{l}^i\ne \pi_{k}^i$, $\pi_{l}^j=\pi_{k}^j,\forall j\ne i$, and $\pi_{l}^i\in\Pi_{\pi_k^{-i}}^i$. When there is a directed edge from $\pi_k$ to $\pi_l$, we also say that $\pi_l$ is an \emph{out-neighbor} of $\pi_k$ in the strict best reply graph. We then define the \emph{strict best reply path} and the \emph{weakly acyclic game}.
	
	\begin{definition}\label{def:path}
		A sequence of deterministic joint policies $\pi_0,\pi_1,\ldots$ is called a strict best reply path if for each $k$, $\pi_k$ and $\pi_{k+1}$ differ for exactly one agent, say agent $i$, and $\pi_{k+1}^i$ is a strict best reply with respect to $\pi_k$.
	\end{definition}
	
	\begin{definition}\label{def:best}
		A discounted stochastic game is called \textit{weakly acyclic} under strict best replies if there is a strict best  reply path starting from each deterministic joint policy and ending at a deterministic equilibrium policy.
	\end{definition}

	\begin{figure}[ht]
		\centering
		\includegraphics[scale=0.5]{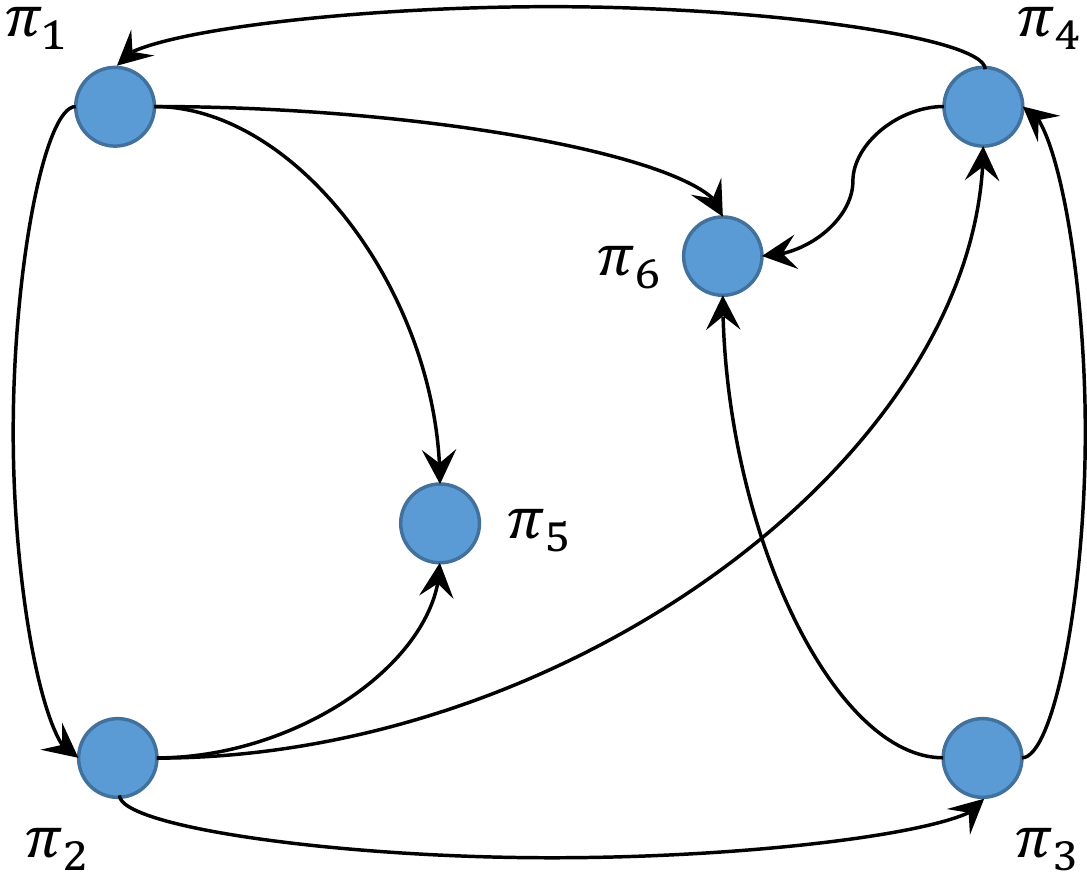}
		\caption[]{The strict best reply graph of a stochastic game.} \label{fig3}
	\end{figure}
	
	Figure~\ref{fig3} shows the strict best reply graph of a weakly acyclic game.
	A node with no outgoing edges is an equilibrium policy ($\pi_5$ and $\pi_6$ in this graph). The game illustrated in Figure~\ref{fig3} is weakly acyclic under strict best replies since there is a path from every node to an equilibrium. Also note that a weakly acyclic game may have cycles in its strict best reply graph, for example, $\pi_1\rightarrow\pi_2\rightarrow\pi_{3}\rightarrow\pi_4\rightarrow\pi_1$ in Figure~\ref{fig3}.

	
	If the game has no cycles in its strict best reply graph, we may consider the process of letting only one agent switch to one of its best replies at each step, and such a process will continue until no agent has strict best replies, at which time the joint policy of all agents is a deterministic equilibrium joint policy. However, as described above, the strict best reply graph of a weakly acyclic game may contain cycles. We next introduce the Best Reply Process with Inertia~(\citet{young2004strategic}) as Algorithm~\ref{al:0}, which assigns to each agent a strict positive probability of choosing each of its strict best replies. 
	
	\begin{algorithm}[ht]
		\addtocounter{algorithm}{-1}
		\caption{Best Reply Process with Inertia (for agent~$i$)}
		\label{al:0}
		\algsetblock[Name]{Parameters}{}{0}{}
		\algsetblock[Name]{Initialize}{}{0}{}
		\algsetblock[Name]{Define}{}{0}{}
		\begin{algorithmic}[1]
			\Statex {Set parameters}
			\Statex \hspace*{5mm} $\lambda^i\in(0,1)$: inertia
			\vspace{2mm}
			\State Initialize  $\pi_0^i \in \Pi^i$ (arbitrary)
			\Iterate {$k\ge 0$}
			\If {$\pi_k^i\in \Pi_{\pi_k^{-i}}^i$}
			\State $\pi_{k+1}^i = \pi_k^i$
			\Else
			\State $\pi_{k+1}^i = \left\{ \begin{array}{cl}   \pi_{k}^i & \mbox{ w.p. } \lambda^i \\ \mbox{any} \ \pi^i\in\Pi_{\pi_k^{-i}}^i & \mbox{ w.p. } (1-\lambda^i)/|\Pi_{\pi_k^{-i}}^i|\end{array}\right. $
			\EndIf
			\EndIterate
		\end{algorithmic}
	\end{algorithm}
	Under the Best Reply Process with Inertia (BRPI), if the joint policy $\pi_k$ is an equilibrium policy at step $k$, the policy will never change in the following steps; otherwise, the joint policy at step $k+1$, denoted by $\pi_{k+1}$, can be any joint policy that is an out-neighbor of $\pi_k$ in the strict best reply graph with strict positive probability. For each $\pi\in\Pi$, there exists a strict best reply path of minimum length $L_\pi$, and letting $L:=\max_{\pi\in\Pi}L_\pi$ be the maximum length of the shortest strict best reply path from any policy to an equilibrium policy. Then, starting from an arbitrary joint policy $\pi_0$ and letting all agents update their policies following the BRPI, the joint policy $\pi_{L}$ in $L$ steps later will be an equilibrium policy with some positive probability $p_{\min}$, i.e., $p_{\min}:= \min_{\pi}P\left[\pi_L\in\Pi_{\rm eq}\mid \pi_0=\pi\right]$. We then have the following lemma, which provides a lower bound on $p_{\min}$. 
	\begin{lemma}\label{lm:pmin}
	Let all agents update their policies by the BRPI. We have that
	\begin{align}\label{eq:ppp}
		p_{\min}
		\ge \left(\min_{j\in \{1,\ldots,N\}}\left\{ \frac{1-\lambda^j}{\left|\Pi^{j}\right|}\cdot \prod_{i\ne j}\lambda^i\right\}\right)^L =: \hat{p}.
	\end{align}
\end{lemma}

\begin{proof}[Proof of Lemma~\ref{lm:pmin}]\ 
	Let $\pi_0 = \hat{\pi}_0\in \Pi$ be an arbitrary initial joint policy. If $\hat{\pi}_0\in\Pi_{\rm eq}$, then it holds that $P\left[\pi_L=\hat{\pi}_0\in\Pi_{\rm eq}\mid \pi_0 = \hat{\pi}_0\in\Pi_{\rm eq}\right]=1$; otherwise, let~$l$ be the length of the shortest strict best reply path from $\hat{\pi}_0$ to an equilibrium policy, where $l\le L$. Let the sequence of policies along the path be $\hat{\pi}_0,\hat{\pi}_1,\ldots, \hat{\pi}_l$, with $\hat{\pi}_l\in\Pi_{\rm eq}$. Further, let $i_1,\ldots,i_l$ be the agent that changes its policy at each update, i.e., $\hat{\pi}_{n-1}$ and $\hat{\pi}_{n}$ differ only at agent $i_n$, for all $n=1,\ldots,l$. 
	Then, by the policy updates in the BRPI, we have that
	\begin{align*}
		&\quad P\left[\pi_L\in\Pi_{\rm eq}\mid \pi_0=\hat{\pi}_0\right]\\
		&\ge P\left[\pi_1=\hat{\pi}_1,\ldots,\pi_l = \hat{\pi}_l, \pi_{l+1} = \hat{\pi}_l,\ldots,\pi_L = \hat{\pi}_l\mid \pi_0=\hat{\pi}_0\right]\\
		&= P\left[\pi_1 = \hat{\pi}_1\mid \pi_0 = \hat{\pi}_0\right] P\left[\pi_2 = \hat{\pi}_2\mid \pi_0 = \hat{\pi}_0, \pi_1 = \hat{\pi}_1\right] \cdots P\left[\pi_l = \hat{\pi}_l\mid \pi_0 = \hat{\pi}_0, \ldots \pi_{l-1} = \hat{\pi}_{l-1}\right]\\
		&\qquad \cdot P\left[\pi_{l+1} =\cdots= \pi_L = \hat{\pi}_l\mid \pi_0 = \hat{\pi}_0, \ldots \pi_{l} = \hat{\pi}_{l}\right]\\
		&\ge \left(\frac{1-\lambda^{i_1}}{\left|\Pi^{i_1}\right|}\cdot \prod_{i\ne i_1}\lambda^i\right)
		\left(\frac{1-\lambda^{i_2}}{\left|\Pi^{i_2}\right|}\cdot \prod_{i\ne i_2}\lambda^i\right)\cdots
		\left(\frac{1-\lambda^{i_l}}{\left|\Pi^{i_l}\right|}\cdot \prod_{i\ne i_l}\lambda^i\right)\cdot 1 = \prod_{j\in\{i_1,\ldots,i_l\}}\left(\frac{1-\lambda^j}{\left|\Pi^{j}\right|}\cdot \prod_{i\ne j}\lambda^i\right)\\
		&\ge \left(\min_{j\in \{1,\ldots,N\}} \left\{\frac{1-\lambda^j}{\left|\Pi^{j}\right|}\cdot \prod_{i\ne j}\lambda^i\right\}\right)^l \ge \left(\min_{j\in \{1,\ldots,N\}}\left\{ \frac{1-\lambda^j}{\left|\Pi^{j}\right|}\cdot \prod_{i\ne j}\lambda^i\right\}\right)^L.
	\end{align*}
	Note that since the above holds for any arbitrary initial joint policy $\hat{\pi}_0$, we conclude that it is a lower bound for $p_{\min}$.
\end{proof}

	This implies that the BRPI will reach an equilibrium policy in a finite number of steps w.p.~$1$. We further have the following result.
	
	\begin{proposition}\label{prop:brp}
		Let all agents update their deterministic policies according to the BRPI. We have that
		\begin{align*}
			P \left[\pi_k\in\Pi_{\rm eq}\right]\ge 1-\delta,
		\end{align*}
		provided that
		\begin{align*}
			k\ge \frac{L\cdot \log\delta}{\log\left(1-\left(\min_{j\in \{1,\ldots,N\}}\left\{ \frac{1-\lambda^j}{\left|\Pi^{j}\right|}\cdot \prod_{i\ne j}\lambda^i\right\}\right)^L\right)} + L.
		\end{align*}
	\end{proposition}
	\begin{proof}[Proof of Proposition~\ref{prop:brp}]\
	For any initial joint policy $\pi_0 = \hat{\pi}_0$, we have that
	$
		P\left[\pi_{L}\in\Pi_{\rm eq}\mid \pi_0 = \hat{\pi}_0\right]\ge p_{\min}.
	$
	We first show that
	\begin{align}\label{eq:Ppi}
		P\left[\pi_{nL}\in\Pi_{\rm eq}\right]\ge p_{\min}\left[1 + (1-p_{\min}) + (1-p_{\min})^2 +\cdots +(1-p_{\min})^{n-1}\right] = 1-(1-p_{\min})^n.
	\end{align}
	We show~\eqref{eq:Ppi} by induction on $n$. The base case holds since $P\left[\pi_L\in\Pi_{\rm eq}\right]\ge p_{\min}$. As for the induction step, assuming that~\eqref{eq:Ppi} holds for $\pi_{(n-1)L}$, we have that
	\begin{align*}
		&\quad P\left[\pi_{nL}\in\Pi_{\rm eq}\right]\\
		&= P\left[\pi_{nL}\in\Pi_{\rm eq}\mid \pi_{(n-1)L}\in\Pi_{\rm eq}\right]P\left[\pi_{(n-1)L}\in\Pi_{\rm eq}\right] + P\left[\pi_{nL}\in\Pi_{\rm eq}\mid \pi_{(n-1)L}\notin\Pi_{\rm eq}\right]P\left[\pi_{(n-1)L}\notin\Pi_{\rm eq}\right]\\
		&\ge P\left[\pi_{(n-1)L}\in\Pi_{\rm eq}\right] + p_{\min}P\left[\pi_{(n-1)L}\notin\Pi_{\rm eq}\right]
		= P\left[\pi_{(n-1)L}\in\Pi_{\rm eq}\right] + p_{\min}\left(1-P\left[\pi_{(n-1)L}\in\Pi_{\rm eq}\right]\right)\\
		&= (1-p_{\min})P\left[\pi_{(n-1)L}\in\Pi_{\rm eq}\right] + p_{\min}\ge (1-p_{\min})\left(1-(1-p_{\min})^{n-1}\right)+p_{\min} = 1-(1-p_{\min})^n,
	\end{align*}
	which completes the induction step. Therefore, if~$k$ satisfies
	\begin{align}\label{eq:k}
		k\ge \frac{L\log \delta}{\log(1-p_{\min})}+L,
	\end{align}
	then, we have that
	\begin{align*}
		P\left[\pi_k\in\Pi_{\rm eq}\right] &\ge P\left[\pi_{\left\lfloor\frac{k}{L}\right\rfloor L}\in\Pi_{\rm eq}\right] \ge 1-(1-p_{\min})^{\left\lfloor \frac{k}{L}\right\rfloor}\ge 1-(1-p_{\min})^{\left\lfloor \frac{\log\delta}{\log(1-p_{\min})}+1\right\rfloor}\\
		&\ge 1-(1-p_{\min})^{\frac{\log\delta}{\log(1-p_{\min})}} = 1-(1-p_{\min})^{\log_{1-p_{\min}}\delta} = 1-\delta.
	\end{align*}
	The proof is completed by taking the lower bound of $p_{\min}$ from Lemma~\ref{lm:pmin} to~\eqref{eq:k}.
\end{proof}
	
	We note that in applying the BRPI, each agent~$i$ needs to construct $\Pi_{\pi_k^{-i}}^i$ at step $k$, which can be done according to~\eqref{eq:pi-i} by first computing $Q_{\pi_k^{-i}}^i$ by solving the fixed point equation~\eqref{eq:Qfp}. However, since we assume that agents do not have access to the state transition probabilities $P$, neither do they know the joint policy $\pi_k^{-i}$ of other agents, they would not be able to compute~\eqref{eq:Qfp} directly. In the next section, we introduce and analyze the sample complexity of the Q-learning algorithm for stochastic games, where agents would be able to approximate their best replies and adjust their policies accordingly.

	\section{Decentralized Q-learning in tabular setting.}\label{sec:tabu}
	
	Recall that in the decentralized setting,  at any time~$t$, each agent has access to the history of state realizations up to time~$t$, its own set of actions~$\mathcal{A}^i$ and discount factor $\gamma^i$, as well as its own history of actions. Since each agent is completely oblivious to the existence of other agents, agent~$i$ may view the decision making problem as a stationary Markov decision process, and use the standard Q-learning algorithm:
	\begin{subequations}\label{eq:agentQ}
		\begin{align}
			Q_{t+1}^i(s_t,a_t^i)  &=   (1-\eta_t^{i})Q_t^i(s_t,a_t^i) + \eta_t^{i} \left[ r^i(s_t,a_t^i,a_t^{-i})  +  \gamma^i \max_{a^i\in\mathcal{A}^i} Q_t^i(s_{t+1},a^i) \right],\\
			Q_{t+1}^i(s,a^i) &= Q_t(s,a^i),\quad \forall (s,a^i)\ne (s_t,a^i_t),
		\end{align}
	\end{subequations}
	where $\eta_t^i$ is agent~$i$'s step size at time~$t$. A common approach for agent~$i$ to select its actions is the so-called $\epsilon$-greedy method, i.e., by exploiting the learned $Q$-functions with high probability and randomly exploring any action with some small probability. If agent~$i$ uses Q-learning~\eqref{eq:agentQ} with the $\epsilon$-greedy method while all other agents use a fixed joint policy $\pi^{-i}$, then, agent~$i$ solves a stationary MDP and $P\left[Q_t^i\to Q_{\pi^{-i}}^i\right]= 1$ following the convergence result of Q-learning on stationary MDPs~(\citet{tsitsiklis1994asynchronous}). However, when all agents use Q-learning~\eqref{eq:agentQ} with the $\epsilon$-greedy method, then the MDP becomes nonstationary and convergence of the $Q$~functions is not guaranteed~(\citet{leslie2005individual}). To overcome this difficulty, \citet{arslan2016decentralized} proposed a fully decentralized Q-learning {algorithm}  where all agents use constant policies for extended periods of time, termed \emph{exploration phases}. The $k$th exploration phase runs through time $t=t_k,\ldots,t_{k+1}-1$, where $t_{k+1} = t_k+T_k$ for some positive integer $T_k$ (with $t_0=0$). During the $k$th exploration phase, each agent~$i$ has some deterministic \emph{baseline policy} $\pi_k^i$, but uses the same randomized policy~$\bar{\pi}_k^i$ throughout the phase, where 
	$$
	\bar{\pi}_k^i(s_t):=\left\{\begin{array}{cl} \pi_k^i(s_t), & \textrm{ w.p. } 1-\rho^i\\ \textrm{any } a^i\in\mathcal{A}^i, & \textrm{ w.p. } \rho^i/|\mathcal{A}^i|, \end{array} \right.
	$$
	{for some $\rho^i\in(0,1)$.}
	Equivalently, we can write 
	\begin{equation}
		\label{eq:pibar}
		\bar{\pi}_k^i=(1-\rho^i)\pi_k^i+\rho^i\nu^i,
	\end{equation}
	where $\nu^i$ is the random policy that assigns the uniform distribution on $\mathcal{A}^i$ to each $s$.
	In words, agent~$i$ plays the baseline policy with probability $1-\rho^i$, and plays all actions uniformly with probability $\rho^i/|\mathcal{A}^i|$. We denote by $\bar{\Pi}$ the set of joint policies in the form of~\eqref{eq:pibar} for each agent, i.e., $\bar{\Pi}:=\left\{\bar{\pi}\mid \bar{\pi}^i=(1-\rho^i)\pi^i+\rho^i\nu^i, \pi^i\in\Pi^i,\forall i\in[N]\right\}$. Each agent updates its $Q$~function after each step according to~\eqref{eq:agentQ}, but updates its baseline policy only at the end of every exploration phase, by using the BRPI with some estimated $\Pi_{\pi_k^{-i}}^i$. The complete algorithm is presented as Algorithm~\ref{al:1}.
	
	\begin{algorithm}[ht]
		\caption{Q-learning for agent~$i$}
		\label{al:1}
		\algsetblock[Name]{Parameters}{}{0}{}
		\algsetblock[Name]{Initialize}{}{0}{}
		\algsetblock[Name]{Define}{}{0}{}
		\begin{algorithmic}[1]
			\Statex {Set parameters}
			\Statex \hspace*{5mm} $\mathbb{Q}^i$: some  compact subset of the Euclidian space $\mathbb{R}^{|\mathcal{S}\times\mathcal{A}^i|}$
			\Statex \hspace*{5mm} $\{T_k\}_{k\geq0}$: sequence of integers in $[1,\infty)$
			\Statex \hspace*{5mm} $K\in\mathbb{Z}_+$: number of exploration phases
			\Statex \hspace*{5mm} $\rho^i\in(0,1)$: experimentation probability
			\Statex \hspace*{5mm} $\lambda^i\in(0,1)$: inertia
			\Statex \hspace*{5mm} $\zeta^i\in(0,\infty)$: tolerance level for sub-optimality
			\Statex \hspace*{5mm} $\{\eta_{t}^{i}\}_{t\geq0}$:  sequence  of step sizes 
			\vspace{2mm}
			\State Initialize  $\pi_0^i \in \Pi^i$ (arbitrary), $Q_0^i\in\mathbb{Q}^i$ (arbitrary)
			\State Receive $s_0$
			\For {$k=1,2\ldots$}
			\For {$t=t_k,\ldots,t_{k+1}-1$}
			\State $a_t^i = \bar{\pi}_k^i(s_t):=\left\{\begin{array}{cl} \pi_k^i(s_t), & \textrm{ w.p. } 1-\rho^i\\ \textrm{any } a^i\in\mathcal{A}^i, & \textrm{ w.p. } \rho^i/|\mathcal{A}^i| \end{array} \right.$
			\State Receive $r^i(s_t,a_t^i,a_t^{-i})$
			\State Receive $s_{t+1}$ (selected according to  $P[ \ \cdot \ | \ s_t,a_t^i,a_t^{-i}]$)
			\State $Q_{t+1}^i(s_t,a_t^i)  =   (1-\eta_t^{i})Q_t^i(s_t,a_t^i) + \eta_t^{i} \left[ r^i(s_t,a_t^i,a_t^{-i})  +  \gamma^i \max_{a^i\in\mathcal{A}^i} Q_t^i(s_{t+1},a^i) \right]$
			\State $Q_{t+1}^i(s,a^i) =   Q_t^i(s,a^i)$,   for all $(s,a^i)\neq (s_t,a_t^i)$
			\EndFor
			\State $\Pi_{k+1}^i  = \big\{\hat{\pi}^i\in\Pi^i:  Q_{t_{k+1}}^i(s,\hat{\pi}^i(s))\geq \max_{a^i\in\mathcal{A}^i}Q_{t_{k+1}}^i(s,a^i)-\frac{1}{2}\zeta^i, \ \mbox{for all}  \ s\big\}$
			\If{$\pi_k^i\in \Pi^i_{k+1}$}
			\State $\pi_{k+1}^i = \pi_k^i$
			\Else
			\State $\pi_{k+1}^i = \left\{\begin{array}{cl} \pi_{k}^i, & \textrm{ w.p. } \lambda^i\\ \textrm{any } \pi^i\in\Pi_{k+1}^i, & \textrm{ w.p. } (1-\lambda^i)/|\Pi_{k+1}^i|\end{array} \right.  $
			\EndIf
			\State $Q_{t_{k+1}}^i \gets$ projection of $Q_{t_{k+1}}^i$ onto $\mathbb{Q}^i$
			\EndFor
		\end{algorithmic}
	\end{algorithm}
	
	\citet{arslan2016decentralized} proved that the joint policy $\pi_k$ obtained from Algorithm~\ref{al:1} \emph{asymptotically} converges to some equilibrium policy. We will show in this paper the non-asymptotic convergence guarantees of the algorithm. To proceed, we first impose the following two assumptions. 


\begin{assumption}
	\label{as:alpha}
	There exist some $\kappa >0$, and a finite integer $H\geq 1$, such that for any pair of states $(s^{\prime},s)$, there exists a sequence of joint actions $\tilde{a}_0,\dots,\tilde{a}_{H-1}\in\mathcal{A}$ such that $$P[ s_{H}=s^{\prime} \ | \ (s_0,a_0,\dots,a_{H-1})=(s,\tilde{a}_0,\dots,\tilde{a}_{H-1})]\ge \kappa.$$
\end{assumption}
	
	Recall from the definition of~$\bar{\pi}_k$ that each agent has positive probability of choosing any action $a^i\in\mathcal{A}^i$, which implies that the joint actions taken by all agents can be any $a\in\mathcal{A}$ with positive probability. This, together with Assumption~\ref{as:alpha}, implies that all states communicate with each other in the Markov chain induced by the joint policy~$\bar{\pi}_k$, i.e., the Markov chain is irreducible. This is the same assumption as made in~\citet{arslan2016decentralized} except that we denote by $\kappa$ the lower
	bound on the probabilities.
	

\begin{assumption}\label{as:aperiodic}
	For any joint policy $\bar{\pi}_k$, the induced Markov chain is aperiodic.
\end{assumption}


It is common to assume that the Markov chain induced by the behavior policy is ergodic in analyzing the sample complexity of single-agent Q-learing~(\citet{li2020sample}). Assumption~\ref{as:aperiodic}, together with Assumption~\ref{as:alpha}, ensures that the Markov chain is finite, irreducible, and aperiodic, which implies that the chain is uniformly ergodic~(\citet{paulin2015concentration}) and admits a unique stationary distribution.

Let $\mu_{\bar{\pi}_k}$ be the stationary distribution over all states of the  Markov chain induced by $\bar{\pi}_k$, and let $\mu_{\bar{\pi}_k}^i$ be the stationary distribution over all $(s,a^i)\in \mathcal{S}\times\mathcal{A}^i$ pairs. We further define
	\begin{align}\label{eq:mumin}
		\mu_{\min,k} := \min_{i\in[N]}\min_{(s,a^i)\in\mathcal{S}\times \mathcal{A}^i}\mu_{\bar{\pi}_k}^i\left(s,a^i\right).
	\end{align}
	Here, $\min_{(s,a^i)\in\mathcal{S}\times \mathcal{A}^i}\mu_{\bar{\pi}_k}^i\left(s,a^i\right) := \mu_{\min,k}^i$ is the minimum probability of the stationary distribution over all state-action pairs from the perspective of agent~$i$, and $\mu_{\min,k}$ is obtained by taking the minimum over all agents. Intuitively, the smaller $\mu_{\min,k}$ is, the more samples are needed to ensure that all state-action pairs (from the perspective of each agent) are visited sufficiently many times during the $k$th exploration phase. Moreover, we define the mixing time of agent $i$ at the $k$th exploration phase as:
	\begin{align}\label{eq:tmixi}
		t_{{\rm mix},k}^i(\alpha) := \min\left\{t\ \Big|\ \max_{(s_0,a_0^i)\in\mathcal{S}\times\mathcal{A}^i}d_{\rm TV}\left(P^t(\cdot\mid s_0,a_0^i),\mu_{\bar{\pi}_k}^i\right)\le \alpha\right\},
	\end{align}
	where $\alpha\in (0,1)$, $P^t(\cdot\mid s_0,a^i_0)$ is the distribution of $(s_t,a_t^i)$ conditioned on the initial state-action pair $(s_0,a^i_0)$, and $d_{\rm TV}$ measures the total variation between two distributions.  Intuitively, $t_{{\rm mix},k}^i$ describes how fast sample trajectory of the Markov chain converges to the stationary distribution of state-action pairs from the perspective of agent~$i$. Further, let $t_{{\rm mix},k}(\alpha):=\max_{i\in[N]}t_{{\rm mix},k}^i(\alpha)$. Note that the convergence rate of a uniformly ergodic Markov chain to its stationary distribution is exponential~(\citet{haggstrom2002finite}). We therefore do not expect $t_{{\rm mix},k}$ to be excessively large. 
	
	We next define the minimum separation between the agents' optimal Q-functions (with respect to deterministic policies), which is regarded as an upper bound of the tolerence level $\zeta^i$ for all agents.
	
	\begin{align}\label{eq:zetabar}
		\bar{\zeta}:=\min_{\begin{array}{c} \scriptstyle i,s,a^i,\tilde{a}^i,\pi^{-i}\in\Pi^{-i}: \\ \scriptstyle Q_{{\pi^{-i}}}^i(s,a^i)\not=Q_{{\pi^{-i}}}^i(s,\tilde{a}^i) \end{array}} \left|Q_{{\pi^{-i}}}^i(s,a^i)-Q_{{\pi^{-i}}}^i(s,\tilde{a}^i)\right|.
	\end{align}
	
 {For notational convenience,} we let $A:=\max_{i\in[N]}|\Acal^i|$, $\bar{\gamma} := \max_{i\in[N]}\gamma^i$, and $\underline{\gamma} = \min_{i\in[N]}\gamma^i$.
	We now present our main theorem on the sample complexity of Algorithm~\ref{al:1}.

	\begin{theorem}\label{thm:1}
		Consider a discounted stochastic game that is weakly acyclic under strict best replies~\eqref{eq:pi-i}. Suppose that  each agent updates its policies by Algorithm~\ref{al:1}. Let Assumptions~\ref{as:alpha} and~\ref{as:aperiodic} hold. Then, there exist some constants $c_0$ and $c_1$ such that, for any $0<\delta<1$, one has {that} for all $k\ge K$,
		\begin{align*}
			P\left[\pi_k\in\Pi_{\rm eq}\right] \ge 1-\delta,
		\end{align*}
		provided that for all $i\in[N]$ and $k\in[K]$,
		\begin{subequations}
			\begin{align}
				T_k &\ge \frac{c_0}{\mu_{\min,k}}\left\{\frac{1}{(1-\bar{\gamma})^5{\epsilon}^2}+\frac{t_{{\rm mix},k}\left(\frac{1}{4}\right)}{1-\bar{\gamma}}\right\}\log\left(\frac{NL|\mathcal{S}|AT_k}{\tilde{\delta}}\right)\log\left(\frac{1}{(1-\bar{\gamma})^2{\epsilon}}\right)\label{eq:Tk},\\
				K &\ge \frac{\left[\left(1-\tilde{\delta}\right)^2\hat{p}-\tilde{\delta}^2\right]L}{\left[\tilde{\delta}+\left(1-\tilde{\delta}\right)\hat{p}\right]^2\tilde{\delta}},\\
				\eta_t^i &= \frac{c_1}{\log\left(\frac{NL|\mathcal{S}||\mathcal{A}^i|T_k}{\tilde{\delta}}\right)} \min\left\{\frac{(1-\bar{\gamma})^4{\epsilon}^2}{\bar{\gamma}^2},\frac{1}{t_{{\rm mix},k}\left(\frac{1}{4}\right)}\right\}, \quad\forall t=t_{{k}},\ldots,t_{{k}+1}-1,\\
				\rho^i &=  1 - \left(1-\frac{(\bar{\zeta}/8-{\epsilon})(1-\bar{\gamma})}{\Gamma}\right)^{\frac{1}{N-1}}:=\rho,\\
				\zeta^i &= \frac{\bar{\zeta}}{2},
			\end{align}
		\end{subequations}
	where $\bar{\zeta}$ and $\hat{p}$ are as defined in~\eqref{eq:zetabar} and~\eqref{eq:ppp}, respectively, and $\Gamma$ is some absolute constant~(formally defined in~\eqref{eq:Gamma}) which depends only on the game parameters, ${\epsilon}:={\min\left\{\frac{\bar{\zeta}}{16}, \frac{1}{2(1-\underline{\gamma})}\right\}}$, 
	and $\tilde{\delta}$ is such that
	\begin{align*}
		\delta = 1-\left(\frac{\left(1-\tilde{\delta}\right)\hat{p}}{\tilde{\delta}+\left(1-\tilde{\delta}\right)\hat{p}} - \tilde{\delta}\right)\left(1-\tilde{\delta}\right).
	\end{align*}
	\end{theorem}

Theorem~\ref{thm:1} provides a finite-sample result for Algorithm~\ref{al:1}. To be more explicit on the   results, we further obtain the following bounds for $t_{{\rm mix},k}(\alpha)$ and $\mu_{\min,k}$.
\begin{proposition}\label{prop:bounds}
	For all $k\in[K]$ and $i\in[N]$, we have that
	\begin{subequations}
		\begin{align}
			\mu_{\min,k}^i &\le \left[1-\left(|\Scal|-1\right)\kappa \left(\prod_{i\in[N]}\frac{\rho^i}{|\Acal^i|}\right)^H\right]\cdot \frac{\rho^i}{|\Acal^i|},\\
			\mu_{\min,k}^i &\ge \kappa \frac{\rho^i}{|\Acal^i|} \left(\prod_{i\in[N]}\frac{\rho^i}{|\Acal^i|}\right)^H,\\
			t_{{\rm mix},k}(\alpha)&\le (H+1)\left(\frac{-\log \alpha}{\log \left[\frac{1-(|\Scal|-1)\kappa\left(\prod_{i\in[N]}\frac{\rho^i}{|\Acal^i|}\right)^H}{1-|\Scal|\kappa \left(\prod_{i\in[N]}\frac{\rho^i}{|\Acal^i|}\right)^H}\right]}+1\right).
		\end{align}
	\end{subequations}
	
	With $\rho^i = \rho$ for all $i\in[N]$ as in Theorem~\ref{thm:1}, we deduce that
	\begin{subequations}
		\begin{align}
			\mu_{\min,k} &\le \left[1-\left(|\Scal|-1\right)\kappa \frac{\rho^{NH}}{A^{NH}}\right]\cdot \frac{\rho}{A},\\
			\mu_{\min,k} &\ge \kappa \frac{\rho^{NH+1}}{A^{NH+1}},\label{eq:muminlow}\\
			t_{{\rm mix},k}(\alpha)
			&\le (H+1)\left((-\log \alpha)\frac{A^{NH}}{\kappa\rho^{NH}}+1\right).\label{eq:tmixup}
		\end{align}
	\end{subequations}
\end{proposition}

\begin{corollary}\label{cor:1}
	By applying~\eqref{eq:muminlow} and~\eqref{eq:tmixup} to Theorem~\ref{thm:1}, we may express the sample complexity of each exploration phase~\eqref{eq:Tk} as
	\begin{align}
		T_k\ge \frac{c_0A^{NH+1}}{ \kappa\rho^{NH+1}}\left\{\frac{1}{(1-\bar{\gamma})^5{\epsilon}^2}+\frac{(H+1)\left((\log 4)\frac{A^{NH}}{\kappa\rho^{NH}}+1\right)}{1-\bar{\gamma}}\right\}\log\left(\frac{NL|\mathcal{S}|AT_k}{\tilde{\delta}}\right)\log\left(\frac{1}{(1-\bar{\gamma})^2\hat{\epsilon}_k}\right).
	\end{align}
\end{corollary}

Note that the joint action space of all agents has size $\mathcal{O}(A^N)$. In Proposition~\ref{prop:bounds}, we have eliminated the dependence on~$\mu_{\min, k}$ and~$t_{\rm{mix}, k}$, so that the sample complexity of~$T_k$ is explicitly represented by the parameters of the game. In the following two subsections, we provide complete analyses and proofs for Theorem~\ref{thm:1} and Proposition~\ref{prop:bounds}.


\subsection{Proof of Theorem~\ref{thm:1}.}
We first introduce the following lemma, which is an application of the sample complexity result on single agent Q-learning~(\citet{li2020sample}).

\begin{lemma}
	\label{lm:Qappx}
	Fix any arbitrary $\pi_k\in\Pi$. For any $0<\hat{\delta}<1$ and $0<\epsilon\le \frac{1}{1-{\gamma}_{\min}}$, there exist some constants $c_{0,k}$ and $c_{1,k}^1,\ldots,c_{1,k}^N$ such that
	$$P\left[  \big|Q_{t_{k+1}}^i - Q_{\bar{\pi}_k^{-i}}^{i}\big|_{\infty} \leq \epsilon, \ \forall i\in[N]  \right] \geq 1-\hat{\delta},$$
	provided that the iteration number $T_k$ and the learning rates $\eta_t^i$ obey
	\begin{subequations}\label{eq:Teta}
		\begin{align}
			T_k&\ge\frac{c_{0,k}}{\mu_{\min,k}}\left\{\frac{1}{(1-\bar{\gamma})^5\epsilon^2}+\frac{t_{{\rm mix}, k}\left(\frac{1}{4}\right)}{1-\bar{\gamma}}\right\}\log\left(\frac{N|\mathcal{S}|AT_k}{\hat{\delta}}\right)\log\left(\frac{1}{(1-\bar{\gamma})^2\epsilon}\right),\\
			\eta_t^i &=\frac{c_{1,k}^i}{\log\left(\frac{N|\mathcal{S}||\mathcal{A}^i|T_k}{\hat{\delta}}\right)} \min\left\{\frac{(1-\gamma^i)^4\epsilon^2}{\left({\gamma^i}\right)^2},\frac{1}{t_{{\rm mix},k}^i\left(\frac{1}{4}\right)}\right\},\quad \forall t=t_k,\ldots,t_{k+1}-1, \ i\in [N],\label{eq:etati}
		\end{align}
	\end{subequations}
	where $\bar{\gamma} = \max_i\gamma^i$ and $A=\max_i|\mathcal{A}^i|$.
\end{lemma}
\begin{proof}[Proof of Lemma~\ref{lm:Qappx}]\
	Note that in the $k$th exploration phase, agents adopt the joint policy $\bar{\pi}_k$ as defined in~\eqref{eq:pibar}. Also by Assumptions~\ref{as:alpha} and~\ref{as:aperiodic}, the (finite) Markov chain induced by the joint policy is irreducible and aperiodic. Theorem~1 of~\citet{li2020sample} implies that for any agent $i$, there exist some constants $c_{0,k}^i$ and $c_{1,k}^i$ such that for any  $0<\delta_0<1$ and $0<\epsilon\le \frac{1}{1-{\gamma^i}}$,
	\begin{align*}
		P\left[  \big|Q_{t_{k+1}}^i - Q_{\bar{\pi}_k^{-i}}^{i}\big|_{\infty} \leq \epsilon\right] \geq 1-\delta_0
	\end{align*}
	provided that the iteration number $T_k$ and the learning rates $\eta_t^i$ obey
	\begin{subequations}
		\begin{align*}
			T_k&\ge\frac{c_{0,k}^i}{\mu_{\min,k}^i}\left\{\frac{1}{(1-\gamma^i)^5\epsilon^2}+\frac{t_{{\rm mix}, k}^i\left(\frac{1}{4}\right)}{1-\gamma^i}\right\}\log\left(\frac{|\mathcal{S}||\mathcal{A}^i|T_k}{\delta_0}\right)\log\left(\frac{1}{(1-\gamma^i)^2\epsilon}\right),\\
			\eta_t^i &=\frac{c_{1,k}^i}{\log\left(\frac{|\mathcal{S}||\mathcal{A}^i|T_k}{\delta_0}\right)} \min\left\{\frac{(1-\gamma^i)^4\epsilon^2}{\left({\gamma^i}\right)^2},\frac{1}{t_{{\rm mix},k}^i\left(\frac{1}{4}\right)}\right\},\quad \forall t=t_k,\ldots,t_{k+1}-1,
		\end{align*}
	\end{subequations}
	where
	$
		\mu_{\min,k}^i:=\min_{(s,a^i)\in\mathcal{S}\times \mathcal{A}^i}\mu_{\pi_k}^i\left(s,a^i\right)
	$,
	and $\mu_{\min,k}$ and $t_{{\rm mix}, k}^i$ are as defined in~\eqref{eq:mumin} and~\eqref{eq:tmixi}, respectively.  Let $c_{0,k}:=\max_{i\in[N]}c_{0,k}^i$. Then, with 
	\begin{subequations}
		\begin{align*}
			T_k&\ge\frac{c_{0,k}}{\mu_{\min,k}}\left\{\frac{1}{(1-\bar{\gamma})^5\epsilon^2}+\frac{t_{{\rm mix}, k}\left(\frac{1}{4}\right)}{1-\bar{\gamma}}\right\}\log\left(\frac{|\mathcal{S}|AT_k}{\delta_0}\right)\log\left(\frac{1}{(1-\bar{\gamma})^2\epsilon}\right),\\
			\eta_t^i&=\frac{c_{1,k}^i}{\log\left(\frac{|\mathcal{S}||\mathcal{A}^i|T_k}{\delta_0}\right)} \min\left\{\frac{(1-\gamma^i)^4\epsilon^2}{\left({\gamma^i}\right)^2},\frac{1}{t_{{\rm mix},k}^i\left(\frac{1}{4}\right)}\right\},\quad \forall t=t_k,\ldots,t_{k+1}-1, \ i\in[N],
		\end{align*}
	\end{subequations}
	we have that
	\begin{align*}
		P\left[  \big|Q_{t_{k+1}}^i - Q_{\bar{\pi}_k^{-i}}^{i}\big|_{\infty} \leq \epsilon\right] \geq 1-\delta_0,\quad \forall i\in[N],
	\end{align*}
	which implies that
	\begin{align*}
		P\left[  \big|Q_{t_{k+1}}^i - Q_{\bar{\pi}_k^{-i}}^{i}\big|_{\infty} > \epsilon\right] \le \delta_0,\quad \forall i\in[N].
	\end{align*}
	From the union bound, 
	\begin{align*}
		P\left[  \big|Q_{t_{k+1}}^i - Q_{\bar{\pi}_k^{-i}}^{i}\big|_{\infty} > \epsilon,\ \exists i\in[N]\right] \le \sum_{i\in[N]} P\left[  \big|Q_{t_{k+1}}^i - Q_{\bar{\pi}_k^{-i}}^{i}\big|_{\infty} > \epsilon\right]  \le   N\delta_0.
	\end{align*}
	Therefore,
	\begin{align*}
		P\left[\big|Q_{t_{k+1}}^i - Q_{\bar{\pi}_k^{-i}}^{i}\big|_{\infty}\le \epsilon,\ \forall i\in[N]\right] &= 1 - P\left[\big|Q_{t_{k+1}}^i - Q_{\bar{\pi}_k^{-i}}^{i}\big|_{\infty}> \epsilon,\ \exists i\in[N]\right]\ge 1-N\delta_0.
	\end{align*}
	The proof is completed by taking $\hat{\delta} = N\delta_0$.
\end{proof}

Lemma~\ref{lm:Qappx} bounds the approximation error of Q-learning for each agent, i.e., the difference of the Q-function obtained at the end of the $k$th exploration phase and the optimal Q-function in the best reply to $\bar{\pi}^{-i}$. Our next goal is to bound the approximation error of policy perturbation. 
Recall the definition of the randomized policy in~\eqref{eq:pibar}, and consider the joint policies of all agents except~$i$. With probability $\prod_{j\ne i}(1-\rho^j)$, all agents $j\ne i$ end up playing their baseline policies, which results in $\left|Q^i_{\pi_k^{-i}} - Q_{\bar{\pi}_k^{-i}}^i\right| = 0$, i.e. the approximation error of policy perturbation becomes zero in this case. When not all agents play their baseline policies, let $\varphi^{-i}\in\Delta^{-i}$ be some convex combination of the policies in $\Delta^{-i}$ of the form where each agent~$j\ne i$ either uses a baseline policy $\pi^j\in\Pi^j$ or the uniform distribution. More precisely, let $J$ denote the subset of agents choosing the baseline policies, and let 
\begin{align}\label{eq:phi-i}
	\varphi^{-i}=\sum_{J\subset \{1,\ldots,N\}\setminus \{i\}}a_J\varphi_J^{-i},
\end{align}
where $a_J:= \frac{\prod_{j\in J}(1-\rho^j)\prod_{j\notin J\cup\{i\}}\rho^j}{1-\prod_{j\ne i}(1-\rho^j)}$ and $\varphi_J\in\Delta^{-i}$ is such that $\varphi_J^j = \pi^j$ for $j\in J$ and $\varphi_J^j=\nu^j$ for $j\notin J\cup\{i\}$. Denote by $\bar{\Delta}^{-i}\subset\Delta^{-i}$ the set of all policies in the form of~\eqref{eq:phi-i}. Note that $\bar{\Delta}^{-i}$ is a finite set. Recall the definition of the Bellman operator from~\eqref{eq:bellman}. We then define 
\begin{align}\label{eq:Gamma}
	\Gamma:=\max_{\left(\pi^{-i},\varphi^{-i}\right)\in \Pi^{-i}\times \bar{\Delta}^{-i}}\left|\mathcal{T}_{\pi^{-i}}^i(Q^i_{\pi^{-i}}) - \mathcal{T}_{\varphi^{-i}}^i(Q^i_{\pi^{-i}}) \right|_\infty.
\end{align}
We next have the following lemma on the approximation error due to policy perturbation.
\begin{lemma}
	\label{lm:Qexp}
	Fix any arbitrary $\pi_k\in\Pi$. For any $\tilde{\epsilon}>0$, if $\rho^i$ satisfies
	\begin{align}\label{eq:rhoi}
		\rho^i\le 1 - \left(1-\frac{\tilde{\epsilon}(1-\bar{\gamma})}{\Gamma}\right)^{\frac{1}{N-1}},\quad\forall i\in[N],
	\end{align}
	then, we have that
	$$\left|Q_{\pi_k^{-i}}^i - Q_{\bar{\pi}_k^{-i}}^{i}\right|_{\infty} \leq \tilde{\epsilon}, \quad \forall i\in[N], k\in[K].$$
\end{lemma}
\begin{proof}[Proof of Lemma~\ref{lm:Qexp}]\
	First note that, for all $i\in[N]$ and $k\in[K]$,
	\begin{align}
		\left| Q_{\pi_k^{-i}}^i - Q_{\bar{\pi}_k^{-i}}^{i}\right|_{\infty}
		&=  \left|\mathcal{T}_{\pi_k^{-i}}^i(Q_{\pi_k^{-i}}^i) - \mathcal{T}_{\bar{\pi}_k^{-i}}^i(Q_{\bar{\pi}_k^{-i}}^{i}) \right|_{\infty}\nonumber\\
		&\leq \left|\mathcal{T}_{\pi_k^{-i}}^i(Q_{\pi_k^{-i}}^i) - \mathcal{T}_{\bar{\pi}_k^{-i}}^i(Q_{\pi_k^{-i}}^{i})\right|_{\infty}   +\left|\mathcal{T}_{\bar{\pi}_k^{-i}}^i(Q_{\pi_k^{-i}}^{i}) - \mathcal{T}_{\bar{\pi}_k^{-i}}^i(Q_{\bar{\pi}_k^{-i}}^{i})\right|_{\infty}.\label{eq:TQ}
	\end{align}
	By definition of $\bar{\pi}_k^{-i}$, we have that $P\left[\bar{\pi}_k^{-i} = \pi_k^{-i}\right] = \prod_{j\ne i}(1-\rho^j)$. With probability $1-\prod_{j\ne i}(1-\rho^j)$, $\bar{\pi}_k^{-i}\ne \pi_k^{-i}$ and $\bar{\pi}_k^{-i}\in \bar{\Delta}^{-i}$. Thus, the first term of~\eqref{eq:TQ} can be bounded by
	\begin{align}\label{eq:TQ1}
		\left|\mathcal{T}_{\pi_k^{-i}}^i(Q_{\pi_k^{-i}}^i) - \mathcal{T}_{\bar{\pi}_k^{-i}}^i(Q_{\pi_k^{-i}}^{i})\right|_{\infty}\le \left(1-\prod_{j\ne i}(1-\rho^j)\right)\times \left| \mathcal{T}_{\pi_k^{-i}}^i(Q_{\pi_k^{-i}}^{i})-\mathcal{T}_{\varphi_k^{-i}}^i(Q_{\pi_k^{-i}}^{i})\right|_{\infty},
	\end{align}
	for some $\varphi_k^{-i}\in\bar{\Delta}^{-i}$.
	On the other hand, by the contraction mapping of the Bellman operator, we have that
	\begin{align}\label{eq:TQ2}
		\left|\mathcal{T}_{\bar{\pi}_k^{-i}}^i(Q_{\pi_k^{-i}}^{i}) - \mathcal{T}_{\bar{\pi}_k^{-i}}^i(Q_{\bar{\pi}_k^{-i}}^{i})\right|_{\infty}\le \gamma^i\left| Q_{\pi_k^{-i}}^i - Q_{\bar{\pi}_k^{-i}}^{i}\right|_{\infty}.
	\end{align}
	Substituting~\eqref{eq:TQ1} and~\eqref{eq:TQ2} back into~\eqref{eq:TQ}, we have that
	\begin{align*}
		\left| Q_{\pi_k^{-i}}^i - Q_{\bar{\pi}_k^{-i}}^{i}\right|_{\infty}&\le \left(1-\prod_{j\ne i}(1-\rho^j)\right)\times \left| \mathcal{T}_{\pi_k^{-i}}^i(Q_{\pi_k^{-i}}^{i})-\mathcal{T}_{\varphi_k^{-i}}^i(Q_{\pi_k^{-i}}^{i})\right|_{\infty} + \gamma^i\left| Q_{\pi_k^{-i}}^i - Q_{\bar{\pi}_k^{-i}}^{i}\right|_{\infty}\\
		&\le \left(1-\prod_{j\ne i}(1-\rho^j)\right) \Gamma + \gamma^i\left| Q_{\pi_k^{-i}}^i - Q_{\bar{\pi}_k^{-i}}^{i}\right|_{\infty},
	\end{align*}
	which implies that
	\begin{align*}
		\left| Q_{\pi_k^{-i}}^i - Q_{\bar{\pi}_k^{-i}}^{i}\right|_{\infty}&\le \frac{\left(1-\prod_{j\ne i}(1-\rho^j)\right) \Gamma}{1-\gamma^i}\le \frac{\left(1-\prod_{j\ne i}(1-\rho^j)\right) \Gamma}{1-\bar{\gamma}}.
	\end{align*}
	If for all~$i\in[N]$, $\rho^i\le 1 - \left(1-\frac{\tilde{\epsilon}(1-\bar{\gamma})}{\Gamma}\right)^{\frac{1}{N-1}}$, then, we have that $1-\rho^j\ge \left(1-\frac{\tilde{\epsilon}(1-\bar{\gamma})}{\Gamma}\right)^{\frac{1}{N-1}}$, which implies that $\prod_{j\ne i}(1-\rho^j) \ge 1-\frac{\tilde{\epsilon}(1-\bar{\gamma})}{\Gamma}$, and thus
	\begin{align*}
		\left| Q_{\pi_k^{-i}}^i - Q_{\bar{\pi}_k^{-i}}^{i}\right|_{\infty}\le \frac{\left(1-\prod_{j\ne i}(1-\rho^j)\right) \Gamma}{1-\bar{\gamma}}\le\tilde{\epsilon}.
\end{align*}
	The above holds for all $i\in[N]$ and $k\in[K]$, which completes the proof.
\end{proof}

Recall from~\eqref{eq:zetabar} that $\bar{\zeta}$ is the minimum separation between the entries of agents' optimal  Q-functions (with respect to the deterministic policies):
\begin{align*}
	\bar{\zeta}:=\min_{\begin{array}{c} \scriptstyle i,s,a^i,\tilde{a}^i,\pi^{-i}\in\Pi^{-i}: \\ \scriptstyle Q_{\pi^{-i}}^i(s,a^i)\not=Q_{\pi^{-i}}^i(s,\tilde{a}^i) \end{array}} \left|Q_{\pi^{-i}}^i(s,a^i)-Q_{\pi^{-i}}^i(s,\tilde{a}^i)\right|.
\end{align*}
We assume that $\bar{\zeta}>0$ to avoid trivial cases, and consider $\bar{\zeta}$ as an upper bound on $\zeta^i$ for all $i$. We next define the following random event for any arbitrary $\pi_k\in\Pi$:
\begin{align*}
	E_k:=\Big\{\omega\in\Omega:  \left|Q_{t_{k+1}}^i - Q_{\pi_k^{-i}}^{i}\right|_{\infty} < & \frac{1}{4}\min\{\zeta^i,\bar{\zeta}-\zeta^i\},  \forall i   \Big\}.
\end{align*}
With this definition of $E_k$, we show that, if $E_k$ is not empty and $\pi_k\in\Pi_{\rm eq}$, then $\pi_{k+1} = \pi_k$ with probability~$1$. 
\begin{lemma}\label{lm:pik1}
	Given any $\pi_k\in\Pi$ and the corresponding $E_k$, for all $k$, we have that
	$$P\left[\pi_{k+1} = \pi_k\mid E_k,\ \pi_k\in\Pi_{\rm eq}\right] = 1.$$
\end{lemma}
\begin{proof}[Proof of Lemma~\ref{lm:pik1}]\
	Let $\hat{a}^{i*} := \argmax_{\hat{a}^i}Q^i_{t_{k+1}}\left(s,\hat{a}^i\right)$. Then, conditioned on $E_k$ and $\pi_k\in\Pi_{\rm eq}$, we have that
	\begin{align*}
		\max_{\hat{a}^i}Q^i_{t_{k+1}}\left(s,\hat{a}^i\right) -Q^i_{t_{k+1}}\left(s,\pi_k^i(s)\right) &= Q^i_{t_{k+1}}\left(s,\hat{a}^{i*}\right) -Q^i_{t_{k+1}}\left(s,\pi_k^i(s)\right)\\ &=\left[Q^i_{t_{k+1}}\left(s,\hat{a}^{i*}\right) - Q^i_{\pi_k^{-i}}\left(s,\pi_k^i(s)\right)\right] + \left[Q^i_{\pi_k^{-i}}\left(s,\pi_k^i(s)\right)- Q^i_{t_{k+1}}\left(s,\pi_k^i(s)\right)\right] \\
		&<Q^i_{t_{k+1}}\left(s,\hat{a}^{i*}\right) - Q^i_{\pi_k^{-i}}\left(s,\pi_k^i(s)\right)+\frac{1}{2}\min\left\{\zeta^i,\bar{\zeta}-\zeta^i\right\}\\
		&<\left[Q^i_{t_{k+1}}\left(s,\hat{a}^{i*}\right) - Q^i_{\pi_k^{-i}}\left(s,\hat{a}^{i*}\right)\right] + \left[Q^i_{\pi_k^{-i}}\left(s,\hat{a}^{i*}\right)- Q^i_{\pi_k^{-i}}\left(s,\pi_k^i(s)\right)\right]\\
		&\qquad +\frac{1}{4}\min\left\{\zeta^i,\bar{\zeta}-\zeta^i\right\}\\
		&< \frac{1}{4}\min\left\{\zeta^i,\bar{\zeta}-\zeta^i\right\} + \frac{1}{4}\min\left\{\zeta^i,\bar{\zeta}-\zeta^i\right\}
		\le\frac{1}{2}\min\left\{\zeta^i,\bar{\zeta}-\zeta^i\right\},
	\end{align*}
	where the second-to-last inequality follows since $Q^i_{\pi_k^{-i}}\left(s,\hat{a}^i\right)- Q^i_{\pi_k^{-i}}\left(s,\pi_k^i(s)\right)<0$, which follows from $\pi_k\in\Pi_{\rm eq}$. 
	It follows that $Q^i_{t_{k+1}}\left(s,\pi_k^i(s)\right)\ge \max_{\hat{a}^i}Q^i_{t_{k+1}}\left(s,\hat{a}^i\right) - \frac{1}{2}\zeta^i$ for all~$i$. Then, by Algorithm~\ref{al:1} (lines~11-13), we have that $\pi_{k+1} = \pi_k$ with probability~$1$.
\end{proof}
Recall that $L$ is the maximum length of the shortest strict best reply path from any policy to an equilibrium policy. Our next lemma lower bounds the conditional probability of $\pi_{k+L}$ being an equilibrium policy, given that $\pi_k$ is not an equilibrium policy and $E_k,\ldots,E_{k+L-1}$.
\begin{lemma}\label{lm:p0}
	Let 
	\begin{align}\label{eq:phat}
		\hat{p} := \left(\min_{j\in \{1,\ldots,N\}}\left\{ \frac{1-\lambda^j}{\left|\Pi^{j}\right|}\cdot \prod_{i\ne j}\lambda^i\right\}\right)^L,
	\end{align}
	which is the same as that in~\eqref{eq:ppp}.	We then have that
	\begin{align}
		P\left[  \pi_{k+L} \in \Pi_{\rm eq} \ \big| \ E_k,\dots,E_{k+L-1},   \pi_k\not\in\Pi_{\rm eq} \right] \geq \hat{p}.\label{eq:ne2e}
	\end{align}
\end{lemma}
\begin{proof}[Proof of Lemma~\ref{lm:p0}]\
	We begin with an important observation. Consider some $\pi_{k}\notin \Pi_{\rm eq}$; then, there must exist at least one agent, say agent~$i$, whose policy $\pi^i_k$ is not the best reply to $\pi_k^{-i}$, i.e., $\pi^i_k\notin \Pi^i_{\pi_k^{-i}}$. In this case, we claim that $\pi_k^i \notin \Pi^i_{k+1}$, where $\Pi^i_{k+1}$ is as defined in Algorithm~\ref{al:1} (line~11). In other words, the ``else'' statement in Algorithm~\ref{al:1} (line~15) will be executed. To see this, it suffices to show that $Q_{t_{k+1}}^i(s,\pi^i_k(s)) < \max_{a^i\in\mathcal{A}^i}Q_{t_{k+1}}^i(s,a^i)-\frac{1}{2}\zeta^i$ for some $s\in \mathcal{S}$. 
	Conditioned on $E_k$, we have that 
	$$
	Q^i_{\pi_k^{-i}}(s,a^i) - \frac{1}{4}\min\{\zeta^i, \bar{\zeta}-\zeta^i\} < Q^i_{t_{k+1}}(s,a^i) < Q^i_{\pi_k^{-i}}(s,a^i) + \frac{1}{4}\min\{\zeta^i, \bar{\zeta}-\zeta^i\},$$ 
	i.e., $Q^i_{t_{k+1}}(s,a^i)$ lies within a distance of $\frac{1}{4}\min\{\zeta^i, \bar{\zeta}-\zeta^i\}$ to $Q^i_{\pi_k^{-i}}(s,a^i)$. Moreover, we note that $\frac{1}{4}\min\{\zeta^i, \bar{\zeta}-\zeta^i\} \leq \frac{1}{8}\bar{\zeta}$.  Recall that $\left\{Q^i_{\pi_k^{-i}}(s,a^i): a^i\in \mathcal{A}^i\right\}$ are dispersed with spacing being at least $\bar{\zeta}$, where $\bar{\zeta}$ is as defined in~\eqref{eq:zetabar} as the minimum separation between the optimal $Q$-functions. Thus, it follows that the possible range of $Q^i_{t_{k+1}}(s,a^i)$ for all $a^i\in\mathcal{A}^i$ are mutually exclusive, which implies that the $\tau$-th best action under $Q^i_{\pi_k^{-i}}$ is identical to that under $Q^i_{t_{k+1}}$, i.e., $$
	\argmax_{a^i\in\mathcal{A}^i}\left(Q^i_{\pi_k^{-i}}(s,a^i)\right)_{(\tau)} = \argmax_{a^i\in\mathcal{A}^i}\left(Q^i_{t_{k+1}}(s,a^i)\right)_{(\tau)},
	$$ 
	where $(\cdot)_{(\tau)}$ represents the $\tau$-th largest value. For instance, when $\tau=1$, we have $\argmax_{a^i\in\mathcal{A}^i}Q^i_{\pi_k^{-i}}(s,a^i)= \argmax_{a^i\in\mathcal{A}^i}Q^i_{t_{k+1}}(s,a^i)$, which are denoted by $a^{i*}_{\pi_k^{-i}}(s)$ and $a^{i*}_{t_{k+1}}(s)$, respectively. 
	
	Since $\pi^i_k\notin \Pi^i_{\pi_k^{-i}}$, it follows that $\pi^i_k(s)\neq \argmax_{a^i\in\mathcal{A}^i} Q^i_{\pi_k^{-i}}(s,a^i)=:a^{i*}_{\pi_k^{-1}}(s)$ for some $s\in\mathcal{S}$. Then, we have that
	\begin{align*}
		\max_{a^i\in\mathcal{A}^i}Q_{t_{k+1}}^i(s,a^i) - Q_{t_{k+1}}^i(s,\pi^i_k(s)) 
		&> \left(\max_{a^i\in\mathcal{A}^i}Q_{\pi_k^{-i}}^i(s,a^i) - \frac{1}{8}\bar{\zeta}\right) - \left(Q_{\pi_k^{-i}}^i(s,\pi^i_k(s)) + \frac{1}{8}\bar{\zeta}\right)\\
		&= \left(Q_{\pi_k^{-i}}^i\left(s,a^{i*}_{\pi_k^{-i}}(s)\right) - Q_{\pi_k^{-i}}^i(s,\pi^i_k(s))\right) -  \frac{1}{4}\bar{\zeta}\\
		&\geq \bar{\zeta} - \frac{1}{4}\bar{\zeta} =\frac{3}{4}\bar{\zeta} \geq \frac{3}{4}\zeta^i> \frac{1}{2}\zeta^i
	\end{align*}
	as desired. Now, we are ready to prove the statement.
	
	Let $l$ be the length of the shortest strict best reply path from $\pi_k$ to an equilibrium policy. Then $l\le L$. Let the sequence of policies along the path be $\pi_0,\pi_1,\ldots, \pi_l$, with $\pi_0 = \pi_k\notin \Pi_{\rm eq}$ and $\pi_l\in\Pi_{\rm eq}$. Further, let $i_1,\ldots,i_l$ be the agent that changes her policy at each update, i.e., $\pi_{n-1}$ and $\pi_{n}$ differ only at agent $i_n$, for all $n=1,\ldots,l$. 
	Then, based on the aforementioned observation, we can use the two probabilities in the policy update rule in Algorithm~\ref{al:1}~(line~15) to yield
	\begin{align*}
		&\quad P\left[  \pi_{k+L} \in \Pi_{\rm eq} \ \big| \ E_k,\dots,E_{k+L-1},   \pi_k\not\in\Pi_{\rm eq} \right]\ge  P\left[  \pi_{k+L} =\pi_l \ \big| \ E_k,\dots,E_{k+L-1},   \pi_k\not\in\Pi_{\rm eq} \right]\\
		&\ge P\left[  \pi_{k+1} =\pi_1,\pi_{k+2} =\pi_2,\ldots,\pi_{k+l}=\pi_l, \pi_{k+l+1}=\cdots=\pi_{k+L}=\pi_l \ \big| \ E_k,\dots,E_{k+L-1},   \pi_k\not\in\Pi_{\rm eq} \right]\\
		&\ge 
		P\left[  \pi_{k+1} =\pi_1 \ \big| \ E_k,\dots,E_{k+L-1},   \pi_k=\pi_0 \right]\cdot P\left[  \pi_{k+2} =\pi_2 \ \big| \ E_k,\dots,E_{k+L-1},   \pi_k=\pi_0, \pi_{k+1}=\pi_1 \right]\\
		&\quad \cdot P\left[  \pi_{k+3} =\pi_3 \ \big| \ E_k,\dots,E_{k+L-1},   \pi_k=\pi_0, \pi_{k+1}=\pi_1, \pi_{k+2} = \pi_2 \right] \cdot \cdots \\
		&\quad \cdot P\left[  \pi_{k+l} =\pi_l \ \big| \ E_k,\dots,E_{k+L-1},   \pi_k=\pi_0, \pi_{k+1}=\pi_1,\ldots,\pi_{k+l-1} = \pi_{l-1} \right]\\
		&\quad \cdot P\left[  \pi_{k+l+1} =\pi_l \ \big| \ E_k,\dots,E_{k+L-1},   \pi_k=\pi_0, \pi_{k+1}=\pi_1,\ldots,\pi_{k+l} = \pi_{l} \right] \cdot \cdots\\
		&\quad \cdot P\left[  \pi_{k+L} =\pi_l \ \big| \ E_k,\dots,E_{k+L-1},   \pi_k=\pi_0, \pi_{k+1}=\pi_1,\ldots,\pi_{k+l} = \pi_{l},\ldots,\pi_{k+L-1} = \pi_l \right]\\
		&\ge \prod_{j\in\{i_1,\ldots,i_l\}}\left(\frac{1-\lambda^j}{\left|\Pi^{j}\right|}\cdot \prod_{i\ne j}\lambda^i\right)\ge \left(\min_{j\in \{1,\ldots,N\}} \left\{\frac{1-\lambda^j}{\left|\Pi^{j}\right|}\cdot \prod_{i\ne j}\lambda^i\right\}\right)^l \ge \left(\min_{j\in \{1,\ldots,N\}}\left\{ \frac{1-\lambda^j}{\left|\Pi^{j}\right|}\cdot \prod_{i\ne j}\lambda^i\right\}\right)^L,
	\end{align*}
	where we have used the fact from Lemma~\ref{lm:pik1}: given $\pi_l\in\Pi_{\rm eq}$ and the events $E_k,\ldots,E_{k+L-1}$, the conditional probability that $\pi_s\in\Pi_{\rm eq}$ is~$1$ for all $s\ge l$.
\end{proof}

We will then bound $P\left[E_k,\ldots,E_{k+L-1}\right]$. Before that, we first look at $P[E_k]$. We would like $P[E_k]$ to be as large as possible. Note that $\frac{1}{4}\min\{\zeta^i,\bar{\zeta}-\zeta^i\}\le \frac{1}{8}\bar{\zeta}$, with equality holding when $\zeta^i = \frac{1}{2}\bar{\zeta}$. We next have the following lemma. 
\begin{lemma}
	\label{lm:Pi}
	Let $\zeta^i = \frac{\bar{\zeta}}{2}$ for all $i\in[N]$. Fix an arbitrary $\pi_k\in\Pi$. For any $0<\hat{\delta}<1$,  we have that
	$$P\left[ E_k  \right] \geq 1-\hat{\delta},$$
	provided that
	$	\rho^i\le 1 - \left(1-\frac{\left(\bar{\zeta}/8 - \epsilon\right)(1-\bar{\gamma})}{\Gamma}\right)^{\frac{1}{N-1}}$, and $T_k$ and $\eta_t^i$ satisfy~\eqref{eq:Teta}, where $\epsilon$ can take any value in $0<\epsilon<\min\left\{\frac{\bar{\zeta}}{8}, \frac{1}{1-{\gamma}_{\min}}\right\}$.
\end{lemma}
\begin{proof}[Proof of Lemma~\ref{lm:Pi}]\
	A direct implication of Lemma~\ref{lm:Qappx} and Lemma~\ref{lm:Qexp} is that when $T_k$ and $\eta_t^i$ satisfy~\eqref{eq:Teta}, and $\rho^i$ satisfies~\eqref{eq:rhoi}, then, by triangle inequality, we have that
	\begin{align}
		P\left[\left|Q^i_{t_{k+1}}-Q^i_{\pi_k^{-i}}\right|_\infty\le \epsilon+\tilde{\epsilon},\quad \forall i\in[N]\right]\ge 1-\hat{\delta}.
	\end{align}
	The lemma then follows by taking $\tilde{\epsilon} = \frac{1}{8}\bar{\zeta}-\epsilon$.
\end{proof}
We then have the following lemma which bounds $P\left[E_k,\ldots,E_{k+L-1}\right]$.
\begin{lemma}\label{lm:EkL}
	For any arbitrary sequence of joint policies $\pi_k,\ldots,\pi_{k+L-1}\in\Pi$, and for any $0<\tilde{\delta}<1$, we have that
	\begin{align*}
		P\left[E_k,\ldots,E_{k+L-1}\right]\ge 1-\tilde{\delta},
	\end{align*}
	provided that for all $i\in[N]$ and for all $\hat{k}\in\{k,\ldots,k+L-1\}$,
	\begin{subequations}\label{eq:Teta2}
		\begin{align}
			T_{\hat{k}}&\ge\frac{c_{0,\hat{k}}}{\mu_{\min,\hat{k}}}\left\{\frac{1}{(1-\bar{\gamma})^5\epsilon^2}+\frac{t_{{\rm mix}, \hat{k}}\left(\frac{1}{4}\right)}{1-\bar{\gamma}}\right\}\log\left(\frac{NL|\mathcal{S}|AT_{\hat{k}}}{\tilde{\delta}}\right)\log\left(\frac{1}{(1-\bar{\gamma})^2\epsilon}\right),\\
			\eta_t^i &=\frac{c_{1,\hat{k}}^i}{\log\left(\frac{NL|\mathcal{S}||\mathcal{A}^i|T_{\hat{k}}}{\tilde{\delta}}\right)} \min\left\{\frac{(1-\gamma^i)^4\epsilon^2}{\left({\gamma^i}\right)^2},\frac{1}{t_{{\rm mix},\hat{k}}^i\left(\frac{1}{4}\right)}\right\},\quad \forall t=t_{\hat{k}},\ldots,t_{\hat{k}+1}-1,\label{eq:etati2}\\
			\rho^i&\le 1 - \left(1-\frac{\left(\bar{\zeta}/8 - \epsilon\right)(1-\bar{\gamma})}{\Gamma}\right)^{\frac{1}{N-1}}\\
			\zeta^i &= \frac{\bar{\zeta}}{2}
		\end{align}
	\end{subequations} 
	where $\epsilon$ can take any value in $0<\epsilon<\min\left\{\frac{\bar{\zeta}}{8}, \frac{1}{1-{\gamma}_{\min}}\right\}$.
\end{lemma}
\begin{proof}[Proof of Lemma~\ref{lm:EkL}]\
	When the conditions of Lemma~\ref{lm:Pi} are satisfied, we have that $P\left[E_k^c\right] < \hat{\delta}$, where $E_k^c$ is the complement of $E_k$. Then,
	\begin{align*}
		P\left[E_k,\ldots,E_{k+L-1}\right]&= 1-P\left[\left(E_k,\ldots,E_{k+L-1}\right)^c\right]=1-P\left[E_k^c\cup\cdots\cup E_{k+L-1}^c\right]\\
		&\ge 1-\left(P\left[E_k^c\right]+P\left[E_{k+1}^c\right]+\cdots +P\left[E_{k+L-1}^c\right]\right) = 1-L\hat{\delta}.
	\end{align*}
	By taking $\tilde{\delta} = L\hat{\delta}$, it follows that the conditions~\eqref{eq:Teta} now become~\eqref{eq:Teta2}, and the lemma is thus proved.
\end{proof}
As for the choice of~$\epsilon$, we would like to make $T_{\hat{k}}$ as small as possible. As~$\epsilon$ increases, the term $\frac{1}{(1-\bar{\gamma}^5\epsilon^2)}$ decreases, while $\rho^i$ also decreases, which leads to a smaller $\mu_{\min,\hat{k}}$ and a larger $t_{{\rm mix}, \hat{k}}$. Therefore, we would like to choose an optimal~$\hat{\epsilon}_{\hat{k}}$ for $T_{\hat{k}}$, such that (ignoring the logarithmic factors)
\begin{align}\label{eq:eps}
	\hat{\epsilon}_{\hat{k}}:=\argmin_{0<\epsilon<\min\left\{\frac{\bar{\zeta}}{8}, \frac{1}{1-{\gamma}_{\min}}\right\}}\frac{1}{\mu_{\min,\hat{k}}}\left(\frac{1}{(1-\bar{\gamma})^4{\epsilon}^2}+{t_{{\rm mix},\hat{k}}\left(\frac{1}{4}\right)}\right),
\end{align}
or, using the result of Proposition~\ref{prop:bounds} for the bounds of $t_{{\rm mix},\hat{k}}\left(\frac{1}{4}\right)$ and $\mu_{\min,\hat{k}}$, 
\begin{align}
	\hat{\epsilon}_{\hat{k}}=\argmin_{\epsilon\in \left(0,\min\left\{\frac{\bar{\zeta}}{8}, \frac{1}{1-{\gamma}_{\min}}\right\}\right)}\frac{A^{(N+1)H}}{\kappa \rho^{(N+1)H}}\left(\frac{1}{(1-\bar{\gamma})^4{\epsilon}^2}+{(H+1)\left(\frac{\log (1/4)}{
			\log \left[\frac{A^{NH}-|\Scal|\kappa\rho^{NH}}{ A^{NH}-\left(|\Scal|-1\right)\kappa\rho^{NH}}\right]
		}+1\right)}\right),
\end{align}
and we can define $\hat{\epsilon}:=\max_{\hat{k}\in\{k,\ldots,k+L-1\}}\hat{\epsilon}_k$, and let $\rho^i = 1 - \left(1-\frac{\left(\bar{\zeta}/8 - \hat{\epsilon}\right)(1-\bar{\gamma})}{\Gamma}\right)^{\frac{1}{N-1}}$, which leads to the optimal sample complexity for $T_{\hat{k}}$. For simplicity, we choose $\epsilon =\frac{1}{2}\min\left\{\frac{\bar{\zeta}}{8},\frac{1}{1-{\gamma}_{\min}}\right\}$  in Theorem~\ref{thm:1} and do not worry about optimizing over $\epsilon$.

Note also that the result of Lemma~\ref{lm:EkL} holds for any realization of $\pi_k\in\Pi$. Therefore, under the same conditions, we in fact have that
\begin{subequations}
	\begin{align}
		P\left[E_k,\ldots,E_{k+L-1} \ \big|\ \pi_k\in\Pi_{\rm eq}\right]\ge 1-\tilde{\delta},\label{eq:condPEin}\\
		P\left[E_k,\ldots,E_{k+L-1} \ \big|\ \pi_k\notin\Pi_{\rm eq}\right]\ge 1-\tilde{\delta}.\label{eq:condPEnot}
	\end{align}
\end{subequations}
By Lemma~\ref{lm:pik1} and~\eqref{eq:condPEin},  under conditions~\eqref{eq:Teta2}, we have that for all~$k$,
\begin{align}\label{eq:in}
	P\left[\pi_{k}=\pi_{k+1}=\cdots=\pi_{k+L}\ \big|\ \pi_k\in\Pi_{\rm eq}\right]\ge 1-\tilde{\delta}.
\end{align}
By Lemma~\ref{lm:p0} and~\eqref{eq:condPEnot}, under conditions~\eqref{eq:Teta2}, we have that for all~$k$,
\begin{align}\label{eq:out}
	P\left[\pi_{k+L}\in\Pi_{\rm eq}\ \big|\  \pi_k\notin\Pi_{\rm eq}\right]\ge \hat{p}\left(1-\tilde{\delta}\right).
\end{align}
As a notation, let $p_k := P\left[\pi_k\in\Pi_{\rm eq}\right]$. Then,~\eqref{eq:in} and~\eqref{eq:out} together imply that
\begin{align}
	p_{(n+1)L} \geq p_{nL}\left(1-\tilde{\delta}\right)  + (1-p_{nL})\hat{p}\left(1-\tilde{\delta}\right). \label{eq:geometricConv}
\end{align}
Rearranging the above, we obtain that
\begin{align}
	p_{(n+1)L} - p_{nL} &\ge    \left(1-\tilde{\delta}\right)\hat{p} - \tilde{\delta}p_{nL}-\left(1-\tilde{\delta}\right)\hat{p}p_{nL} =  \left[\tilde{\delta}+\left(1-\tilde{\delta}\right)\hat{p}\right]\left[\frac{\left(1-\tilde{\delta}\right)\hat{p}}{\tilde{\delta}+\left(1-\tilde{\delta}\right)\hat{p}} - p_{nL}\right]\label{eq:nL1}\\
	&\ge -\tilde{\delta}\label{eq:nL2}
\end{align}
Note that $p_{(n+1)L}-p_{nL}\ge 0$ as long as $p_{nL}\le \frac{\left(1-\tilde{\delta}\right)\hat{p}}{\tilde{\delta}+\left(1-\tilde{\delta}\right)\hat{p}}$. Further,  if $p_{nL}\le \frac{\left(1-\tilde{\delta}\right)\hat{p}}{\tilde{\delta}+\left(1-\tilde{\delta}\right)\hat{p}} - \tilde{\delta}$, then from~\eqref{eq:nL1}, we have that $p_{(n+1)L} - p_{nL}\ge \left[\tilde{\delta}+\left(1-\tilde{\delta}\right)\hat{p}\right]\tilde{\delta}$; if $p_{nL} > \frac{\left(1-\tilde{\delta}\right)\hat{p}}{\tilde{\delta}+\left(1-\tilde{\delta}\right)\hat{p}}$, then $p_{(n+1)L}-p_{nL}\ge -\tilde{\delta}$ from~\eqref{eq:nL2}. Therefore, we have that
\begin{align}
	p_{nL}\ge \frac{\left(1-\tilde{\delta}\right)\hat{p}}{\tilde{\delta}+\left(1-\tilde{\delta}\right)\hat{p}} - \tilde{\delta}, \quad \forall n\ge \tilde{n},
\end{align}
where
\begin{align}
	\tilde{n}:= \frac{\frac{\left(1-\tilde{\delta}\right)\hat{p}}{\tilde{\delta}+\left(1-\tilde{\delta}\right)\hat{p}} - \tilde{\delta}}{\left[\tilde{\delta}+\left(1-\tilde{\delta}\right)\hat{p}\right]\tilde{\delta}} = \frac{\left(1-\tilde{\delta}\right)^2\hat{p}-\tilde{\delta}^2}{\left[\tilde{\delta}+\left(1-\tilde{\delta}\right)\hat{p}\right]^2\tilde{\delta}}.
\end{align}
 This, together with~\eqref{eq:in}, implies that for all $n\ge \tilde{n}$,
 \begin{align}\label{eq:deltatilde}
 	P\left[\pi_{nL}=\pi_{nL+1}=\cdots=\pi_{nL+L}\in\Pi_{\rm eq}\right] \ge \left(\frac{\left(1-\tilde{\delta}\right)\hat{p}}{\tilde{\delta}+\left(1-\tilde{\delta}\right)\hat{p}} - \tilde{\delta}\right)\left(1-\tilde{\delta}\right) := f(\tilde{\delta}).
 \end{align}
Therefore, if the number of exploration phases $k\ge K:=\tilde{n}L$, then $P\left[\pi_k\in\Pi_{\rm eq}\right]\ge f(\tilde{\delta})$. Note that $f(\tilde{\delta})$ is continuous, decreasing in $\tilde{\delta}$, and $f(0) = 1$, $f(\delta) < 1-\delta$ for any $0<\delta<1$. Thus, we can take $\tilde{\delta}\in (0,\delta)$ such that 
\begin{align}
	\left(\frac{\left(1-\tilde{\delta}\right)\hat{p}}{\tilde{\delta}+\left(1-\tilde{\delta}\right)\hat{p}} - \tilde{\delta}\right)\left(1-\tilde{\delta}\right) =1-\delta,
\end{align}
which leads to $P\left[\pi_k\in\Pi_{\rm eq}\right]\ge 1-\delta$ for all $k\ge K$, and completes the proof of Theorem~\ref{thm:1}.

\subsection{Proof of Proposition~\ref{prop:bounds}.}
We first show the bounds for $\mu_{\min,k}^i$ and $\mu_{\min,k}$. By Assumption~\ref{as:alpha}, for any $s_1$ and $s_{H+1}$, there exists a sequence of joint actions $\tilde{a}_1,\ldots,\tilde{a}_{H}$ such that $P[ s_{H+1}=s \ | \ (s_1,a_1,\dots,a_{H})=(s,\tilde{a}_1,\dots,\tilde{a}_{H})]\ge \kappa$. Thus, we have that 
\begin{align}
	P\left(s_{H+1}={s}\mid s_1\right) &= \sum_{a_1,\ldots,a_H} P\left(s_{H+1} = s\mid (s_1,a_1,\ldots,a_H)\right)P\left((s_1,a_1,\ldots,a_H)\mid s_1\right)\nonumber\\
	&\ge P\left(s_{H+1} = s\mid (s_1,\tilde{a}_1,\ldots,\tilde{a}_H)\right)P\left((s_1,\tilde{a}_1,\ldots,\tilde{a}_H)\mid s_1\right)\nonumber\\
	&\ge \kappa \left(\prod_{i\in[N]}\frac{\rho^i}{|\Acal^i|}\right)^H,\quad\forall s_1,s\in\Scal,\label{eq:Plow}
\end{align}
where the last inequality follows from Assumption~\ref{as:alpha} and the action selection (Line~5) of Algorithm~\ref{al:1}. Then, we can write the lower bound for the stationary distribution  $\mu_{\bar{\pi}_k}$ over all states:
\begin{align}
	\mu_{\bar{\pi}_k}(s) &= \sum_{s_1\in\Scal}  \mu_{\bar{\pi}_k}(s_1)P\left(s_{H+1}={s}\mid s_1\right) \ge \kappa \left(\prod_{i\in[N]}\frac{\rho^i}{|\Acal^i|}\right)^H,\quad\forall s\in\Scal,\label{eq:muslow}
\end{align}
where the inequality follows since~\eqref{eq:Plow} is a uniform lower bound for all $s_1,s$ and $\sum_{s_1\in\Scal}\mu_{\pi_k}(s_1)=1$. Note that~\eqref{eq:muslow} also implies the following upper bound:
	\begin{align}
	\mu_{\bar{\pi}_k}(s) &=  1- \sum_{\bar{s}\in\Scal\setminus \{s\}}\mu_{\bar{\pi}_k}(\bar{s}) \le 1-\left(|\Scal|-1\right)\kappa \left(\prod_{i\in[N]}\frac{\rho^i}{|\Acal^i|}\right)^H, \quad\forall s\in\Scal. \label{eq:musup}
\end{align}
From~\eqref{eq:muslow} and~\eqref{eq:musup}, we can also write the bounds for the minimum probability of the stationary distribution over state-action pairs (from the perspective of agent~$i$):
\begin{subequations}\label{eq:muminik}
	\begin{align}
		\mu^i_{\min,k} &= \min_{(s,a^i)\in\Scal\times \Acal^i}\mu_{\bar{\pi}_k}^i\left(s,a^i\right) = \min_{(s,a^i)\in\Scal\times \Acal^i}P(a^i\mid s)\mu_{\bar{\pi}_k}(s) \ge \kappa \frac{\rho^i}{|\Acal^i|} \left(\prod_{i\in[N]}\frac{\rho^i}{|\Acal^i|}\right)^H,\\
		\mu_{\min,k}^i&=\min_{(s,a^i)\in\mathcal{S}\times \mathcal{A}^i}\mu_{\bar{\pi}_k}^i\left(s,a^i\right)= \min_{(s,a^i)\in\mathcal{S}\times \mathcal{A}^i}\mu_{\bar{\pi}_k}(s)\cdot  P(a^i\mid s)\le \left[1-\left(|\Scal|-1\right)\kappa \left(\prod_{i\in[N]}\frac{\rho^i}{|\Acal^i|}\right)^H\right]\cdot \frac{\rho^i}{|\Acal^i|}.
	\end{align}
\end{subequations}
With $\rho^i = \rho $ for all $i\in[N]$ and recalling that $A:=\max_{i\in[N]}\Acal^i$,~\eqref{eq:muminik} leads to
\begin{subequations}
	\begin{align}
		\mu_{\min,k} &= \min_{i\in[N]}\mu^i_{\min,k} \ge \kappa\frac{\rho^{NH+1}}{A^{NH+1}},\\
		\mu_{\min,k} &= \min_{i\in[N]}\mu^i_{\min,k} \le  \left[1-\left(|\Scal|-1\right)\kappa \frac{\rho^{NH}}{A^{NH}}\right]\cdot \frac{\rho}{A}.
	\end{align}
\end{subequations}

We next show the upper bound for $t_{{\rm mix},k}$. To proceed, we first introduce the following lemma, which follows directly from the coupling inequality under the \emph{Doeblin condition} (see~\citet{diaconis2011mathematics} for a detailed explanation).

\begin{lemma}\label{lem:temp}
	If $P\left(s_{H+1} = \bar{s}, a_{H+1}^i = \bar{a}^i \mid s_0, a_0^i\right) \ge c \mu_{\bar{\pi}_k}^i\left(\bar{s},\bar{a}^i\right)$ for all $s_0, a_0^i$ and $\bar{s},\bar{a}^i$, then, 
	$$
	d_{\rm TV}\left(P^t(\cdot\mid s_0,a_0^i),\mu_{\bar{\pi}_k}^i\right) = \max_{\bar{s},\bar{a}^i}\left|P\left(s_{t} = \bar{s}, a_t^i = \bar{a}^i \mid s_0, a_0^i\right) - \mu_{\bar{\pi}_k}^i\left(\bar{s},\bar{a}^i\right)\right|\le (1-c)^{\left\lfloor t/(H+1)\right\rfloor}.
	$$
\end{lemma}
We will use Lemma~\ref{lem:temp} to prove the upper bound of $t_{{\rm mix},k}$. To apply Lemma~\ref{lem:temp}, we need to find the parameter~$c$ such that $P\left(s_{H+1} = \bar{s}, a_{H+1}^i = \bar{a}^i \mid s_0, a_0^i\right) \ge c \mu_{\bar{\pi}_k}^i\left(\bar{s},\bar{a}^i\right)$ for all $s_0, a_0^i$ and $\bar{s},\bar{a}^i$. This leads to the following lemma.
\begin{lemma}\label{lem:c}
	For all $i\in[N]$ and for any $(s_0,a_0^i), (\bar{s},\bar{a}^i)\in \Scal\times \Acal^i$, we have that 
	\begin{align}\label{eq:lemc}
		P\left(s_{H+1} = \bar{s}, a_{H+1}^i = \bar{a}^i \mid s_0, a_0^i\right) \ge c \mu_{\bar{\pi}_k}^i\left(\bar{s},\bar{a}^i\right),
	\end{align}
	where 
	\begin{align}
		c = \frac{\kappa \left(\prod_{i\in[N]}\frac{\rho^i}{|\Acal^i|}\right)^H}{1-\left(|\Scal|-1\right)\kappa \left(\prod_{i\in[N]}\frac{\rho^i}{|\Acal^i|}\right)^H}.
	\end{align}
\end{lemma}
\begin{proof}[Proof of Lemma~\ref{lem:c}]\
	By Assumption~\ref{as:alpha}, we know that for all $s,s'$, there exist $\tilde{a}_0,\ldots,\tilde{a}_{H-1}$ such that
	\begin{align*}
		P[ s_{H}=s^{\prime} \ | \ (s_0,a_0,\dots,a_{H-1})=(s,\tilde{a}_0,\dots,\tilde{a}_{H-1})]\ge \kappa.
	\end{align*}
Then, we have that
\begin{align}
	P\left(s_{H+1} = \bar{s}, a_{H+1}^i = \bar{a}^i \mid s_0, a_0^i\right) &= P\left(s_{H+1}=\bar{s}\mid s_0, a_0^i\right) \cdot P\left(a_{H+1}^i=\bar{a}\mid s_{H+1}=\bar{s}, s_0, a_0^i\right)\nonumber\\
	&= P\left(s_{H+1}=\bar{s}\mid s_0, a_0^i\right) \cdot P\left(a_{H+1}^i=\bar{a}\mid s_{H+1}=\bar{s}\right)\nonumber\\
	&\ge P\left(s_{H+1}=\bar{s}\mid s_0, a_0^i\right) \cdot \frac{\rho^i}{|\Acal^i|}.\label{eq:Pl}
\end{align}
To find a~$c$ such that~\eqref{eq:lemc} holds, it suffices to find a lower bound for $ P\left(s_{H+1}=\bar{s}\mid s_0, a_0^i\right)$. Note that
\begin{align*}
	P\left(s_{H+1}=\bar{s}\mid s_0, a_0^i\right) &= \sum_{s_1}P\left(s_{H+1}=\bar{s}\mid s_1, s_0, a_0^i\right) P\left(s_1\mid s_0,a_0^i\right)\\
	&= \sum_{s_1}P\left(s_{H+1}=\bar{s}\mid s_1\right) P\left(s_1\mid s_0,a_0^i\right).
\end{align*}
We make the following observations.
\begin{itemize}
	\item For any $s_1$, by the law of total probability, we have that
	\begin{align*}
		P\left(s_{H+1}=\bar{s}\mid s_1\right) &= \sum_{\{a_1,\ldots,a_H\}}P\left(s_{H+1}=\bar{s}\mid s_1, a_1,\ldots,a_H\right)P\left(a_1,\ldots,a_H\mid s_1\right)\\
		&= P\left(s_{H+1}=\bar{s}\mid s_1, \tilde{a}_1,\ldots,\tilde{a}_H\right)P\left(\tilde{a}_1,\ldots,\tilde{a}_H\mid s_1\right)\\
		&\qquad + \sum_{\{a_1,\ldots,a_H\}\ne \{\tilde{a}_1,\ldots,\tilde{a}_H\}}P\left(s_{H+1}=\bar{s}\mid s_1, a_1,\ldots,a_H\right)P\left(a_1,\ldots,a_H\mid s_1\right)\\
		&\ge \kappa \left(\prod_{i\in[N]}\frac{\rho^i}{|\Acal^i|}\right)^H.
	\end{align*}
	\item Since the above is true for any $s_1\in\Scal$, and $\sum_{s_1}P\left(s_1\mid s_0,a_0^i\right)=1$, we have that the convex combination $$P\left(s_{H+1}=\bar{s}\mid s_0, a_0^i\right) = \sum_{s_1}P\left(s_{H+1}=\bar{s}\mid s_1\right) P\left(s_1\mid s_0,a_0^i\right)\ge \kappa \left(\prod_{i\in[N]}\frac{\rho^i}{|\Acal^i|}\right)^H.$$
\end{itemize}
Thus,~\eqref{eq:Pl} now becomes
\begin{align}
	P\left(s_{H+1} = \bar{s}, a_{H+1}^i = \bar{a}^i \mid s_0, a_0^i\right) &\ge  P\left(s_{H+1}=\bar{s}\mid s_0, a_0^i\right) \cdot \frac{\rho^i}{|\Acal^i|} \ge \kappa \left(\prod_{i\in[N]}\frac{\rho^i}{|\Acal^i|}\right)^H\cdot \frac{\rho^i}{|\Acal^i|}. \label{eq:Pl1}
\end{align}
With~\eqref{eq:Pl1}, we can see that any~$c$ satisfying $\kappa \left(\prod_{i\in[N]}\frac{\rho^i}{|\Acal^i|}\right)^H\cdot \frac{\rho^i}{|\Acal^i|} \ge c \mu_{\bar{\pi}_k}^i\left(\bar{s},\bar{a}^i\right)$ for all $\bar{s},\bar{a}^i$ and for all $i$ will lead to~\eqref{eq:lemc}. Equivalently, we need
\begin{align*}
	c \le \kappa \left(\prod_{i\in[N]}\frac{\rho^i}{|\Acal^i|}\right)^H\cdot \min_{i\in[N]}\frac{ \frac{\rho^i}{|\Acal^i|}}{\mu_{\min,k}^i}.
\end{align*}
One option is to choose $c = \kappa \left(\prod_{i\in[N]}\frac{\rho^i}{|\Acal^i|}\right)^H\cdot \min_{i\in[N]}{ \frac{\rho^i}{|\Acal^i|}}$, which essentially uses $\mu^i_{\min,k}\le 1$. However, note that we can use the upper bound of $\mu^i_{\min,k}$ from~\eqref{eq:muminik}, and choose that
\begin{align*}
	c = \kappa \left(\prod_{i\in[N]}\frac{\rho^i}{|\Acal^i|}\right)^H\cdot \min_{i\in[N]}\frac{ \frac{\rho^i}{|\Acal^i|}}{\left[1-\left(|\Scal|-1\right)\kappa \left(\prod_{i\in[N]}\frac{\rho^i}{|\Acal^i|}\right)^H\right]\cdot \frac{\rho^i}{|\Acal^i|}} = \frac{\kappa \left(\prod_{i\in[N]}\frac{\rho^i}{|\Acal^i|}\right)^H}{1-\left(|\Scal|-1\right)\kappa \left(\prod_{i\in[N]}\frac{\rho^i}{|\Acal^i|}\right)^H}.
\end{align*}
This completes the proof of the lemma.
\end{proof}

Combining Lemma~\ref{lem:temp} and Lemma~\ref{lem:c}, we have that for all $i\in[N]$ and $(s_0,a_0^i)\in\Scal\times \Acal^i$,
\begin{align*}
	d_{\rm TV}\left(P^t(\cdot\mid s_0,a_0^i),\mu_{\bar{\pi}_k}^i\right) \le (1-c)^{\left\lfloor t/(H+1)\right\rfloor}&= \left[\frac{1-\left(|\Scal|-1\right)\kappa \left(\prod_{i\in[N]}\frac{\rho^i}{|\Acal^i|}\right)^H-\kappa \left(\prod_{i\in[N]}\frac{\rho^i}{|\Acal^i|}\right)^H}{1-\left(|\Scal|-1\right)\kappa \left(\prod_{i\in[N]}\frac{\rho^i}{|\Acal^i|}\right)^H}\right]^{\lfloor t/(H+1)\rfloor}\\
	&=\left[\frac{1-|\Scal|\kappa\left(\prod_{i\in[N]}\frac{\rho^i}{|\Acal^i|}\right)^H}{1-\left(|\Scal|-1\right)\kappa \left(\prod_{i\in[N]}\frac{\rho^i}{|\Acal^i|}\right)^H}\right]^{\lfloor t/(H+1)\rfloor}.
\end{align*}

Let $t$ be such that
\begin{align}
	\left[\frac{1-|\Scal|\kappa\left(\prod_{i\in[N]}\frac{\rho^i}{|\Acal^i|}\right)^H}{1-\left(|\Scal|-1\right)\kappa \left(\prod_{i\in[N]}\frac{\rho^i}{|\Acal^i|}\right)^H}\right]^{\lfloor t/(H+1)\rfloor} \le \left[\frac{1-|\Scal|\kappa\left(\prod_{i\in[N]}\frac{\rho^i}{|\Acal^i|}\right)^H}{1-\left(|\Scal|-1\right)\kappa \left(\prod_{i\in[N]}\frac{\rho^i}{|\Acal^i|}\right)^H}\right]^{ t/(H+1)-1}= \alpha. \label{eq:t}
\end{align} 
Then, we have that
$
	\max_{(s_0,a_0^i)\in\Scal\times \Acal^i}d_{\rm TV}\left(P^t(\cdot\mid s_0,a_0^i),\mu_{\bar{\pi}_k}^i\right)\le \alpha, \forall i\in[N],
$
which implies that $t_{{\rm mix},k}(\alpha) = \max_{i\in[N]}t^i_{{\rm mix},k} \le t$. Therefore, by solving~\eqref{eq:t} for $t$, we obtain an upper bound for $t_{{\rm mix},k}$:

%

\begin{align}
	t_{{\rm mix},k}(\alpha)\le t = (H+1)\left(\frac{-\log \alpha}{\log \left[\frac{1-(|\Scal|-1)\kappa\left(\prod_{i\in[N]}\frac{\rho^i}{|\Acal^i|}\right)^H}{1-|\Scal|\kappa \left(\prod_{i\in[N]}\frac{\rho^i}{|\Acal^i|}\right)^H}\right]}+1\right).\label{eq:tmix}
\end{align}

With $\rho^i = \rho$ for all $i\in[N]$,~\eqref{eq:tmix} becomes
\begin{align*}
	t_{{\rm mix},k}(\alpha)&\le (H+1)\left(\frac{-\log \alpha}{\log \left[\frac{1-(|\Scal|-1)\kappa\left(\frac{\rho}{\prod_{i\in[N]}|\Acal^i|}\right)^H}{1-|\Scal|\kappa \left(\prod_{i\in[N]}\frac{\rho}{\prod_{i\in[N]}|\Acal^i|}\right)^H}\right]}+1\right)\\
	&\le (H+1)\left(\frac{(-\log \alpha)\left[1-(|\Scal|-1)\kappa\left(\frac{\rho}{\prod_{i\in[N]}|\Acal^i|}\right)^H\right]}{\kappa\left(\frac{\rho}{\prod_{i\in[N]}|\Acal^i|}\right)^H}+1\right)\\
	&\le (H+1)\left((-\log \alpha)\frac{\prod_{i\in[N]}|\Acal^i|^H}{\kappa \rho^H}+1\right)\le (H+1)\left((-\log \alpha)\frac{A^{NH}}{\kappa\rho^{NH}}+1\right).
\end{align*}
This completes the proof of Proposition~\ref{prop:bounds}.

\subsection{Numerical Experiments: A Grid-World Game.} \label{sec: experiments tabular}
To demonstrate the effectiveness of Algorithm~\ref{al:1}, we test it on the classical grid-world experiment~(\citet{sutton2018reinforcement}). We consider a $3\times3$ grid. Each state is represented by a pair~$(s(1),s(2))$, where $s(1)$ is the ``x"-coordinate and $s(2)$ is the ``y"-coordinate on the grid. The set of states are $\mathcal{S} = \{(1,1), (1,2), (1,3), (2,1), (2,2), (2,3), (3,1), (3,2), (3,3)\}$ (as shown in Figure~\ref{fig: grid_illustration}). The upper-left state~$(1,1)$ is called the \emph{terminal state}: once the system reaches state~$(1,1)$, it will stay there and will not transit to other states. There are two agents, where agent~$1$ decides her actions in the vertical direction, and agent~$2$ decides her actions in the horizontal direction, i.e.,  $\mathcal{A}^1 = \{{\rm up, stay, down}\}$, $\mathcal{A}^2 = \{{\rm left, stay, right}\}$. The joint actions of agent~$1$ and agent~$2$ will determine the next state, which could be the same as the current state (if both agents stay) or move one cell in the respective direction on the grid (except when the system is at the terminal state, it will not transit to other states for any joint actions taken by the agents). To ensure the next state stays within the grid, the agents can only choose actions which do not point to a direction outside of the grid. For instance, at state~$(1,1)$, $(1,2)$ and~$(1,3)$, agent~$1$ can only choose ``down" or ``stay", but not ``up".  Thus, for each agent, there are~$6$ states where the agent can take two actions, and~$3$ states where the agent can take three actions, which implies that we have $1728\times 1728$ possible joint policies in total (for each of agent, $|\Pi^i| = 2^6 \times 3^3 = 1728$). When the system is at the terminal state, all actions yield a reward of~$0$. When the system is at any non-terminal state, all the actions yield a reward of~$-1$.

\begin{figure}[ht]
	\includegraphics[width=0.95\linewidth]{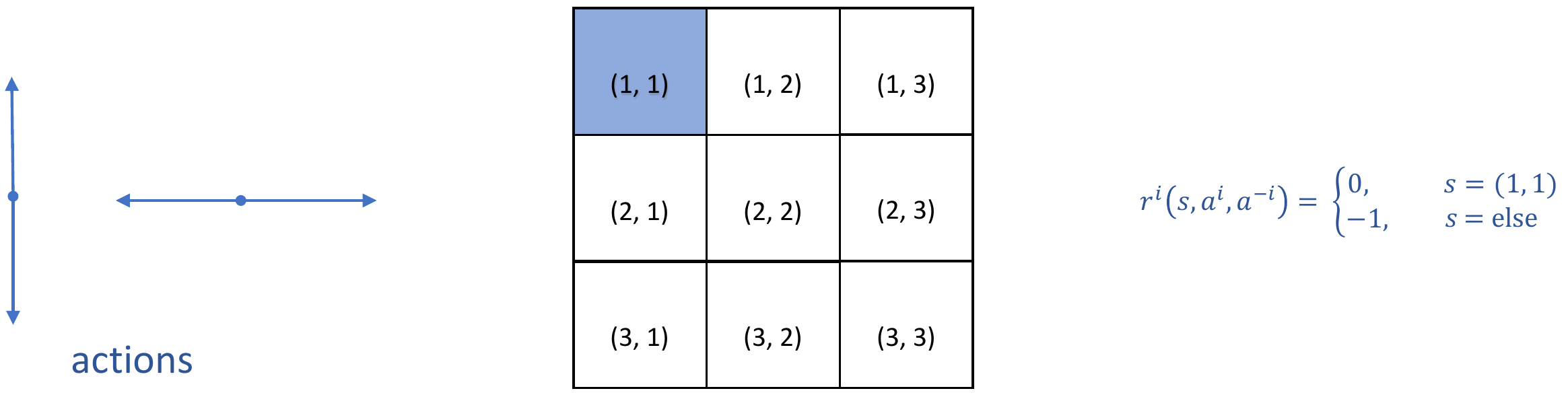}
	\centering
	\caption{Illustration for the grid-world experiment.
	}\label{fig: grid_illustration}
\end{figure}

For this grid-world stochastic game, if not considering actions taken on the terminal state, there are~$16$ \emph{optimal} equilibrium joint policies (as shown in Figure~\ref{fig: optimal}). These optimal equilibria will lead the agents to reach the terminal state as soon as possible and produce the maximum value of reward. Besides these optimal equilibria, there are also some \emph{suboptimal} equilibria. For example, one suboptimal equilibrium is when agent~$1$ chooses ``down" for all of the non-terminal states while agent~$2$ chooses ``right" for all of these states. This joint policy is still a Markov perfect equilibrium, but results in both agents earning less rewards (relative to those of the optimal equilibria). To obtain the set of all equilibrium joint policies, for each $\pi^{-i}\in\Pi^{-i}$, we use the standard Q-learning algorithm~\eqref{eq:agentQ} to get the set of its best reply policies $\Pi^{i}_{\pi^{-i}}$. Then, for each best reply policy $\pi^{*i}\in\Pi^i_{\pi^{-i}}$, we check if $\pi^{-i}$ also turns out to be a best reply policy of $\pi^{*i}$. If so, then the joint policy $(\pi^{-i}, ~\pi^{*i})$ consitutes a Markov perfect equilibrium. Moreover, if we start from any joint policy in $\Pi^1\times \Pi^2$, there is a strict best reply path to one of these equilibria, which means the grid-world stochastic game is weakly acyclic under strict best replies.

\begin{figure}[ht]
	\includegraphics[width=0.3\linewidth]{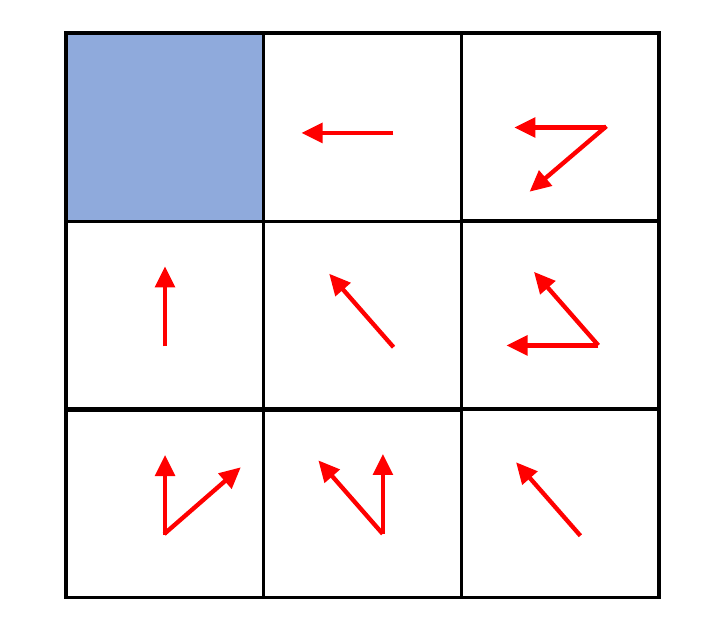}
	\centering
	\caption{The optimal equilibria for the grid-world experiment.
	}\label{fig: optimal}
\end{figure}

In our numerical experiments, we let $\rho^i = 0.4$, $\lambda^i = 0.3$, $\gamma^i = 0.75$, $\eta_k^i = 1/k^{0.5}$ for all $i$ and $k$. We fix the length of the exploration phases (the inner ``for" loop of Algorithm~\ref{al:1}), i.e., $T_k = T$, $\forall k = 1, \ldots, K$. Since both the length of the exploration phases~$T$ and the number of policy updates~$K$ are critical for the learning process, we run our experiments with different parameters of~$T$ and~$K$. For each set of~$K$ and~$T$, we start from a random joint policy in $\Pi^1 \times \Pi^2$ and run Algorithm~\ref{al:1}, with the initial values of the $Q$-table set to all~$0$'s. Then, we have a set of joint policies $\{\pi_k\mid k=1,\ldots, K\}$, and record the fraction of these policies that belong to the set of equilibrium joint policies obtained before. This process is repeated ~$50$ times, and the experimental results are shown in Figure~\ref{fig: tabular}. The solid lines represent the fraction of $\pi_k$'s which are equilibrium joint policies (the number of $\pi_k$'s which are equilibrium joint policies divided by~$K$), averaged over~$50$ repeated runs, and the shaded region represents the min-max interval. 

\begin{figure}[ht]
	\includegraphics[width=0.85\linewidth]{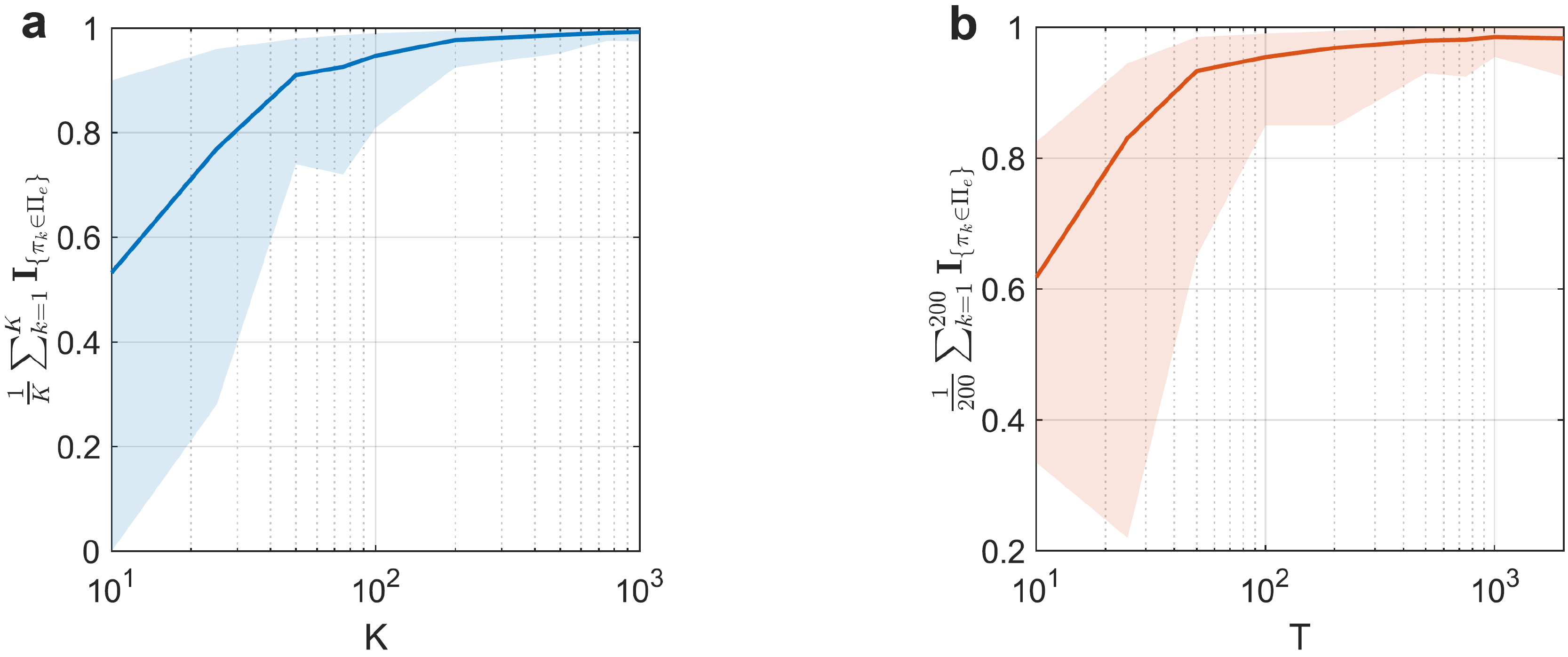}
	\centering
	\caption{Experimental results of grid-world when {\bf Algorithm~\ref{al:1}} is applied. {\bf a}, The fraction of times at which $\pi_k$ visits an equilibrium when $K$ ranges in the interval $[10,~ 1000]$. {\bf b}, The fraction of times at which $\pi_k$ visits an equilibrium when~$T$ ranges in the interval $[10,~ 2000]$. In~{\bf a}, we fix the value of~$T$ as~$200$. In~{\bf b}, we fix the value of $K$ as~$200$. The solid lines are the average of~$50$ repeated runs, and the shaded regions represent the min-max intervals.
	}\label{fig: tabular}
\end{figure}

It can be observed from Figure~\ref{fig: tabular}a that the minimum of the fraction $\frac{1}{K}\sum_{k= 1}^K \mathbf{I}_{\{\pi_k \in \Pi_{eq}\}}$ over the~$50$ repeated runs is greater than~$0$ when $K \geq 20$. This implies that the joint policy converges to an equilibrium at the end of the algorithm for all given initial policies. A similar phenomenon can also be observed from Figure~\ref{fig: tabular}b. Meanwhile, it can be seen that $\pi_k$ visits an equilibrium more often as $K$ and $T$ increase. This is consistent with Theorem~\ref{thm:1}, as $\pi_K$ is expected to be at equilibrium with higher probability as $K$ and $T$ increase.

\section{Decentralized Q-learning with linear function approximation.}\label{sec:linear}
	
Algorithm~\ref{al:1} in the previous section extends the single agent Q-learning algorithm. Specifically, in each exploration phase (an inner ``for" loop of Algorithm~\ref{al:1}), each agent updates and keeps track of its~$Q$ function which has dimension $|\Scal||\Acal^i|$. One major challange of such an algorithm is the curse of dimensionality -- when the number of state-action pairs is large, it becomes intractable. One popular approach to overcome this obstacle is to approximate the optimal $Q$~functions with functions from a much smaller space. We next describe the linear function approximation method where each agent's $Q$~function is approximated by a linear combination of $d$ basis functions. {To the best of our knowledge, there is no existing result on either the convergence or the sample complexity of decentralized Q-learning with linear function approximation for general-sum stochastic games. } 

Let $\left\{\phi^i_j:\Scal \times \Acal^i\mapsto \mathbb{R}\mid 1\le j\le d\right\}$ be the set of basis functions (features) of agent~$i$. We denote by $\phi^i := \left[\phi^i_1,\phi^i_2,\ldots,\phi^i_d\right]\in \mathbb{R}^{|\Scal||\Acal^i|\times d}$ the feature matrix of agent~$i$. The linear subspace $\Wcal^i$ spanned by the features $\left\{\phi^i_j\right\}$ is $\Wcal^i=\left\{{Q}^i_{\theta^i} := {\phi^i}^\top\theta^i\mid \theta^i\in\Theta^i\subset\mathbb{R}^d\right\}$ where $\Theta^i$ is some compact subset of $\mathbb{R}^d$ which contains the zero point and has diameter $D^i$, i.e., $D^i = \sup\left\{\left\|\theta^i_j -\theta^i_{j'}\right\|_2\ \Big|\ \theta^i_{j},\theta^i_{j'}\in \Theta^i\right\}$. We use $\Wcal^i$ as the approximation function space for agent~$i$. 

With {the linear function class}, we start with approximating agent~$i$'s optimal $Q$-function satisfying the fixed point equation of the Bellman operator, as given in~\eqref{eq:Qfp}. Specifically, for any joint policy played by all other agents, $\pi^{-i}\in\Delta^{-i}$, we define
\begin{align}\label{eq:thetapi-i}
	\theta^i_{\pi^{-i}} := \argmin_{\theta^i\in\Theta^i}\left\|Q^i_{\pi^{-i}}-{\phi^i}^\top \theta^i\right\|^2_2,
\end{align}
{where we recall that $Q^i_{\pi^{-i}}$ satisfies \eqref{eq:Qfp}.}
Agent~$i$'s set of deterministic best replies to $\pi^{-i}$ under linear function approximation is then {given by}
\begin{align}\label{eq:pipi-i}
	\widetilde{\Pi}^i_{\pi^{-i}} =  \left\{\tilde{\pi}^{i}\in\Pi^i: \ \phi^i\left(s,\tilde{\pi}^{i}(s)\right)^\top \theta^i_{\pi^{-i}}=\max_{a^i\in\mathcal{A}^i} \phi^i(s,a^i)^\top\theta^i_{\pi^{-i}},   \quad \forall s\in\mathcal{S} \right\}.
\end{align}
With~\eqref{eq:pipi-i}, we now define the (deterministic) linear approximated equilibrium. 
\begin{definition}\label{def:laequi}
	A deterministic joint policy $\pi^*\in\Pi$ is a linear approximated equilibrium if
	\begin{align*}
		\pi^{*i} \in \widetilde{\Pi}^i_{(\pi^*)^{-i}},\quad \forall i\in[N].
	\end{align*}
\end{definition}
We denote by $\widetilde{\Pi}_{\rm eq}$ the set of linear approximated equilibria. Our goal is to find a set of $\left\{\theta^i_{(\pi^*)^{-i}}\mid \forall i\in[N]\right\}$ for some $\pi^*\in\widetilde{\Pi}_{\rm eq}$ such that  ${\phi^i}^\top \theta^i_{(\pi^*)^{-i}}$ best represents $Q^i_{(\pi^*)^{-i}}$ in the sense of~\eqref{eq:thetapi-i} for all $i\in[N]$. For this goal to be feasible, we assume that $\widetilde{\Pi}_{\rm eq}$ is nonempty.
\begin{assumption}\label{as:exist}
	There exists at least one deterministic joint policy that is a linear approximated equilibrium.
\end{assumption}
Similar to the tabular case, we can define the \emph{best reply graph} on the set of deterministic joint policies, where each vertex is a deterministic joint policy and there is a direct edge from $\pi_k$ to $\pi_l$ if for some $i\in[N]$, $\pi^i_l\ne \pi^i_k$, $\pi^i_l = \pi^i_k,\forall j\ne i$, and $\pi^i_l\in \widetilde{\Pi}^i_{\pi_k^{-i}}$. The \emph{strict best reply path} and the \emph{weakly acyclic game} are defined analogously as in Definition~\ref{def:path} and Definition~\ref{def:best}. We again consider the weakly acyclic game, with the newly defined best replies as in~\eqref{eq:pipi-i}. There exists a strict best reply path from any $\pi\in\Pi$ to some ${\pi^*}\in\widetilde{\Pi}_{\rm eq}$. Let $\widetilde{L}_{\pi}$ be the minimum length of the strict best reply paths from $\pi$ to a linear approximated equilibrium policy, and let $\widetilde{L}:=\max_{\pi\in\Pi}\widetilde{L}_{\pi}$. We will again apply the BRPI, where the deterministic best reply sets $\widetilde{\Pi}^i_{\pi^{-i}}$ are approximated using Q-learning with linear function approximations.

In the fully decentralized setting, each agent is completely oblivious to other agents. Agent~$i$ may use the standard Q-learning algorithm under linear function approximation~(\citet{bertsekas1995neuro,melo2008analysis}), i.e., 
\begin{align}\label{eq:theta}
	\theta^i_{t+1}  =  \theta_t^i + \eta^i_t \phi^i(s_t,a^i_t)\left[ r^i(s_t,a_t^i,a_t^{-i})  +  \gamma^i \max_{a^i\in\mathcal{A}^i} \phi^i(s_{t+1},a^i)^\top \theta^i_t - \phi^i(s_t,a^i_t)^\top \theta_t^i\right],
\end{align}
and selects its actions by taking $a_t^i = \argmax_{a^i\in\Acal^i}\phi^i(s,a^i)^\top \theta^i$ with high probability and randomly exploring any actions with some small probability. The same problem as in the tabular setting arises here: if all agents use~\eqref{eq:theta} and select their actions with the $\epsilon$-greedy method, the environment becomes nonstationary and the convergence of the $\theta^i$'s is not guaranteed. In the same spirit of Algorithm~\ref{al:1}, we let each agent play the behavior policy $\bar{\pi}^i_k$ as defined in~\eqref{eq:pibar} during the $k$th exploration phase, so that the environment is stationary within each exploration phase. Instead of maintaining and updating a $|\Scal||\Acal^i|$-dimensional $Q$~function, agent~$i$ updates a $d$-dimensional vector $\theta^i$ according to~\eqref{eq:theta}. Note that~\eqref{eq:theta} may also be viewed as the stochastic approximation algorithm for solving the following equation:
%
\begin{align}\label{eq:projbellman}
	\E_{\mu_{\bar{\pi}_k}^i}\left[\phi^i(s,a^i)\left(r^i(s,a^i,a^{-i})  +  \gamma^i \max_{\hat{a}^{i}\in\mathcal{A}^i} \phi^i(s',\hat{a}^{i})^\top \theta^i - \phi^i(s,a^i)^\top \theta^i\right)\right]=0.
\end{align}

The complete decentralized Q-learning algorithm with linear
function approximation is presented as Algorithm~\ref{al:2}. In essence,
we replace the update of $Q$-functions in Algorithm~\ref{al:1} with~\eqref{eq:theta} as in line~8 of Algorithm~\ref{al:2}. Furthermore, in determining whether to update the current baseline policy $\pi_k^i$, we use the  best reply set $\widetilde{\Pi}^i_{k+1}$, which is defined with $\phi^i_{t_{k+1}}(s,a^i)^\top \theta^i_{t_{k+1}}$, in replacement of the set $\Pi^i_{k+1}$. 

\begin{algorithm}[ht]
	\caption{Q-learning for agent~$i$ with linear function approximation}
	\label{al:2}
	\algsetblock[Name]{Parameters}{}{0}{}
	\algsetblock[Name]{Initialize}{}{0}{}
	\algsetblock[Name]{Define}{}{0}{}
	\begin{algorithmic}[1]
		\Statex {Set parameters}
		\Statex \hspace*{5mm} $\Theta^i$: some compact subset of the Euclidian space $\mathbb{R}^{d}$ with diameter $D^i$
		\Statex \hspace*{5mm} $\{T_k\}_{k\geq0}$: sequence of integers in $[1,\infty)$
		\Statex \hspace*{5mm} $K\in\mathbb{Z}_+$: number of exploration phases
		\Statex \hspace*{5mm} $\rho^i\in(0,1)$: experimentation probability
		\Statex \hspace*{5mm} $\lambda^i\in(0,1)$: inertia
		\Statex \hspace*{5mm} $\zeta^i\in(0,\infty)$: tolerance level for sub-optimality
		\Statex \hspace*{5mm} $\{\eta_{t}^{i}\}_{t\geq0}$:  sequence  of step sizes 
		\vspace{2mm}
		\State Initialize  $\pi_0^i \in \Pi^i$ (arbitrary), $\theta_0^i\in\mathbb{R}^d$ (arbitrary)
		\State Receive $s_0$
		\For {$k=1,2\ldots$}
		\For {$t=t_k,\ldots,t_{k+1}-1$}
		\State $a_t^i = \bar{\pi}_k^i(s_t):=\left\{\begin{array}{cl} \pi_k^i(s_t), & \textrm{ w.p. } 1-\rho^i\\ \textrm{any } a^i\in\mathcal{A}^i, & \textrm{ w.p. } \rho^i/|\mathcal{A}^i| \end{array} \right.$
		\State Receive $r^i(s_t,a_t^i,a_t^{-i})$
		\State Receive $s_{t+1}$ (selected according to  $P[ \ \cdot \ | \ s_t,a_t^i,a_t^{-i}]$)
		\State $\theta^i_{t+1}  =  \theta_t^i + \eta^i_t \phi(s_t,a^i_t)\left[ r^i(s_t,a_t^i,a_t^{-i})  +  \gamma^i \max_{a^i\in\mathcal{A}^i} \phi^i(s_{t+1},a^i)^\top \theta^i_t - \phi(s_t,a^i_t)^\top \theta_t^i\right]$
		\EndFor
		\State $\widetilde{\Pi}_{k+1}^i  = \left\{\tilde{\pi}^i\in\Pi^i:  \phi_{t_{k+1}}^i(s,\tilde{\pi}^i(s))^\top \theta^i_{t_{k+1}}\geq \max_{a^i\in\mathcal{A}^i}\phi_{t_{k+1}}^i(s,a^i)^\top \theta^i_{t_{k+1}}-\frac{1}{2}\zeta^i_\theta, \ \mbox{for all}  \ s\right\}$
		\If{$\pi_k^i\in\widetilde{\Pi}_{k+1}^i$,}
		\State $\pi_{k+1}^i = \pi_k^i$
		\Else
		\State $\pi_{k+1}^i = \left\{\begin{array}{cl} \pi_{k}^i, & \textrm{ w.p. } \lambda^i\\ \textrm{any } \pi^i\in\widetilde{\Pi}^i_{k+1}, & \textrm{ w.p. } (1-\lambda^i)/|\widetilde{\Pi}^i_{k+1}|\end{array} \right.  $
		\EndIf
		\State $\theta_{t_{k+1}}^i \gets$ projection of $\theta_{t_{k+1}}^i$ onto ${\Theta}^i$
		\EndFor
	\end{algorithmic}
\end{algorithm}

We will show in this section the finite-sample convergence guarantee of Algorithm~\ref{al:2}. To proceed, we first impose the following assumptions for all agents, in addition to Assumptions~\ref{as:alpha},~\ref{as:aperiodic}, and~\ref{as:exist}.

\begin{assumption}\label{as:lin}
	The features $\{\phi^i_j\}_{1\le j\le d}$ are linearly independent and are normalized so that $\|\phi^i(s,a^i)\|\le 1$ for all state-action pairs $(s,a^i)$.
\end{assumption}

Assumption~\ref{as:lin} is imposed without loss of generality~(\citet{chen2019finite}): we can always scale the basis functions to ensure that $\max_{(s,a^i)\in\Scal\times \Acal^i}\|\phi^i(s,a^i)\|\le 1$, and any dependent features can be discarded. 

\begin{assumption}\label{as:conv}
	\eqref{eq:projbellman} has a unique solution, and there exists $\xi^i>0$ such that the following inequality holds for all $\theta^i\in\Theta^i$:
	\begin{align}
		(\gamma^{i})^2\E_\mu\left[\max_{a^i\in\Acal^i}(\phi^i(s,a^i)^\top \theta^i)^2\right] - \E_\mu\left[\left(\phi^i(s,a^i)^\top\theta^i\right)^2\right]\le \xi^i\|\theta^i\|^2_2.
	\end{align}
\end{assumption}
We note that~\eqref{eq:projbellman} may not admit a solution in general, and the iteration for $\theta^i_t$ in~\eqref{eq:theta} might diverge. Assumption~\ref{as:conv}, which is exactly the same as Assumption~3.3 in \citet{chen2019finite},  ensures the convergence of~\eqref{eq:theta}. See~\citet{chen2019finite,lee2019unified,melo2008analysis} for detailed discussions on this point.

Let $b$ be an upper bound on the minimum Bellman error under joint policy $\pi\in\Pi\cup \bar{\Pi}$, i.e., 
\begin{align}\label{eq:b}
	\min_{\theta^i\in\Theta^i}\left\|Q_{\theta^i}^i - \Tcal^i_{{\pi}^{-i}}\left(Q_{\theta^i}^i\right)\right\|_2 \le b,\quad \forall  i\in[N], \pi^{-i}\in\Pi^{-i}\cup\bar{\Pi}^{-i}.
\end{align}
Similar to~\eqref{eq:zetabar} where $\bar{\zeta}$ is defined, we define the minimum separation between the linear approximated Q-functions:
\begin{align}\label{eq:zetabart}
	\bar{\zeta}_\theta:=\min_{\begin{array}{c} \scriptstyle i,s,a^i,\tilde{a}^i,\pi^{-i}\in\Pi^{-i}: \\ \scriptstyle Q_{\theta^i_{\pi^{-i}}}^i(s,a^i)\not=Q_{\theta^i_{\pi^{-i}}}^i(s,\tilde{a}^i) \end{array}} \left|Q_{\theta^i_{\pi^{-i}}}^i(s,a^i)-Q_{\theta^i_{\pi^{-i}}}^i(s,\tilde{a}^i)\right|.
\end{align}
We next have the following assumption on the bound of minimum Bellman error.
\begin{assumption}\label{as:bound}
	The upper bound $b$ on the minimum Bellman error as given in~\eqref{eq:b} satisfies $b<\frac{(1-\bar{\gamma})\bar{\zeta}_\theta}{8}$.
\end{assumption}
Assumption~\ref{as:bound} bounds the minimum Bellman error from a constant factor of the minimum separation between the linear approximated Q-factors. This assumption is necessary to achieve a good approximation of $\widetilde{\Pi}^i_{\pi_k^{-i}}$ by $\widetilde{\Pi}^i_{k+1}$ in line~10 of Algorithm~\ref{al:2}, so that the algorithm closely mimics the BRPI. Without Assumption~\ref{as:bound}, there is no guarantee that a policy that is not a linear approximated equilibrium will be updated, i.e., an agent may not update its current policy even if it is not a best reply. 

For notational convenience, we let $D:=\max_{i\in[N]}D^i$, $\xi_{\min} := \min_{i\in[N]}\xi^i$, and $\eta_{\min} = \min_{i,t}\eta^i_t$. 
With the above definitions and assumptions, we now present our second main theorem on the sample complexity of Algorithm~\ref{al:2}, whose proof is given in Appendix~\ref{sec:appa}.

\begin{theorem}\label{thm:2}
	Consider a discounted stochastic game that is weakly acyclic under strict best replies~\eqref{eq:pipi-i}. Suppose that  each agent updates its policies by Algorithm~\ref{al:2}. Let Assumptions~\ref{as:alpha},~\ref{as:aperiodic},~\ref{as:exist},~\ref{as:lin},~\ref{as:conv}, and~\ref{as:bound} hold. Then, for any $0<\delta<1$, one has that for all $k\ge K$, 
	\begin{align*}
		P\left[\pi_k\in\widetilde{\Pi}_{\rm eq}\right]\ge 1-\delta
	\end{align*}
	provided that for all $i\in[N]$ and $k\in[K]$, 
	\begin{subequations}
		\begin{align}
			\eta^i_t&=\eta^i \le \frac{ \epsilon^2\tilde{\delta}\xi^i}{456N\widetilde{L}\left(1+\gamma^i+r_{\max}^i\right)^2(D^i+1)^2t_{{\rm mix},k}^i(\eta^i)}, \quad\forall t=t_{{k}},\ldots,t_{{k}+1}-1,\\
			T_{{k}}&\ge t_{{\rm mix},k}(\eta_{\min}) + \frac{\log \frac{\epsilon^2\tilde{\delta}}{2N\widetilde{L}(2D+1)^2}}{\log\left(1-\xi_{\min}\eta_{\min}/2\right)},\\
			K &\ge \frac{\left[\left(1-\tilde{\delta}\right)^2\tilde{p}-\tilde{\delta}^2\right]\widetilde{L}}{\left[\tilde{\delta}+\left(1-\tilde{\delta}\right)\tilde{p}\right]^2\tilde{\delta}},\\
			\rho^i&\le 1 - \left(1-\frac{\left(\bar{\zeta}_\theta/8 - \epsilon\right)(1-\bar{\gamma})-2b}{\widetilde{\Gamma}}\right)^{\frac{1}{N-1}}:=\rho,\\
			\zeta^i_\theta &= \frac{\bar{\zeta}_\theta}{2},
		\end{align}
	\end{subequations} 
where $\widetilde{\Gamma}$ is some absolute constant which depends only on the game parameters (formally defined in~\eqref{eq:Gammat}), $\tilde{p}:=\hat{p}^{\widetilde{L}/L}$, ${\epsilon}:=\min\left\{\frac{\bar{\zeta}_\theta}{16}-\frac{b}{1-\bar{\gamma}}, \frac{1}{2(1-\underline{\gamma})}\right\}$, 
and $\tilde{\delta}$ is such that
\begin{align*}
	\delta = 1-\left(\frac{\left(1-\tilde{\delta}\right)\tilde{p}}{\tilde{\delta}+\left(1-\tilde{\delta}\right)\tilde{p}} - \tilde{\delta}\right)\left(1-\tilde{\delta}\right).
\end{align*}
\end{theorem}
\begin{corollary}
	Recall from Proposition~\ref{prop:bounds} that $t_{{\rm mix},k}(\alpha)\le (H+1)\left((-\log \alpha)\frac{A^{NH}}{\kappa\rho^{NH}}+1\right)$. By applying this upper bound on the mixing time	to Theorem~\ref{thm:2}, we may express the $\eta^i_t$ and $T_k$ in Theorem~\ref{thm:2} as
	\begin{subequations}
		\begin{align}
			\eta^i_t& \le \frac{ \epsilon^2\tilde{\delta}\xi^i\kappa\rho^{NH}}{456N\widetilde{L}\left(1+\gamma^i+r_{\max}^i\right)^2(D^i+1)^2(H+1)\left(\kappa\rho^{NH}-{A^{NH}}\log \eta^i\right)}, \quad\forall t=t_{{k}},\ldots,t_{{k}+1}-1,\\
			T_{{k}}&\ge (H+1)\left((-\log \eta_{\min})\frac{A^{NH}}{\kappa\rho^{NH}}+1\right) + \frac{\log \frac{\epsilon^2\tilde{\delta}}{2N\widetilde{L}(2D+1)^2}}{\log\left(1-\xi_{\min}\eta_{\min}/2\right)}.
		\end{align}
	\end{subequations}
\end{corollary}

%

We note that Theorem~\ref{thm:2} provides the sample complexity for the joint deterministic baseline policy $\pi_k$ to converge to a \emph{linear approximated equilibrium} in $\widetilde{\Pi}_{\rm eq}$, as defined in Definition~\ref{def:laequi}, which may or may not be the equilibrium as defined in Definition~\ref{def:equi} due to linear approximation. However, for the special case when each agent's optimal $Q$-function $Q^i_{\bar{\pi}_k^{-i}}$ is realizable in $\Wcal^i$ for all joint behavior policies $\bar{\pi}_k\in\bar{\Pi}$, we will be able to show the convergence of $\pi_k$ to an equilibrium in $\Pi_{\rm eq}$. Formally, we have the following realizability assumption.
\begin{assumption}\label{as:realizable}
	For any joint policy $\bar{\pi}_k$ as in~\eqref{eq:pibar}, agent~$i$'s optimal $Q$-function $Q^i_{\bar{\pi}_k^{-i}}$ is realizable in $\Wcal^i$, i.e., there exists $\theta^i_{\bar{\pi}_k^{-i}}\in\Theta^i\subset \mathbb{R}^d$ such that $Q^i_{\bar{\pi}_k^{-i}}(s,a^i) = \phi^i(s,a^i)^\top \theta^i_{\bar{\pi}_k^{-i}},\forall (s,a^i)\in\Scal\times \Acal^i$.
\end{assumption}

With this additional assumption, we arrive at the following result, whose proof can be found in Appendix~\ref{sec:appb}.
\begin{theorem}\label{thm:3}
	Consider a discounted stochastic game that is weakly acyclic under strict best replies~\eqref{eq:pi-i}. Suppose that  each agent updates its policies by Algorithm~\ref{al:2}. Let Assumptions~\ref{as:alpha},~\ref{as:aperiodic},~\ref{as:lin},~\ref{as:conv}, and~\ref{as:realizable} hold. Then, for any $0<\delta<1$, one has for all $k\ge K$,
	\begin{align*}
		P\left[\pi_k\in\Pi_{\rm eq}\right] \ge 1-\delta,
	\end{align*}
	provided that for all $i\in[N]$ and $k\in[K]$,
	\begin{subequations}
		\begin{align}
			\eta^i_t&=\eta^i \le \frac{ \epsilon^2\tilde{\delta}\xi^i}{456NL\left(1+\gamma^i+r_{\max}^i\right)^2(D^i+1)^2t_{{\rm mix},k}^i(\eta^i)}, \quad\forall t=t_k,\ldots,t_{k+1}-1,\\
			T_k&\ge t_{{\rm mix},k}(\eta_{\min}) + \frac{\log \frac{\epsilon^2\tilde{\delta}}{2NL(2D+1)^2}}{\log\left(1-\xi_{\min}\eta_{\min}/2\right)},\\
			K &\ge \frac{\left[\left(1-\tilde{\delta}\right)^2\hat{p}-\tilde{\delta}^2\right]L}{\left[\tilde{\delta}+\left(1-\tilde{\delta}\right)\hat{p}\right]^2\tilde{\delta}}\\
			\rho^i &=  1 - \left(1-\frac{(\bar{\zeta}/2-{\epsilon})(1-\bar{\gamma})}{\Gamma}\right)^{\frac{1}{N-1}}\\
			\zeta^i &= \frac{\bar{\zeta}}{2}
		\end{align}
	\end{subequations}
	where $\Gamma$, $\bar{\zeta}$ and $\hat{p}$ are {absolute constants} as defined in~\eqref{eq:Gamma},~\eqref{eq:zetabar} and~\eqref{eq:phat}, respectively (which depend only on the game parameters), ${\epsilon}:={\min\left\{\frac{\bar{\zeta}}{16}, \frac{1}{2(1-\underline{\gamma})}\right\}}$, 
	and $\tilde{\delta}$ is such that
	\begin{align*}
		\delta = 1-\left(\frac{\left(1-\tilde{\delta}\right)\tilde{p}}{\tilde{\delta}+\left(1-\tilde{\delta}\right)\tilde{p}} - \tilde{\delta}\right)\left(1-\tilde{\delta}\right).
	\end{align*} 
\end{theorem}

\subsection{Numerical Experiments.}
Similar to the tabular cases, we use again the grid-world stochastic game to demonstrate the effectiveness of Algorithm~\ref{al:2}. All of the details about the agents, states, actions, reward, and parameters in this setting are the same as those in Section~\ref{sec: experiments tabular}. We construct the feature vectors using a polynomial basis~(\citet{sutton2018reinforcement}). Specifically, we use order-$3$ polynomial-basis features, where each feature can be written as
\[\phi_j^i(s, a) = s(1)^{c_{1, j}}s(2)^{c_{2, j}}a^{c_{3, j}},  \]
where $j = 1, \ldots, d$, each $c_{k, j}$ is an integer in the set $\{0, 1, 2, 3\}$. Since $|\mathcal{S}\times \mathcal{A}^i| = 21$ (for both agents) in this stochastic game, we choose $d = 18$. To obtain the set of linear approximated equilibria, for each $\pi^{-i}\in\Pi^{-i}$ ($i = 1, 2$), we first compute $\theta^i_{\pi^{-i}}$ by solving the quadratic programming problem \eqref{eq:thetapi-i}, and deduce the set of best reply policies under linear function approximation $\widetilde{\Pi}_{\pi^{-i}}^i$ from~\eqref{eq:pipi-i}. Then, for each best reply policy $\tilde{\pi}^i\in\widetilde{\Pi}^i_{\pi^{-i}}$, we check if $\tilde{\pi}^{i}$ is also a best reply policy of~$\pi^{-i}$ under linear function approximation. If so, then the joint policy $(\pi^{-i}, \tilde{\pi}^i)$ is a linear approximated equilibrium. 

As in Section~\ref{sec: experiments tabular}, we also run experiments with different parameters of $K$ and $T$, and the process is repeated ~$50$ times.
The experimental results are presented in Figure~\ref{fig: linear}. We can see a similar phenomenon as in the  tabular case. It can be seen from Figure~\ref{fig: linear}a that the minimum of the fraction $\frac{1}{K}\sum_{k= 1}^K \mathbf{I}_{\{\pi_k \in \Pi_{eq}\}}$ of the~$50$ repeated runs is greater than~$0$ when $K \geq 100$. It can also be seen from Figure~\ref{fig: linear}b that the minimum of $\frac{1}{K}\sum_{k= 1}^K \mathbf{I}_{\{\pi_k \in \Pi_{eq}\}}$ is greater than~$0$ when $T \geq 200$. This implies that the joint policy converges to a linear approximated equilibrium at the end of the algorithm for all given initial policies. Moreover, it can be seen that $\pi_k$ visits a linear approximated equilibrium more often as $K$ and $T$ increase. This is consistent with Theorem~\ref{thm:2}, as~$\pi_K$ is expected to be at a linear approximated equilibrium with higher probability as~$K$ and~$T$ increase.

\begin{figure}[ht]
	\includegraphics[width=0.85\linewidth]{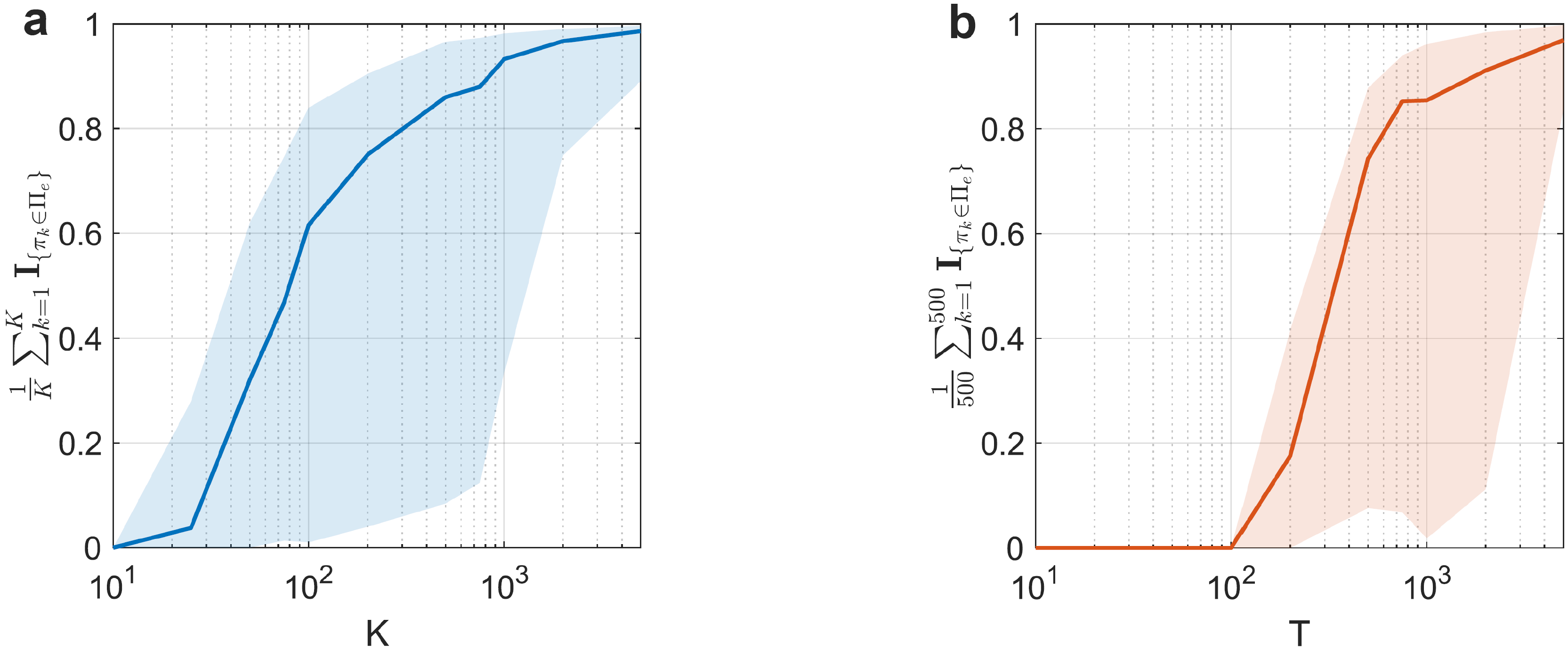}
	\centering
	\caption{Experimental results of grid-world when {\bf Algorithm~\ref{al:2}} is applied. {\bf a}, The fraction of times at which~$\pi_k$ visits an equilibrium when $K$ ranges in the interval $[10,~ 5000]$. {\bf b}, The fraction of times at which $\pi_k$ visits an equilibrium when $T$ ranges in the interval $[10,~ 5000]$. In~{\bf a}, we fix the value of~$T$ as~$1000$. In~{\bf b}, we fix the value of $K$ as~$500$. The solid lines are the average of~$50$ runs, and the shaded regions represent the min-max interval.
	}\label{fig: linear}
\end{figure}

\section{Conclusions.}\label{sec:con}


This paper was aimed at deriving sample complexity results on decentralized Q-learning algorithms for a class of general-sum stochastic games -- weakly acyclic games. The main takeaways of this paper can be summarized as follows. First, we have established finite-sample guarantees for Algorithm~\ref{al:1}, whose asymptotic convergence was shown earlier in~\citet{arslan2016decentralized}. Second, we have imposed linear function approximation to the algorithm (as Algorithm~\ref{al:2}), and provided finite-sample analysis for its convergence to a linear approximated equilibrium -- a new notion of equilibrium that we have introduced.

Regarding the first point, we note that there are some nontrivial generalizations from the asymptotic convergence to finite-sample analysis of Algorithm~\ref{al:1}. One example is in the proof of Lemma~\ref{lm:p0}, where we have to ensure that $\Pi^i_{k+1}$ well approximates $\Pi^i_{\pi^{-i}_k}$, so that any non-equilibrium policy would not stop updating before converging to equilibrium (the ``if-else'' statement in Algorithm~1). Under the linear approximation setting, the behavior that $Q^i_{\theta^i_{k+1}}$ and $Q^i_{\theta^i_{\pi_k^{-i}}}$ are separated by a distance related to the minimum Bellman error (Lemma~11) adds considerable complications to the analysis, as it may hinder the algorithm from updating a linear approximated non-equilibrium policy. We address this issue in the proof of Lemma~\ref{lm:p0t}. In addition, Lemmas~\ref{lm:p0} and~\ref{lm:p0t} provide closed-form expressions for the lower bounds $\hat{p}$ and $\tilde{p}$, which did not exist in~\citet{arslan2016decentralized} but are crucial for developing the finite-sample guarantees.
Another example of new development lies in Proposition~\ref{prop:bounds}. The notions of the minimum probability of stationary distribution~$\mu_{\min}$ and the mixing time~$t_{\rm mix}$ that appeared in the sample complexity results seem implicit, while Proposition~\ref{prop:bounds} bounds~$\mu_{\min}$ and~$t_{\rm mix}$ under the current game set up, and thus the results in both theorems can be expressed explicitly in terms of the game parameters. 
We further note that, even with these (and other) developments, there is room for improvement on the results. For instance, in the paper we picked~$\epsilon$ as the middle point of its possible range, while it remains open to optimize over~$\epsilon$ in~\eqref{eq:eps} while keeping it simple to obtain a tighter bound on~$T_k$.

Regarding the second point, we note that in the linear approximated equilibrium, each agent's policy is a best reply (to other agents' joint policy) within the linear space spanned by the set of features~(basis functions). When the dimension of the feature set is large enough so that the original Q-functions can be fully recovered for all state-action pairs, each linear approximated equilibrium is naturally also a Markov perfect equilibrium and vice versa. When we have a smaller feature set, the relationship between linear approximated equilibria and Markov perfect equilibria can be general: a joint policy can be both a linear approximated equilibrium and a Markov perfect equilibrium, or it can only be one of these equilibria but not the other one. Fixing a certain (small) number of features, it would be interesting to investigate the question of how to select features so that the set of linear approximated equilibria overlaps the most with the set of Markov perfect equilibria, which could be a promising future research direction.


\section*{Acknowledgments.}
This work was supported in part by the National Science Foundation (NSF) Award No.~1832230 and in part by the Air Force Office of Scientific Research (AFOSR) Grant No. FA9550-19-1-0353. Zuguang Gao and John R. Birge would like to acknowledge the support from the University of Chicago Booth
School of Business. The authors would like to thank Kaiqing Zhang for fruitful discussions.

\bibliographystyle{informs2014} 
\bibliography{example_paper.bib} 


%
\clearpage
%
%
\begin{APPENDICES}

\section{Proof of Theorem~\ref{thm:2}}\label{sec:appa}
We first introduce the following lemma, which is an application of the sample complexity result on the convergence of~$\theta$ for single agent Q-learning~(\citet{chen2019finite}).
\begin{lemma}\label{lem:theta}
	Fix any arbitrary $\pi_k\in\Pi$. For any  $\epsilon\ge 0$ and $0<\hat{\delta}<1$, we have  that
	$$P\left[  \big|\theta_{t_{k+1}}^i - \theta_{\bar{\pi}_k^{-i}}^{i}\big|_{\infty} \le \epsilon, \ \forall i\in[N]  \right] \geq 1-\hat{\delta},$$
	provided that the iteration number $T_k$ and the learning rates $\eta_t^i$ obey
	\begin{subequations}\label{eq:Tetatheta}
		\begin{align}
			\eta^i_t&=\eta^i \le \frac{ \epsilon^2\hat{\delta}\xi^i}{456N\left(1+\gamma^i+r_{\max}^i\right)^2(D^i+1)^2t_{{\rm mix},k}^i(\eta^i)}, \quad\forall t=t_k,\ldots,t_{k+1}-1,\forall i\in[N],\label{eq:etatitheta}\\
			T_k&\ge t_{{\rm mix},k}(\eta_{\min}) + \frac{\log \frac{\epsilon^2\hat{\delta}}{2N(2D+1)^2}}{\log\left(1-\xi_{\min}\eta_{\min}/2\right)},
		\end{align}
	\end{subequations}
	where $D = \max_iD^i$, $\eta_{\min}=\min_i\eta^i$, and $\xi_{\min} = \min_{i}\xi^i$.
\end{lemma}
\begin{proof}[Proof of Lemma~\ref{lem:theta}]\
	Note that in the $k$th exploration phase, agents adopt the joint policy $\bar{\pi}_k$ as defined in~\eqref{eq:pibar}. Under Assumptions~\ref{as:alpha},~\ref{as:aperiodic},~\ref{as:lin}, and~\ref{as:conv}, Theorem~3.1 of~\citet{chen2019finite} implies that for any agent~$i$, 
	\begin{align*}
		\E\left[\left\|\theta_{t_{k+1}}^i-\theta_{\bar{\pi}_k^{-i}}^i\right\|^2_2\right]\le \beta_1\left(1-\xi^i\eta_t^i/2\right)^{T_k-t_{{\rm mix},k}^i(\eta_t^i)} + 2\beta_2\eta_t^it_{{\rm mix},k}^i(\eta_t^i)/\xi^i,
	\end{align*}
	where $\beta_1 = (2D^i+1)^2$, $\beta_2 = 114\left(1+\gamma^i+r_{\max}^i\right)^2(D^i+1)^2$, and $\eta_t^i = \eta^i$ are such that $\eta^i\le \frac{\xi^i}{228\left(1+\gamma^i+r_{\max}^i\right)^2t_{{\rm mix},k}^i(\eta^i)}$. 
	
	For any given $\epsilon > 0$ and $0<\delta_0<1$, let $\eta^i_t$ and $T_k$ be such that 
	\begin{subequations}
		\begin{align}
			\eta^i_t&=\eta^i \le \frac{ \epsilon^2\delta_0\xi^i}{456\left(1+\gamma^i+r_{\max}^i\right)^2(D^i+1)^2t_{{\rm mix},k}^i(\eta^i)}, \quad\forall t=t_k,\ldots,t_{k+1}-1 \label{eq:etait}\\
			T_k&\ge t_{{\rm mix},k}^i(\eta^i) + \frac{\log \frac{\epsilon^2\delta_0}{2(2D^i+1)^2}}{\log\left(1-\xi^i\eta^i/2\right)},\quad\forall i\in[N].
		\end{align}
	\end{subequations}
	Then, we have that
	\begin{align*}
		\E\left[\left\|\theta_{t_{k+1}}^i-\theta_{\bar{\pi}_k^{-i}}^i\right\|^2_2\right]&\le \beta_1\left(1-\xi^i\eta_t^i/2\right)^{T_k-t_{{\rm mix},k}^i(\eta_t^i)} + 2\beta_2\eta_t^it_{{\rm mix},k}^i(\eta_t^i)/\xi^i
		\le \frac{\epsilon^2\delta_0}{2} + \frac{\epsilon^2\delta_0}{2} = \epsilon^2\delta_0,\quad \forall i\in[N].
	\end{align*}
	By Markov inequality, this further implies that
	\begin{align*}
		P\left[  \left\|\theta_{t_{k+1}}^i - \theta_{\bar{\pi}_k^{-i}}^{i}\right\|_2^2 \ge \epsilon^2 \right] \le \frac{\E\left[\left\|\theta_{t_{k+1}}^i-\theta_{\bar{\pi}_k^{-i}}^i\right\|^2_2\right]}{\epsilon^2}\le {\delta_0}, \quad \forall i\in[N],
	\end{align*}
	which is equivalent to 
	\begin{align*}
		P\left[  \left\|\theta_{t_{k+1}}^i - \theta_{\bar{\pi}_k^{-i}}^{i}\right\|_2 \ge \epsilon \right] \le \frac{\E\left[\left\|\theta_{t_{k+1}}^i-\theta_{\bar{\pi}_k^{-i}}^i\right\|^2_2\right]}{\epsilon^2}\le {\delta_0}, \quad \forall i\in[N].
	\end{align*}
	From the union bound, 
	\begin{align*}
		P\left[  \left\|\theta_{t_{k+1}}^i - \theta_{\bar{\pi}_k^{-i}}^{i}\right\|_{2} \ge \epsilon,\ \exists i\in[N]\right] \le \sum_{i\in[N]} P\left[  \left\|\theta_{t_{k+1}}^i - \theta_{\bar{\pi}_k^{-i}}^{i}\right\|_{2} \ge \epsilon\right]  \le   N\delta_0.
	\end{align*}
	Therefore,
	\begin{align*}
		P\left[\left\|\theta_{t_{k+1}}^i - \theta_{\bar{\pi}_k^{-i}}^{i}\right\|_{2}\le \epsilon,\ \forall i\in[N]\right] &\ge 1 - P\left[\left\|\theta_{t_{k+1}}^i - \theta_{\bar{\pi}_k^{-i}}^{i}\right\|_{2}\ge \epsilon,\ \exists i\in[N]\right]\ge 1-N\delta_0.
	\end{align*}
	Furthermore, since $\big|\theta_{t_{k+1}}^i - \theta_{\bar{\pi}_k^{-i}}^{i}\big|_{\infty}\le \left\|\theta_{t_{k+1}}^i - \theta_{\bar{\pi}_k^{-i}}^{i}\right\|_{2}$, we have that 
	\begin{align*}
		P\left[  \big|\theta_{t_{k+1}}^i - \theta_{\bar{\pi}_k^{-i}}^{i}\big|_{\infty} \le \epsilon, \ \forall i\in[N]  \right]\ge P\left[\left\|\theta_{t_{k+1}}^i - \theta_{\bar{\pi}_k^{-i}}^{i}\right\|_{2}\le \epsilon,\ \forall i\in[N]\right] \ge 1-N\delta_0.
	\end{align*}
	By taking $\hat{\delta} = N\delta_0$, $\eta^i_t$ as in~\eqref{eq:etait}, and 
	$$
	T_k \ge t_{{\rm mix},k}(\eta_{\min}) + \frac{\log \frac{\epsilon^2\hat{\delta}}{2N(2D+1)^2}}{\log\left(1-\xi_{\min}\eta_{\min}/2\right)} \ge t_{{\rm mix},k}^i(\eta^i) + \frac{\log \frac{\epsilon^2\delta_0}{2(2D^i+1)^2}}{\log\left(1-\xi^i\eta^i/2\right)},\quad\forall i\in[N],
	$$
	the proof is completed.
	%
	%
	%
\end{proof}

Lemma~\ref{lem:theta} bounds the approximation error of $\theta^i_{t_{k+1}}$ for each agent. By noting that $|\phi^i|_\infty \le \|\phi^i\|_2 \le 1$, Lemma~\ref{lem:theta} implies that for an arbitrary $\bar{\pi}_k$ as in~\eqref{eq:pibar} and for any $\epsilon>0$ and $0<\hat{\delta}<1$,
\begin{align}\label{eq:thetatoq}
	P\left[  \left|Q_{\theta_{t_{k+1}}^i}^i - Q^i_{\theta_{\bar{\pi}_k^{-i}}^{i}}\right|_{\infty} \le \epsilon, \ \forall i\in[N]  \right] = P\left[  \left|{\phi^i}^\top\theta_{t_{k+1}}^i - {\phi^i}^\top\theta_{\bar{\pi}_k^{-i}}^{i}\right|_{\infty} \le \epsilon, \ \forall i\in[N]  \right] \geq 1-\hat{\delta}
\end{align}
when the same conditions~\eqref{eq:Tetatheta} are satisfied.

Our next goal is to bound the approximation error of policy perturbation. 
Recall the definition of the randomized policy in~\eqref{eq:pibar}, and consider the joint policies of all agents except~$i$. With probability $\prod_{j\ne i}(1-\rho^j)$, all agents $j\ne i$ end up playing their baseline policies, which results in $\left|Q_{\theta_{\pi_k^{-i}}^i}^i - Q_{\theta_{\bar{\pi}_k^{-i}}^i}^{i}\right|_{\infty} = 0$, i.e. the approximation error of policy perturbation becomes zero in this case. When not all agents play their baseline policies, let $\varphi^{-i}\in\Delta^{-i}$ be some convex combination of the policies in $\Delta^{-i}$ of the form where each agent~$j\ne i$ either uses a baseline policy $\pi^j\in\Pi^j$ or the uniform distribution. More precisely, let $J$ denote the subset of agents choosing the baseline policies, and let 
\begin{align}\label{eq:phi-it}
	\varphi^{-i}=\sum_{J\subset \{1,\ldots,N\}\setminus \{i\}}a_J\varphi_J^{-i},
\end{align}
where $a_J:= \frac{\prod_{j\in J}(1-\rho^j)\prod_{j\notin J\cup\{i\}}\rho^j}{1-\prod_{j\ne i}(1-\rho^j)}$ and $\varphi_J\in\Delta^{-i}$ is such that $\varphi_J^j = \pi^j$ for $j\in J$ and $\varphi_J^j=\nu^j$ for $j\notin J\cup\{i\}$. Denote by $\bar{\Delta}^{-i}\subset\Delta^{-i}$ the set of all policies in the form of~\eqref{eq:phi-it}. Note that $\bar{\Delta}^{-i}$ is a finite set. We then define 
\begin{align}\label{eq:Gammat}
	\widetilde{\Gamma}:=\max_{\left(\pi^{-i},\varphi^{-i}\right)\in \Pi^{-i}\times \bar{\Delta}^{-i}}\left| \mathcal{T}_{\pi^{-i}}^i\Big(Q_{\theta_{\pi^{-i}}^i}^i\Big)-\mathcal{T}_{\varphi^{-i}}^i\Big(Q_{\theta_{\pi^{-i}}^i}^i\Big)\right|_{\infty}.
\end{align}
We next have the following lemma on the approximation error due to policy perturbation.
\begin{lemma}
	\label{lm:Qexptheta}
	Fix any arbitrary $\pi_k\in\Pi$. For any $\tilde{\epsilon}>0$, if $\rho^i$ satisfies
	\begin{align}\label{eq:rhoitheta}
		\rho^i\le 1 - \left(1-\frac{\tilde{\epsilon}(1-\bar{\gamma})}{\widetilde{\Gamma}}\right)^{\frac{1}{N-1}},\quad\forall i\in[N],
	\end{align}
	then, we have that
	$$\left|Q_{\theta_{\pi_k^{-i}}^i}^i - Q_{\theta_{\bar{\pi}_k^{-i}}^i}^{i}\right|_{\infty} = \left|{\phi^i}^\top \theta^i_{\pi_k^{-i}}-{\phi^i}^\top \theta^i_{\bar{\pi}_k^{-i}}\right|_\infty \leq \tilde{\epsilon} + \frac{2b}{1-\bar{\gamma}}, \quad \forall i\in[N], k\in[K].$$
\end{lemma}
\begin{proof}[Proof of Lemma~\ref{lm:Qexptheta}]\
	First note that, for all $i\in[N]$ and $k\in[K]$,
	\begin{align}
		\left|Q_{\theta_{\pi_k^{-i}}^i}^i -  Q_{\theta_{\bar{\pi}_k^{-i}}^i}^{i}\right|_{\infty}
		&= \left|Q_{\theta_{\pi_k^{-i}}^i}^i-\mathcal{T}_{\pi_k^{-i}}^i\Big(Q_{\theta_{\pi_k^{-i}}^i}^i\Big) + \mathcal{T}_{\pi_k^{-i}}^i\Big(Q_{\theta_{\pi_k^{-i}}^i}^i\Big) - \mathcal{T}_{\bar{\pi}_k^{-i}}^i\Big(Q_{\theta_{\bar{\pi}_k^{-i}}^i}^{i}\Big)+ \mathcal{T}_{\bar{\pi}_k^{-i}}^i\Big(Q_{\theta_{\bar{\pi}_k^{-i}}^i}^{i}\Big)-Q_{\theta_{\bar{\pi}_k^{-i}}^i}^{i}\right|_\infty \nonumber\\
		&= \left|Q_{\theta_{\pi_k^{-i}}^i}^i-\mathcal{T}_{\pi_k^{-i}}^i\Big(Q_{\theta_{\pi_k^{-i}}^i}^i\Big)\right|_\infty + \left| \mathcal{T}_{\pi_k^{-i}}^i\Big(Q_{\theta_{\pi_k^{-i}}^i}^i\Big) - \mathcal{T}_{\bar{\pi}_k^{-i}}^i\Big(Q_{\theta_{\bar{\pi}_k^{-i}}^i}^{i}\Big)\right|_\infty + \left| \mathcal{T}_{\bar{\pi}_k^{-i}}^i\Big(Q_{\theta_{\bar{\pi}_k^{-i}}^i}^{i}\Big)- Q_{\theta_{\bar{\pi}_k^{-i}}^i}^{i}\right|_\infty \nonumber\\
		&\le b + \left| \mathcal{T}_{\pi_k^{-i}}^i\Big(Q_{\theta_{\pi_k^{-i}}^i}^i\Big) - \mathcal{T}_{\bar{\pi}_k^{-i}}^i\Big(Q_{\theta_{\bar{\pi}_k^{-i}}^i}^{i}\Big)\right|_\infty + b \nonumber\\
		&\leq \left| \mathcal{T}_{\pi_k^{-i}}^i\Big(Q_{\theta_{\pi_k^{-i}}^i}^i\Big) - \mathcal{T}_{\bar{\pi}_k^{-i}}^i\Big(Q_{\theta_{{\pi}_k^{-i}}^i}^{i}\Big)\right|_\infty + \left| \mathcal{T}_{\bar{\pi}_k^{-i}}^i\Big(Q_{\theta_{\pi_k^{-i}}^i}^i\Big) - \mathcal{T}_{\bar{\pi}_k^{-i}}^i\Big(Q_{\theta_{\bar{\pi}_k^{-i}}^i}^{i}\Big)\right|_\infty+ 2b.\label{eq:TQtheta}
	\end{align}
	By definition of $\bar{\pi}_k^{-i}$, we have that $P\left[\bar{\pi}_k^{-i} = \pi_k^{-i}\right] = \prod_{j\ne i}(1-\rho^j)$. With probability $1-\prod_{j\ne i}(1-\rho^j)$, $\bar{\pi}_k^{-i}\ne \pi_k^{-i}$ and $\bar{\pi}_k^{-i}\in \bar{\Delta}^{-i}$. Thus, the first term of~\eqref{eq:TQtheta} can be bounded by
	\begin{align}\label{eq:TQ1theta}
		\left| \mathcal{T}_{\pi_k^{-i}}^i\Big(Q_{\theta_{\pi_k^{-i}}^i}^i\Big) - \mathcal{T}_{\bar{\pi}_k^{-i}}^i\Big(Q_{\theta_{{\pi}_k^{-i}}^i}^{i}\Big)\right|_\infty\le \left(1-\prod_{j\ne i}(1-\rho^j)\right)\times \left| \mathcal{T}_{\pi_k^{-i}}^i\Big(Q_{\theta_{\pi_k^{-i}}^i}^i\Big)-\mathcal{T}_{\varphi_k^{-i}}^i\Big(Q_{\theta_{\pi_k^{-i}}^i}^i\Big)\right|_{\infty},
	\end{align}
	for some $\phi_k^{-i}\in\bar{\Delta}^{-i}$.
	On the other hand, by the contraction mapping of the Bellman operator, we have that
	\begin{align}\label{eq:TQ2theta}
		\left| \mathcal{T}_{\bar{\pi}_k^{-i}}^i\Big(Q_{\theta_{\pi_k^{-i}}^i}^i\Big) - \mathcal{T}_{\bar{\pi}_k^{-i}}^i\Big(Q_{\theta_{\bar{\pi}_k^{-i}}^i}^{i}\Big)\right|_\infty\le \gamma^i\left|Q_{\theta_{\pi_k^{-i}}^i}^i - Q_{\theta_{\bar{\pi}_k^{-i}}^i}^{i}\right|_{\infty}.
	\end{align}
	Substituting~\eqref{eq:TQ1theta} and~\eqref{eq:TQ2theta} back into~\eqref{eq:TQtheta}, we have that
	\begin{align*}
		\left| Q_{\theta^i_{\pi_k^{-i}}}^i - Q_{\theta^i_{\bar{\pi}_k^{-i}}}^{i}\right|_{\infty}&\le \left(1-\prod_{j\ne i}(1-\rho^j)\right)\times \left| \mathcal{T}_{\pi_k^{-i}}^i\Big(Q_{\theta_{\pi_k^{-i}}^i}^i\Big)-\mathcal{T}_{\varphi_k^{-i}}^i\Big(Q_{\theta_{\pi_k^{-i}}^i}^i\Big)\right|_{\infty} + \gamma^i\left|Q_{\theta_{\pi_k^{-i}}^i}^i - Q_{\theta_{\bar{\pi}_k^{-i}}^i}^{i}\right|_{\infty}+2b\\
		&\le \left(1-\prod_{j\ne i}(1-\rho^j)\right) \widetilde{\Gamma} + \gamma^i\left|Q_{\theta_{\pi_k^{-i}}^i}^i - Q_{\theta_{\bar{\pi}_k^{-i}}^i}^{i}\right|_{\infty} + 2b,
	\end{align*}
	which implies that
	\begin{align*}
			\left| Q_{\theta^i_{\pi_k^{-i}}}^i - Q_{\theta^i_{\bar{\pi}_k^{-i}}}^{i}\right|_{\infty}&\le \frac{\left(1-\prod_{j\ne i}(1-\rho^j)\right) \widetilde{\Gamma}+2b}{1-\gamma^i}\le \frac{\left(1-\prod_{j\ne i}(1-\rho^j)\right) \widetilde{\Gamma}+2b}{1-\bar{\gamma}}.
	\end{align*}
	If for all~$i\in[N]$, $\rho^i\le 1 - \left(1-\frac{\tilde{\epsilon}(1-\bar{\gamma})}{\widetilde{\Gamma}}\right)^{\frac{1}{N-1}}$, then, we have that $1-\rho^j\ge \left(1-\frac{\tilde{\epsilon}(1-\bar{\gamma})}{\widetilde{\Gamma}}\right)^{\frac{1}{N-1}}$, which implies that $\prod_{j\ne i}(1-\rho^j) \ge 1-\frac{\tilde{\epsilon}(1-\bar{\gamma})}{\widetilde{\Gamma}}$, and thus
	\begin{align*}
		\left| Q_{\theta^i_{\pi_k^{-i}}}^i - Q_{\theta^i_{\bar{\pi}_k^{-i}}}^{i}\right|_{\infty}\le \frac{\left(1-\prod_{j\ne i}(1-\rho^j)\right) \widetilde{\Gamma}+2b}{1-\bar{\gamma}}\le\tilde{\epsilon} + \frac{2b}{1-\bar{\gamma}}.
	\end{align*}
	The above holds for all $i\in[N]$ and $k\in[K]$, which completes the proof.
\end{proof}

Recall from~\eqref{eq:zetabart} that $\bar{\zeta}_\theta$ is the minimum separation of agents' optimal linear approximated Q-factors.
By Assumption~\ref{as:bound}, we have that $\bar{\zeta}_\theta>\frac{8b}{1-\bar{\gamma}}$. We consider $\bar{\zeta}_\theta$ as an upper bound on $\zeta^i_\theta$ for all $i$. We next define the following random event for any arbitrary $\pi_k\in\Pi$:
\begin{align*}
	\widetilde{E}_k:=\left\{\omega\in\Omega:  \left|Q_{\theta^i_{t_{k+1}}}^i - Q_{\theta^i_{\pi_k^{-i}}}^{i}\right|_{\infty} <  \frac{1}{4}\min\{\zeta^i_\theta,\bar{\zeta}_\theta-\zeta^i_\theta\},  \forall i   \right\}.
\end{align*}
With this definition of $\widetilde{E}_k$, we show that, if $\widetilde{E}_k$ is not empty and $\pi_k\in\widetilde{\Pi}_{\rm eq}$, then $\pi_{k+1} = \pi_k$ with probability~$1$. 
\begin{lemma}\label{lm:pik1t}
	Given any $\pi_k\in\Pi$ and the corresponding $\widetilde{E}_k$, for all $k$, we have that
	$$P\left[\pi_{k+1} = \pi_k\mid \widetilde{E}_k,\ \pi_k\in\widetilde{\Pi}_{\rm eq}\right] = 1.$$
\end{lemma}
\begin{proof}[Proof of Lemma~\ref{lm:pik1t}]\
	Let $\hat{a}^{i*} := \argmax_{\hat{a}^i}Q^i_{t_{k+1}}\left(s,\hat{a}^i\right)$. Then, conditioned on $E_k$ and $\pi_k\in\widetilde{\Pi}_{\rm eq}$, we have that
	\begin{align*}
		&\qquad \max_{\hat{a}^i}Q^i_{\theta^i_{t_{k+1}}}\left(s,\hat{a}^i\right) -Q^i_{\theta^i_{t_{k+1}}}\left(s,\pi_k^i(s)\right)\\
		 &= Q^i_{\theta^i_{t_{k+1}}}\left(s,\hat{a}^{i*}\right) -Q^i_{\theta^i_{t_{k+1}}}\left(s,\pi_k^i(s)\right)\\ &=\left[Q^i_{\theta^i_{t_{k+1}}}\left(s,\hat{a}^{i*}\right) - Q^i_{\theta^i_{\pi_k^{-i}}}\left(s,\pi_k^i(s)\right)\right] + \left[Q^i_{\theta^i_{\pi_k^{-i}}}\left(s,\pi_k^i(s)\right)- Q^i_{\theta^i_{t_{k+1}}}\left(s,\pi_k^i(s)\right)\right] \\
		&<Q^i_{\theta^i_{t_{k+1}}}\left(s,\hat{a}^{i*}\right) - Q^i_{\theta^i_{\pi_k^{-i}}}\left(s,\pi_k^i(s)\right)+\frac{1}{4}\min\left\{\zeta^i_\theta,\bar{\zeta}_\theta-\zeta^i_\theta\right\}\\
		&<\left[Q^i_{\theta^i_{t_{k+1}}}\left(s,\hat{a}^{i*}\right) - Q^i_{\theta^i_{\pi_k^{-i}}}\left(s,\hat{a}^{i*}\right)\right] + \left[Q^i_{\theta^i_{\pi_k^{-i}}}\left(s,\hat{a}^{i*}\right)- Q^i_{\theta^i_{\pi_k^{-i}}}\left(s,\pi_k^i(s)\right)\right]+\frac{1}{4}\min\left\{\zeta^i_\theta,\bar{\zeta}_\theta-\zeta^i_\theta\right\}\\
		&< \frac{1}{4}\min\left\{\zeta^i_\theta,\bar{\zeta}_\theta-\zeta^i_\theta\right\} + \frac{1}{4}\min\left\{\zeta^i_\theta,\bar{\zeta}_\theta-\zeta^i_\theta\right\}\\
		&\le \frac{1}{2}\min\left\{\zeta^i_\theta,\bar{\zeta}_\theta-\zeta^i_\theta\right\},
	\end{align*}
	where the second-to-last inequality follows since $Q^i_{\theta^i_{\pi_k^{-i}}}\left(s,\hat{a}^{i*}\right)- Q^i_{\theta^i_{\pi_k^{-i}}}\left(s,\pi_k^i(s)\right)<0$, which follows from $\pi_k\in\widetilde{\Pi}_{\rm eq}$. 
	It follows that $Q^i_{\theta^i_{t_{k+1}}}\left(s,\pi_k^i(s)\right)\ge \max_{\hat{a}^i}Q^i_{\theta^i_{t_{k+1}}}\left(s,\hat{a}^i\right) - \frac{1}{2}\zeta^i_\theta$ for all~$i$. Then, by Algorithm~\ref{al:2} (lines~10-12), we have that $\pi_{k+1} = \pi_k$ with probability~$1$.
\end{proof}
Recall that $\widetilde{L}$ is the maximum length of the shortest strict best reply path from any policy to a linear approximated equilibrium policy. Our next lemma lower bounds the conditional probability of $\pi_{k+L}$ being a linear approximated equilibrium policy, given that $\pi_k$ is not a linear approximated equilibrium policy and given $\widetilde{E}_k,\ldots,\widetilde{E}_{k+\widetilde{L}-1}$.
\begin{lemma}\label{lm:p0t}
	Let 
	\begin{align}\label{eq:phatt}
		\tilde{p} := \left(\min_{j\in \{1,\ldots,N\}}\left\{ \frac{1-\lambda^j}{\left|\Pi^{j}\right|}\cdot \prod_{i\ne j}\lambda^i\right\}\right)^{\widetilde{L}} = \hat{p}^{\widetilde{L}/L}.
	\end{align}
	We have that
	\begin{align}
		P\left[  \pi_{k+\widetilde{L}} \in \widetilde{\Pi}_{\rm eq} \ \big| \ \widetilde{E}_k,\dots,\widetilde{E}_{k+\widetilde{L}-1},   \pi_k\not\in\widetilde{\Pi}_{\rm eq} \right] \geq \tilde{p}.\label{eq:ne2et}
	\end{align}
\end{lemma}
\begin{proof}[Proof of Lemma~\ref{lm:p0t}]\
	The proof is similar to that of Lemma~\ref{lm:p0}. Consider some $\pi_{k}\notin \widetilde{\Pi}_{\rm eq}$; there must exist at least one agent, say $i$, whose policy $\pi^i_k$ is not the best reply to $\pi_k^{-i}$, i.e., $\pi^i_k\notin \widetilde{\Pi}^i_{\pi_k^{-i}}$. In this case, we claim that $\pi_k^i \notin \widetilde{\Pi}^i_{k+1}$, where $\widetilde{\Pi}^i_{k+1}$ is as defined in Algorithm~\ref{al:2} (line~10). In other words, the ``else'' statement in Algorithm~\ref{al:2} (line~14) will be executed. To see this, it suffices to show that $\phi_{t_{k+1}}^i(s,\pi_k^i(s))^\top \theta^i_{t_{k+1}}< \max_{a^i\in\mathcal{A}^i}\phi_{t_{k+1}}^i(s,a^i)^\top \theta^i_{t_{k+1}}-\frac{1}{2}\zeta^i_{\theta}$, for some $s\in \mathcal{S}$. 
	Conditioned on $\widetilde{E}_k$, we have that
	$$
	Q^i_{\theta^i_{\pi_k^{-i}}}(s,a^i) - \frac{1}{4}\min\{\zeta^i_{\theta}, \bar{\zeta}_{\theta}-\zeta^i_{\theta}\} < Q^i_{\theta^i_{t_{k+1}}}(s,a^i) < Q^i_{\theta^i_{\pi_k^{-i}}}(s,a^i) + \frac{1}{4}\min\{\zeta^i_{\theta}, \bar{\zeta}_{\theta}-\zeta^i_{\theta}\},
	$$ 
	i.e., $Q^i_{\theta^i_{t_{k+1}}}(s,a^i)$ lies within a distance of $\frac{1}{4}\min\{\zeta^i_{\theta}, \bar{\zeta}_{\theta}-\zeta^i_{\theta}\}$ to $Q^i_{\theta^i_{\pi_k^{-i}}}(s,a^i)$. Moreover, we note that $\frac{1}{4}\min\{\zeta^i_{\theta}, \bar{\zeta}_{\theta}-\zeta^i_{\theta}\} \leq \frac{1}{8} \bar{\zeta}_{\theta}$. 
	Recall that $\{Q^i_{\theta^i_{\pi_k^{-i}}}(s,a^i): a^i\in \mathcal{A}^i\}$ are dispersed with spacing being at least $\bar{\zeta}_{\theta}$, where $\bar{\zeta}_{\theta}$ is as defined in~\eqref{eq:zetabart} as the minimum separation between the approximated $Q$-factors. Thus, it follows that the possible ranges of $Q^i_{\theta^i_{t_{k+1}}}(s,a^i)$ for all $a^i\in\mathcal{A}^i$ are mutually exclusive, which implies that the $\tau$-th best action under $Q^i_{\theta^i_{\pi_k^{-i}}}$ is identical to that under $Q^i_{\theta^i_{t_{k+1}}}$, i.e., $$
	\argmax_{a^i\in\mathcal{A}^i}\Big(Q^i_{\theta^i_{\pi_k^{-i}}}(s,a^i)\Big)_{(\tau)} = \argmax_{a^i\in\mathcal{A}^i}\Big(Q^i_{\theta^i_{t_{k+1}}}(s,a^i)\Big)_{(\tau)},
	$$ 
	where $(\cdot)_{(\tau)}$ represents the $\tau$-th largest value. For instance, when $\tau=1$, we have that $\argmax_{a^i\in\mathcal{A}^i}Q^i_{\theta^i_{\pi_k^{-i}}}(s,a^i)= \argmax_{a^i\in\mathcal{A}^i}Q^i_{\theta^i_{t_{k+1}}}(s,a^i)$, which are denoted by $a^{i*}_{\theta^i_{\pi_k^{-i}}}(s)$ and $a^{i*}_{\theta^i_{t_{k+1}}}(s)$, respectively. 
	
	Since $\pi^i_k\notin \widetilde{\Pi}^i_{\pi_k^{-i}}$, it follows that $\pi^i_k(s)\neq \argmax_{a^i\in\mathcal{A}^i} Q^i_{\theta^i_{\pi_k^{-i}}}(s,a^i)=:a^{i*}_{\theta^i_{\pi_k^{-1}}}(s)$ for some $s\in\mathcal{S}$. Then, we have that
	\begin{align*}
		\max_{a^i\in\mathcal{A}^i}Q_{\theta^i_{t_{k+1}}}^i(s,a^i) - Q_{\theta^i_{t_{k+1}}}^i(s,\pi^i_k(s)) &> \left(\max_{a^i\in\mathcal{A}^i}Q_{\theta^i_{\pi_k^{-i}}}^i(s,a^i) - \frac{1}{8}\bar{\zeta}_{\theta}\right) - \left(Q_{\theta^i_{\pi_k^{-i}}}^i(s,\pi^i_k(s)) + \frac{1}{8}\bar{\zeta}_{\theta}\right)\\
		&= \left(Q_{\theta^i_{\pi_k^{-i}}}^i\Big(s,a^{i*}_{\theta^i_{\pi_k^{-i}}}(s)\Big) - Q_{\theta^i_{\pi_k^{-i}}}^i(s,\pi^i_k(s))\right) -  \frac{1}{4}\bar{\zeta}_{\theta}\\
		&\geq \bar{\zeta}_{\theta} - \frac{1}{4}\bar{\zeta}_{\theta} =\frac{3}{4}\bar{\zeta}_{\theta} \geq \frac{3}{4}\zeta^i_{\theta} >  \frac{1}{2}\zeta^i_{\theta}
	\end{align*}
	as desired. Now, we are ready to prove the statement.
	
	Let $l$ be the length of the shortest strict best reply path from $\pi_k$ to a linear approximated equilibrium policy. Then, $l\le \widetilde{L}$. Let the sequence of policies along the path be $\pi_0,\pi_1,\ldots, \pi_l$, with $\pi_0 = \pi_k\notin \widetilde{\Pi}_{\rm eq}$ and $\pi_l\in\widetilde{\Pi}_{\rm eq}$. Further, let $i_1,\ldots,i_l$ be the agent that changes its policy at each update, i.e., $\pi_{n-1}$ and $\pi_{n}$ differ only at agent $i_n$, for all $n=1,\ldots,l$. Then, we  use the two probabilities in the policy update rule in Algorithm~\ref{al:2}~(line~14) to yield 
	\begin{align*}
		&\quad P\left[  \pi_{k+\widetilde{L}} \in \widetilde{\Pi}{\rm eq} \ \big| \ E_k,\dots,E_{k+\widetilde{L}-1},   \pi_k\not\in\widetilde{\Pi}_{\rm eq} \right]\ge  P\left[  \pi_{k+\widetilde{L}} =\pi_l \ \big| \ E_k,\dots,E_{k+\widetilde{L}-1},   \pi_k\not\in\widetilde{\Pi}_{\rm eq} \right]\\
		&\ge P\left[  \pi_{k+1} =\pi_1,\pi_{k+2} =\pi_2,\ldots,\pi_{k+l}=\pi_l, \pi_{k+l+1}=\cdots=\pi_{k+\widetilde{L}}=\pi_l \ \big| \ E_k,\dots,E_{k+\widetilde{L}-1},   \pi_k\not\in\widetilde{\Pi}_{\rm eq} \right]\\
		&\ge 
		P\left[  \pi_{k+1} =\pi_1 \ \big| \ E_k,\dots,E_{k+\widetilde{L}-1},   \pi_k=\pi_0 \right]\cdot P\left[  \pi_{k+2} =\pi_2 \ \big| \ E_k,\dots,E_{k+\widetilde{L}-1},   \pi_k=\pi_0, \pi_{k+1}=\pi_1 \right]\\
		&\quad \cdot P\left[  \pi_{k+3} =\pi_3 \ \big| \ E_k,\dots,E_{k+\widetilde{L}-1},   \pi_k=\pi_0, \pi_{k+1}=\pi_1, \pi_{k+2} = \pi_2 \right] \cdot \cdots \\
		&\quad \cdot P\left[  \pi_{k+l} =\pi_l \ \big| \ E_k,\dots,E_{k+\widetilde{L}-1},   \pi_k=\pi_0, \pi_{k+1}=\pi_1,\ldots,\pi_{k+l-1} = \pi_{l-1} \right]\\
		&\quad \cdot P\left[  \pi_{k+l+1} =\pi_l \ \big| \ E_k,\dots,E_{k+\widetilde{L}-1},   \pi_k=\pi_0, \pi_{k+1}=\pi_1,\ldots,\pi_{k+l} = \pi_{l} \right] \cdot \cdots\\
		&\quad \cdot P\left[  \pi_{k+\widetilde{L}} =\pi_l \ \big| \ E_k,\dots,E_{k+\widetilde{L}-1},   \pi_k=\pi_0, \pi_{k+1}=\pi_1,\ldots,\pi_{k+l} = \pi_{l},\ldots,\pi_{k+\widetilde{L}-1} = \pi_l \right]\\
		&\ge \prod_{j\in\{i_1,\ldots,i_l\}}\left(\frac{1-\lambda^j}{\left|\Pi^{j}\right|}\cdot \prod_{i\ne j}\lambda^i\right)\ge \left(\min_{j\in \{1,\ldots,N\}} \left\{\frac{1-\lambda^j}{\left|\Pi^{j}\right|}\cdot \prod_{i\ne j}\lambda^i\right\}\right)^l \ge \left(\min_{j\in \{1,\ldots,N\}}\left\{ \frac{1-\lambda^j}{\left|\Pi^{j}\right|}\cdot \prod_{i\ne j}\lambda^i\right\}\right)^{\widetilde{L}},
	\end{align*}
	where we have used the fact from Lemma~\ref{lm:pik1t}: given $\pi_l\in\widetilde{\Pi}_{\rm eq}$ and the events $E_k,\ldots,E_{k+\widetilde{L}-1}$, the conditional probability that $\pi_s\in\widetilde{\Pi}_{\rm eq}$ is~$1$ for all $s\ge l$.
\end{proof}

We will then bound $P\left[\widetilde{E}_k,\ldots,\widetilde{E}_{k+\widetilde{L}-1}\right]$. Before doing that, we first look at $P[\widetilde{E}_k]$. We would like $P[\widetilde{E}_k]$ to be as large as possible. Note that $\frac{1}{4}\min\{\zeta^i_\theta,\bar{\zeta}_\theta-\zeta^i_\theta\}\le \frac{1}{8}\bar{\zeta}_\theta$, with equality holding when $\zeta^i_\theta = \frac{1}{2}\bar{\zeta}_\theta$. We next have the following lemma. 
\begin{lemma}
	\label{lm:Pit}
	Let $\zeta^i_\theta = \frac{\bar{\zeta}_\theta}{2}$ for all $i\in[N]$. Fix an arbitrary $\pi_k\in\Pi$. For any $0<\hat{\delta}<1$,  we have that
	$$P\left[ \widetilde{E}_k  \right] \geq 1-\hat{\delta},$$
	provided that
	$	\rho^i\le 1 - \left(1-\frac{\left(\bar{\zeta}_\theta/8 - \epsilon\right)(1-\bar{\gamma})-2b}{\widetilde{\Gamma}}\right)^{\frac{1}{N-1}}$, and $T_k$ and $\eta_t^i$ satisfy~\eqref{eq:Tetatheta}, where $\epsilon$ can take any value in $0<\epsilon<\min\left\{\frac{\bar{\zeta}_\theta}{8}-\frac{2b}{1-\bar{\gamma}}, \frac{1}{1-{\gamma}_{\min}}\right\}$.
\end{lemma}
\begin{proof}[Proof of Lemma~\ref{lm:Pit}]\
	A direct implication of Lemma~\ref{lem:theta} with \eqref{eq:thetatoq} and Lemma~\ref{lm:Qexptheta} is that when $T_k$ and $\eta_t^i$ satisfy~\eqref{eq:Tetatheta}, and $\rho^i$ satisfies~\eqref{eq:rhoitheta}, then, by triangle inequality, we have that
	\begin{align}
		P\left[\left|Q^i_{\theta^i_{t_{k+1}}}-Q^i_{\theta^i_{\pi_k^{-i}}}\right|_\infty\le \epsilon+\tilde{\epsilon} + \frac{2b}{1-\bar{\gamma}},\quad \forall i\in[N]\right]\ge 1-\hat{\delta}.
	\end{align}
	The lemma then follows by taking $\tilde{\epsilon} = \frac{1}{8}\bar{\zeta}_\theta-\epsilon - \frac{2b}{1-\bar{\gamma}}$.
\end{proof}
We then have the following lemma which bounds $P\left[\widetilde{E}_k,\ldots,\widetilde{E}_{k+\widetilde{L}-1}\right]$.
\begin{lemma}\label{lm:EkLt}
	For any arbitrary sequence of joint policies $\pi_k,\ldots,\pi_{k+\widetilde{L}-1}\in\Pi$, and for any $0<\tilde{\delta}<1$, we have that
	\begin{align*}
		P\left[\widetilde{E}_k,\ldots,\widetilde{E}_{k+\widetilde{L}-1}\right]\ge 1-\tilde{\delta},
	\end{align*}
	provided that for all $i\in[N]$ and for all $\hat{k}\in\{k,\ldots,k+\widetilde{L}-1\}$,
	\begin{subequations}\label{eq:Teta2t}
		\begin{align}
			\eta^i_t&=\eta^i \le \frac{ \epsilon^2\tilde{\delta}\xi^i}{456N\widetilde{L}\left(1+\gamma^i+r_{\max}^i\right)^2(D^i+1)^2t_{{\rm mix},k}^i(\eta^i)}, \quad\forall t=t_{\hat{k}},\ldots,t_{\hat{k}+1}-1,\forall i\in[N],\label{eq:etati2t}\\
			T_{\hat{k}}&\ge t_{{\rm mix},k}(\eta_{\min}) + \frac{\log \frac{\epsilon^2\tilde{\delta}}{2N\widetilde{L}(2D+1)^2}}{\log\left(1-\xi_{\min}\eta_{\min}/2\right)},\\
			\rho^i&\le 1 - \left(1-\frac{\left(\bar{\zeta}_\theta/8 - \epsilon\right)(1-\bar{\gamma})-2b}{\widetilde{\Gamma}}\right)^{\frac{1}{N-1}},\\
			\zeta^i &= \frac{\bar{\zeta}}{2},
		\end{align}
	\end{subequations} 
	where $\epsilon$ can take any value in $0<\epsilon<\min\left\{\frac{\bar{\zeta}_\theta}{8}-\frac{2b}{1-\bar{\gamma}}, \frac{1}{1-{\gamma}_{\min}}\right\}$.
\end{lemma}
\begin{proof}[Proof of Lemma~\ref{lm:EkLt}]\
	When the conditions of Lemma~\ref{lm:Pit} are satisfied, we have that $P\left[\widetilde{E}_k^c\right] < \hat{\delta}$, where $\widetilde{E}_k^c$ is the complement of $\widetilde{E}_k$. Then,
	\begin{align*}
		P\left[\widetilde{E}_k,\ldots,\widetilde{E}_{k+\widetilde{L}-1}\right]&= 1-P\left[\left(\widetilde{E}_k,\ldots,\widetilde{E}_{k+\widetilde{L}-1}\right)^c\right]=1-P\left[\widetilde{E}_k^c\cup\cdots\cup \widetilde{E}_{k+\widetilde{L}-1}^c\right]\\
		&\ge 1-\left(P\left[\widetilde{E}_k^c\right]+P\left[\widetilde{E}_{k+1}^c\right]+\cdots +P\left[\widetilde{E}_{k+\widetilde{L}-1}^c\right]\right) = 1-\widetilde{L}\hat{\delta}.
	\end{align*}
	By taking $\tilde{\delta} = \widetilde{L}\hat{\delta}$, it follows that the conditions~\eqref{eq:Tetatheta} now become~\eqref{eq:Teta2t}, and the lemma is proved.
\end{proof}


For simplicity, we choose $\epsilon =\frac{1}{2}\min\left\{\frac{\bar{\zeta}_\theta}{8}-\frac{2b}{1-\bar{\gamma}}, \frac{1}{1-{\gamma}_{\min}}\right\}$  in Theorem~\ref{thm:2}.
Note that the result of Lemma~\ref{lm:EkLt} holds for any realization of $\pi_k\in\Pi$. Therefore, under the same conditions, we in fact have that
\begin{subequations}
	\begin{align}
		P\left[\widetilde{E}_k,\ldots,\widetilde{E}_{k+\widetilde{L}-1} \ \big|\ \pi_k\in\widetilde{\Pi}_{\rm eq}\right]\ge 1-\tilde{\delta},\label{eq:condPEint}\\
		P\left[\widetilde{E}_k,\ldots,\widetilde{E}_{k+\widetilde{L}-1} \ \big|\ \pi_k\notin\widetilde{\Pi}_{\rm eq}\right]\ge 1-\tilde{\delta}.\label{eq:condPEnott}
	\end{align}
\end{subequations}
By Lemma~\ref{lm:pik1t} and~\eqref{eq:condPEint},  under conditions~\eqref{eq:Teta2t}, we have that for all~$k$,
\begin{align}\label{eq:int}
	P\left[\pi_{k}=\pi_{k+1}=\cdots=\pi_{k+\widetilde{L}}\ \big|\ \pi_k\in\widetilde{\Pi}_{\rm eq}\right]\ge 1-\tilde{\delta}.
\end{align}
By Lemma~\ref{lm:p0t} and~\eqref{eq:condPEnott}, under conditions~\eqref{eq:Teta2t}, we have that for all~$k$,
\begin{align}\label{eq:outt}
	P\left[\pi_{k+\widetilde{L}}\in\widetilde{\Pi}_{\rm eq}\ \big|\  \pi_k\notin\widetilde{\Pi}_{\rm eq}\right]\ge \tilde{p}\left(1-\tilde{\delta}\right).
\end{align}
As a notation, let ${p}_k := P\left[\pi_k\in\widetilde{\Pi}_{\rm eq}\right]$. Then,~\eqref{eq:int} and~\eqref{eq:outt} together imply that
\begin{align}
	p_{(n+1)L} \geq p_{nL}\left(1-\tilde{\delta}\right)  + (1-p_{nL})\tilde{p}\left(1-\tilde{\delta}\right). \label{eq:geometricConvt}
\end{align}
Rearranging the above, we obtain that
\begin{align}
	p_{(n+1)L} - p_{nL} &\ge    \left(1-\tilde{\delta}\right)\tilde{p} - \tilde{\delta}p_{nL}-\left(1-\tilde{\delta}\right)\tilde{p}p_{nL} =  \left[\tilde{\delta}+\left(1-\tilde{\delta}\right)\tilde{p}\right]\left[\frac{\left(1-\tilde{\delta}\right)\tilde{p}}{\tilde{\delta}+\left(1-\tilde{\delta}\right)\tilde{p}} - p_{nL}\right]\label{eq:nL1t}\\
	&\ge -\tilde{\delta}\label{eq:nL2t}
\end{align}
Note that $p_{(n+1)L}-p_{nL}\ge 0$ as long as $p_{nL}\le \frac{\left(1-\tilde{\delta}\right)\tilde{p}}{\tilde{\delta}+\left(1-\tilde{\delta}\right)\tilde{p}}$. Further,  if $p_{nL}\le \frac{\left(1-\tilde{\delta}\right)\tilde{p}}{\tilde{\delta}+\left(1-\tilde{\delta}\right)\tilde{p}} - \tilde{\delta}$, then from~\eqref{eq:nL1t}, we have that $p_{(n+1)L} - p_{nL}\ge \left[\tilde{\delta}+\left(1-\tilde{\delta}\right)\tilde{p}\right]\tilde{\delta}$; if $p_{nL} > \frac{\left(1-\tilde{\delta}\right)\tilde{p}}{\tilde{\delta}+\left(1-\tilde{\delta}\right)\tilde{p}}$, then $p_{(n+1)L}-p_{nL}\ge -\tilde{\delta}$ from~\eqref{eq:nL2t}. Therefore, we have that
\begin{align}
	p_{nL}\ge \frac{\left(1-\tilde{\delta}\right)\tilde{p}}{\tilde{\delta}+\left(1-\tilde{\delta}\right)\tilde{p}} - \tilde{\delta}, \quad \forall n\ge \tilde{n},
\end{align}
where
\begin{align}
	\tilde{n}:= \frac{\frac{\left(1-\tilde{\delta}\right)\tilde{p}}{\tilde{\delta}+\left(1-\tilde{\delta}\right)\tilde{p}} - \tilde{\delta}}{\left[\tilde{\delta}+\left(1-\tilde{\delta}\right)\tilde{p}\right]\tilde{\delta}} = \frac{\left(1-\tilde{\delta}\right)^2\tilde{p}-\tilde{\delta}^2}{\left[\tilde{\delta}+\left(1-\tilde{\delta}\right)\tilde{p}\right]^2\tilde{\delta}}.
\end{align}
This, together with~\eqref{eq:int}, implies that for all $n\ge \tilde{n}$,
\begin{align}\label{eq:deltatildet}
	P\left[\pi_{nL}=\pi_{nL+1}=\cdots=\pi_{nL+L}\in\widetilde{\Pi}_{\rm eq}\right] \ge \left(\frac{\left(1-\tilde{\delta}\right)\tilde{p}}{\tilde{\delta}+\left(1-\tilde{\delta}\right)\tilde{p}} - \tilde{\delta}\right)\left(1-\tilde{\delta}\right) := f(\tilde{\delta}).
\end{align}
Therefore, if the number of exploration phases $k\ge K:=\tilde{n}\widetilde{L}$, then $P\left[\pi_k\in\widetilde{\Pi}_{\rm eq}\right]\ge f(\tilde{\delta})$. Note that $f(\tilde{\delta})$ is continuous, decreasing in $\tilde{\delta}$, and $f(0) = 1$, $f(\delta) < 1-\delta$ for any $0<\delta<1$. Thus, we can take $\tilde{\delta}\in (0,\delta)$ such that 
\begin{align}
	\left(\frac{\left(1-\tilde{\delta}\right)\tilde{p}}{\tilde{\delta}+\left(1-\tilde{\delta}\right)\tilde{p}} - \tilde{\delta}\right)\left(1-\tilde{\delta}\right) =1-\delta,
\end{align}
which leads to $P\left[\pi_k\in\widetilde{\Pi}_{\rm eq}\right]\ge 1-\delta$ for all $k\ge K$, and this completes the proof of Theorem~\ref{thm:2}.

\section{Proof of Theorem~\ref{thm:3}}\label{sec:appb}

Following Lemma~\ref{lem:theta} and by noting that $|\phi^i|_\infty \le \|\phi^i\|_2 \le 1$, Lemma~\ref{lem:theta} implies that for an arbitrary $\bar{\pi}_k$ as in~\eqref{eq:pibar} and for any $\epsilon>0$ and $0<\hat{\delta}<1$,
\begin{align*}\label{eq:thetatoq1}
	P\left[  \left|Q_{\theta_{t_{k+1}}^i}^i - Q^i_{\theta_{\bar{\pi}_k^{-i}}^{i}}\right|_{\infty} \le \epsilon, \ \forall i\in[N]  \right] = P\left[  \left|{\phi^i}^\top\theta_{t_{k+1}}^i - {\phi^i}^\top\theta_{\bar{\pi}_k^{-i}}^{i}\right|_{\infty} \le \epsilon, \ \forall i\in[N]  \right] \geq 1-\hat{\delta}
\end{align*}
when the conditions~\eqref{eq:Tetatheta} are satisfied.
Under Assumption~\ref{as:realizable}, we have that $Q^i_{\theta_{\bar{\pi}_k^{-i}}^{i}} = Q^i_{\bar{\pi}_k^{-i}}$, which leads to
\begin{align*}
	P\left[  \left|Q_{\theta_{t_{k+1}}^i}^i - Q^i_{{\bar{\pi}_k^{-i}}}\right|_{\infty} \le \epsilon, \ \forall i\in[N]  \right]  \geq 1-\hat{\delta}. 
\end{align*}
The above is exactly the same as the result in Lemma~\ref{lm:Qappx}, and the rest of the proof follows the proof of Theorem~\ref{thm:1}.
\end{APPENDICES}
\end{document}